\documentclass{amsart}

\usepackage{amsmath}
\usepackage{amsthm}
\usepackage{amsfonts}
\usepackage{amssymb}
\usepackage{bbm}
\usepackage{natbib}
\usepackage{subcaption}
\usepackage{graphicx}
\usepackage{tikz}
\usepackage{./tikz/style/mytikzstyles}
\usepackage{pgfgantt}
\ganttset{group/.append style={orange},
	milestone/.append style={red},
	progress label node anchor/.append style={text=red}}
\usepackage{tikz-cd}
\usepackage[text={440pt,575pt},headheight=9pt,centering]{geometry}
\usepackage{algorithm}
\usepackage{algpseudocode}
\usepackage{mathtools}
\usepackage{hyperref}
\hypersetup{
	colorlinks=true,
	linkcolor=blue,
	filecolor=magenta,
	urlcolor=cyan,
	citecolor=blue
}

\usepackage{enumerate}

\graphicspath{{images/},{tikz/}, {}}

\definecolor{tolblue}{HTML}{4477AA}

\newcommand{\x}{\mathbf{x}}
\newcommand{\y}{\mathbf{y}}

\newcommand{\X}{\mathbf{X}}
\newcommand{\Y}{\mathbf{Y}}
\newcommand{\W}{\mathbf{W}}

\newcommand{\N}{\mathbf{N}}

\newcommand{\G}{{G}}
\newcommand{\E}{{E}}
\newcommand{\RR}{\mathbb{R}}

\newcommand{\PP}{\mathbb{P}}
\newcommand{\EE}{\mathbb{E}}

\newcommand{\LL}{\mathcal{L}}
\newcommand{\Cov}{\text{Cov}}

\newcommand{\indep}{\perp \!\!\! \perp}
\newcommand{\einsfun}{\mathbf{1}} 
\newcommand{\bs}{\boldsymbol}

\newcommand{\HR}{H\"usler--Reiss}

\def\g{\boldsymbol}

\newcommand{\perpe}{\perp_e}

\newcommand{\ci}{\perp\!\!\!\perp}
\DeclareMathOperator{\diag}{diag}

\newcommand{\pa}{\mathrm{pa}}

\newcommand{\an}{\mathrm{an}}

\newcommand{\nd}{\mathrm{nd}}

\newcommand{\Estar}{\E^\star}
\newcommand{\Gstar}{G^\star}

\newcommand{\intervene}{\operatorname{do}}

\theoremstyle{plain}%
\newtheorem{thm}{Theorem}
\newtheorem{proposition}{Proposition}
\newtheorem{corollary}{Corollary}
\newtheorem{lemma}{Lemma}
\theoremstyle{definition}%
\newtheorem{example}{Example}%
\newtheorem{remark}{Remark}%
\newtheorem{definition}{Definition}
\newtheorem{assumption}{Assumption}

\allowdisplaybreaks

\usepackage[foot]{amsaddr}

\begin{document}
\title{Extremes of structural causal models}

\author[S.~Engelke]{Sebastian Engelke$^1$}
\address[]{$^1$Research Institute for Statistics and Information Science, University of Geneva, Switzerland}
\email{sebastian.engelke@unige.ch}

\author[N.~Gnecco]{Nicola Gnecco$^2$}
\address[]{$^2$Department of Mathematics, Imperial College London, United Kingdom}
\email{n.gnecco@imperial.ac.uk}

\author[F.~R\"ottger]{Frank R\"ottger$^3$}
\address[]{$^3$Department of Mathematics and
	Computer Science, Eindhoven University of Technology, The Netherlands}
\email{f.rottger@tue.nl}

\date{\today}


\begin{abstract}
  The behavior of extreme observations is well-understood for time series or spatial data, but little is known if the data generating process is a structural causal model (SCM).
  We study the behavior of extremes in this model class, both for the observational distribution and under extremal interventions. We show that under suitable regularity conditions on the structure functions, the extremal behavior is described by a multivariate Pareto distribution, which can be represented as a new SCM on an extremal graph. Importantly, the latter is a sub-graph of the graph in the original SCM, which means that causal links can disappear in the tails.  
  We further introduce a directed version of extremal graphical models and show that an extremal SCM satisfies the corresponding Markov properties. 
  Based on a new test of extremal conditional independence, we propose 
  two algorithms for learning the extremal causal structure from data. The first is an extremal version of the PC-algorithm, and the second is a pruning algorithm that removes edges from the original graph to consistently recover the extremal graph. The methods are illustrated on river data with known causal ground truth.
\end{abstract}

\maketitle

\section{Introduction}

Causal relations between the components of a random vector $\X = (X_v: v\in V)$ with index set $V = \{1,\dots, d\}$, are typically expressed through structural assignments on a directed acyclic graph (DAG) $\G = (V,\E)$   
\begin{align}\label{SCM0}
	X_v := f_v( \X_{\text{pa}_G(v)}, \varepsilon_v), \quad v \in V.
\end{align}
Here, $\text{pa}_G(v)$ are the parents of $X_v$ in $\G$, the structure functions $f_v: \mathbb R^{|\pa_G(v)|}\times \mathbb R \to \mathbb R$ specify the causal mechanism, and $(\varepsilon_v: v \in V) \sim \mathbb P_\varepsilon$ are independent noise variables.
We refer to the tuple $\mathcal M = (G, \{f_v : v \in V\}, \mathbb P_\varepsilon)$ as a structural causal model (SCM) on $\G$ over the random vector $\X$, which not only specifies the observational distribution of $\X$, but also describes the distribution under interventions \citep[][]{PJS2017}.
There are numerous methods for causal discovery 
 and causal effect estimation if the intervention is in range of the training data set \citep{shimizu2006linear,kalisch2007estimating,peters2014causal}.

The behavior of extreme observations has been extensively studied for various stochastic systems including times series \citep{leadbetter2012extremes}, random fields \citep{KSdH2009} and undirected graphs \citep{segers_2020, EH2020}. 
For directed graphs, only special cases such as linear or max-linear models have been considered \citep{gissibl2018max,GMPE2019}.
As a first main result, we determine the limits of general SCMs~\eqref{SCM0} under the observational distribution of $\X$ if one of the variables become extremely large. 
We state sufficient regularity conditions on the structure functions~$f_v$ at infinity that ensure that $\X$ is multivariate regularly varying \citep{res2008}, a well-known property that guarantees the existence of a limiting vector $\Y= (Y_v: v\in V)$ describing the tail dependence.      
This limit is called a multivariate Pareto distribution and can be characterized as a new SCM on the extremal DAG $G_e=(V,E_e)$
\begin{align}\label{eSCM_intro}
	\mathcal M_e = (G_e, \{\Psi_v : v \in V\}, \mathbb P_\varepsilon),
\end{align}
where we call $\Psi_v: \mathbb R^{|\pa_{G_E}(v)|} \times \mathbb R \to \mathbb R$ the extremal structure functions, which 
satisfy the homogeneity $\Psi_v( \y + t \einsfun, e) = \Psi_v( \y, e) + t$ for all $\y \in \mathbb R^{|\pa(v)|}$ and $e,t\in\mathbb R$. Interestingly, the extremal graph $G_e$ can be shown to be a subgraph of $G$ in the sense that $E_e \subseteq E$. That means that causal links that are present in the overall system can disappear in extreme observations.
Conversely, any set of additive extremal structure functions of the form $\Psi_v(\y, e) = \Phi_v(\y)  + e$ with homogeneous functions $\Phi_v: \mathbb R^{|\pa_{G_e}(v)|} \to \mathbb R$ yields a valid extremal SCM. It turns out that many of the standard extreme value models in the literature have such an additive structure. For instance, if all $\Phi_v$ are linear and the noise terms $\varepsilon_v$ follow normal distributions, then the resulting multivariate Pareto distribution is a directed \HR{} model.

As a second step, we consider the limit of the SCM in~\eqref{SCM0}
under extremal interventions in the tail of one or several components of $\X$. In fact, under the same assumptions on the structure functions as for the observational distribution, we can establish interventional limits. This convergence can be seen as an interventional version of classical multivariate regular variation.
In particular, we show how extremal interventions propagate through the graph, and how the limit can be seen as the extremal SCM~\eqref{eSCM_intro} under suitable interventions.
We further introduce a notion of an extremal cause between two variables $X_i$ and $X_j$, and show that it can be read off from the extremal DAG $G_e$.
For instance, Figure~\ref{fig:intro} shows the scatter of two variables $X_3$ and $X_4$ simulated from different SCMs on the same graph $G$. While in both cases there is a causal effect in the original graph, this link disappears in the corresponding extremal graph $G_e$ for the data on the right-hand side, that is, an extremal intervention on $X_3$ does not lead to an extremal outcome in $X_4$. 
The notion of an extremal cause therefore describes the effect of interventions beyond the training data.

The limiting multivariate Pareto distribution $\Y$ describes the extremal causal structure of $\X$ and, in fact, it satisfies conditional independence statements in the extremal sense as introduced in \cite{EH2020}.
We define directed extremal graphical models as a multivariate Pareto distribution that satisfies the corresponding notions of local and global Markov properties. 
As an equivalent characterization, we show that the density (if it exists) of an extremal directed graphical model factorizes along the directed graph.
With this definition, the extremal limit $\Y$ of the SCM~\eqref{SCM0}
is indeed an extremal graphical model on the graph $G_e$. 
This answers a question raised by Steffen Lauritzen in the discussion contribution of~\cite{EH2020}. 

The mathematical foundations for directed extremal graphical models and extremal SCMs laid out in this paper enable a broad range of statistical methods for the analysis of causal relationships in the distributional tail. As a first tool, we propose an extremal version of the well-known PC-algorithm \citep{spirtes2000causation}, relying on a new test for extremal conditional independence in the \HR{} model class.
A second algorithm, called extremal pruning, leverages an estimate of the DAG $G$ of the initial SCM~\eqref{SCM0}. Based on the fact that causal links can only disappear in the tails, the algorithm prunes suitable edges and guarantees consistent recovery of the extremal graph $G_e$ from data.

After a brief discussion of related work at the interface of causality and extremes in Section~\ref{sec:related}, in Section~\ref{sec:preliminaries} we provide the required preliminaries for this paper. Section~\ref{sec:eSCM} contains the main limit results on 
extremes of structural causal models, both for the observational distribution and under extremal interventions. 
We introduce the notion of directed extremal graphical models in Section~\ref{sec:DEGM} and show the link to structural causal models from the previous section. 
Structure learning methods for the extremal graph are presented in Section~\ref{sec:structure_learning} and illustrated on river data from the Danube basin in Section~\ref{sec:application}. 
All proofs have been relegated to the Supplementary Material, which also contains details on directed graphical models and additional results on the \HR{} distribution.
The code to reproduce all results of this paper can be found in the GitHub repository \url{https://github.com/nicolagnecco/extremeSCM}.

\begin{figure}[tb]
	\centering
	\begin{minipage}{.45\textwidth}
	  \centering
	  \includegraphics[width=.75\linewidth]{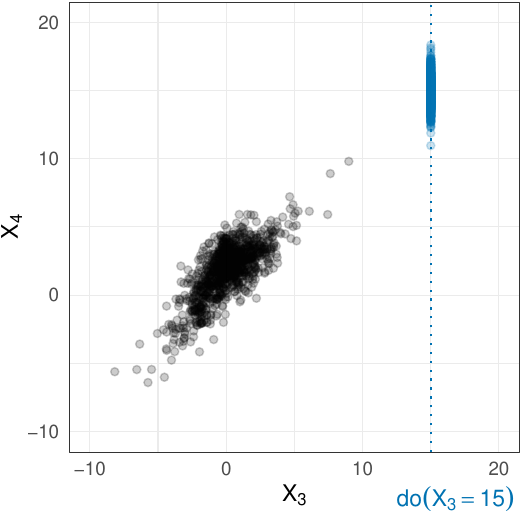}
	\end{minipage}
	\begin{minipage}{.45\textwidth}
	  \centering
	  \includegraphics[width=.75\linewidth]{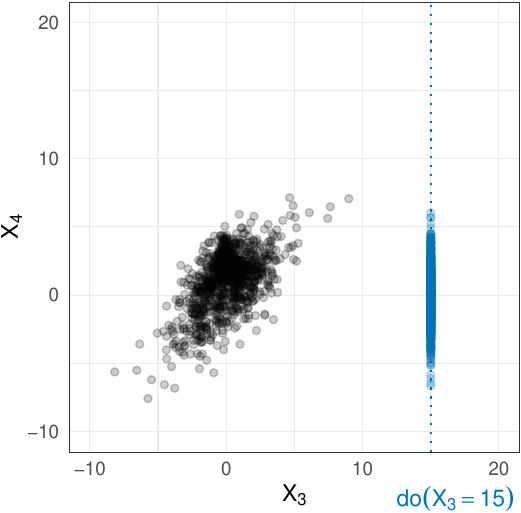}
	\end{minipage}
	\caption{Scatter plots of $X_3$ and $X_4$ simulated according to the SCMs in Example~\ref{ex:tail} (left) and Example~\ref{ex_different_eSCM} (right) on graph $G$ in Figure~\ref{subfig:b}. Blue lines represent an extremal intervention on $X_3$ and samples from the interventional distribution.
  Left: the data exhibits an extremal causal effect from $X_3$ to $X_4$ and the limiting extremal graph satisfies $G_e = G$. Right: the corresponding extremal graph $G_e$ is the DAG in Figure~\ref{subfig:c}, which is different from $G$ since $X_3$ is not an extremal cause of $X_4$; extremal interventions on $X_3$ do not lead to an extreme outcome in $X_4$.}
  \label{fig:intro}
\end{figure}

\subsection{Related work}\label{sec:related}

The study of graphical models and causal inference in the context of extreme observations has attracted increasing attention in recent years. 
A particular focus has been on sparse extremal dependence modeling on undirected graphs. \cite{segers_2020} and \cite{asenova2023extremes} study extremes of probabilistic graphical models on trees and block graphs, respectively, and \cite{EH2020} introduce the concept of general, undirected extremal graphical models. There has been ample of work on the estimation \citep{asenova2021inference, roettger2023parametric, lederer2023extremes} and structure learning \citep{eng_vol_2022, hu2022, engelke2022a, wan2023graphical,engelke2024extremalgraphicalmodelinglatent} of these models, but the directed case has only been considered for very particular model classes so far \citep{gissibl2018max,GMPE2019}; for a recent review see \cite{engelke2024graphicalmodelsmultivariateextremes}.

Most of the work on causality related to extremes considers treatment effects in the tails of the distributions. Closely related to extreme quantile estimation \citep{chernozhukov2005, ExGAM2, Velthoenetal2019}, this line of work studies the difference in the extremal quantiles of the response variable under different potential outcomes \citep{zhang2018extremal,deu2021}. The treatment is binary in these cases, which fundamentally differs from our work where we consider continuous variables that can be extreme simultaneously. 
Interestingly, \cite{wang2024extremebasedcausaleffectlearning}
observe that for light-tailed responses, confounding may become negligible in the extremes, and that average treatment effects can be estimated from the tails only. While in a different setting, this parallels the phenomenon in our theory that causal links can disappear in the extremal DAG $G_e$. 
Extremal causal discovery in multivariate systems with continuous data has mostly been studied using bivariate summary statistics \citep{mha2020, mhalla2024causaldiscoverymultivariateextremes, dei2023}, but a theoretical justification in terms of interventions is often difficult.
For reviews on causality and extremes see  \cite{engelke2021a} and \cite{chavezdemoulin2024causalityextremes}.

Our results on extremal interventions are related to the work on extrapolation that studies estimation of regression functions beyond the training data range \citep{shen2024engressionextrapolationlensdistributional,buritica024progressionextrapolationprincipleregression, bodik2024extremetreatmenteffectextrapolating}.
The methods of our paper, on the other hand, allow us to learn structural causal models and describe extrapolation under both observational and interventional distributions.

\section{Preliminaries}\label{sec:preliminaries}

\subsection{Directed graphical models and structural causal models}

In this section we provide some background on directed graphical models and their connection to structural causal models. 
A directed graph is a pair $G = (V,E)$ consisting of a node (or vertex) set $V = \{1,\dots, d\}$ and an edge set $E \subset V\times V$ containing ordered pairs of distinct nodes.
An undirected graph is a graph where there is an edge $(i,j) \in E$ if and only if $(j,i) \in E$; the two edges are then graphically represented by a single undirected connection.
In this paper, we always consider the sub-class of directed acyclic graphs (DAGs) defined as a directed graph that does not contain any directed cycles.
For a DAG $G$, the set of parents $\pa_G(i)$ of node $i\in V$ is defined as all nodes $j\in V$ such that $(j,i)\in E$. The set of non-descendants $\nd_G(i)$ of node $i$ contains all nodes $j$ for
which there is no directed path from $i$ to $j$ in $G$, and the set of ancestors $\an_G(i)$ is the collection of all nodes $j$ for which there is a directed path from $j$ to $i$ in $G$; we sometimes drop the subscript when there is no ambiguity about the corresponding graph. 
For details and more graph notions we refer to \citet[Chapter 2]{lau1996}. 

Let $G = (V,E)$ be a DAG and $\X = (X_v : v\in V)$ a random vector with components indexed by the nodes~$V$. A probabilistic graphical model connects the graph structure $G$ to conditional independence statements on the distribution of $\X$ through Markov properties. 
The directed local Markov property, for instance, requires that
\[ X_i \indep \X_{\nd(i)\setminus \pa(i)} \mid \X_{\pa(i)},\quad  \text{for all } i \in V.\] 
Another notion is the directed global Markov property, which in general is stronger than the local Markov property, but both can be shown to be equivalent for conditional independence models that form a semi-graphoid \citep[Proposition 4]{Lauritzen1990}; the latter is satisfied for probabilistic conditional independence. We refer to Supplementary Material~\ref{app:d-separation} for details on the directed global Markov property and the related notion of $d$-separation.

Directed graphical models are closely connected to SCMs in~\eqref{SCM0} on the underlying DAG $G = (V,E)$. We give here a formal definition including the behavior of the system under interventions.

\begin{definition}\label{def:SCM_intervention}
	Let $G = (V,E)$ be a DAG. For measurable functions $f_j:\mathbb R^{|\pa(j)|} \times \mathbb R \to \mathbb R$, $j\in V$, a structural causal model (SCM) for the random vector $\X = (X_j: j\in V)$ is defined as $\mathcal M := (G, \{f_j: j\in V\}, \mathbb P_{\varepsilon})$, where $\varepsilon = (\varepsilon_j : j\in V)$ denote the exogenous variables with product measure $\mathbb P_{\varepsilon}$. The observational distribution of $\X$ is then given by the generative mechanism in~\eqref{SCM0}. The SCM under perfect interventions on variables $X_i:= \xi_i \in \mathbb R$ for a subset $\mathcal I\subset V$ is denoted by $\mathcal M^{\intervene(\X_\mathcal{I}:=\boldsymbol{\xi}_\mathcal{I})} := (G, [\tilde f_j: j\in V], \mathbb P_{\varepsilon})$, where
	\begin{align}\label{f_intervened}
		\tilde f_j(\x_{\pa(j)}, e) = 
	\begin{cases}
		\xi_j, & j \in \mathcal I,\\
		f(\x_{\pa(j)}, e), & j \in V \setminus \mathcal I.\\
	\end{cases}
	\end{align}
\end{definition}

The SCMs defined above are called causal since they not only describe the observational distribution of $\X = (X_v : v\in V)$, but also how the system behaves under intervention on one or several of the variables; see \cite{bon2021} for a more complete treatment of definitions and properties.
If the model only describes the observational distribution, it is also known as a structural equation model.

In general, SCMs can have arbitrary structure functions $f_v$ and noise distributions of $\varepsilon_v$. For identifiability and statistical inference, often a subset of functions is assumed. An important example are linear SCMs.
\begin{example}\label{ex:linear_SCM}
    Let $\G = (V, \E)$ be a DAG with and let $\X \in \mathbb{R}^{d}$ denote a random vector with components 
    \begin{align}\label{eq:linSCM}
		X_v := \sum_{i\in\pa(i)} b_{iv} X_{i} + \varepsilon_v, \quad v\in V,
    \end{align}
	where $\varepsilon_v$ are independent noise variables, and the matrix $B = (b_{ij})_{i,j=1,\dots, d}$ contains the causal coefficients. Then $\X$ is called a linear SCM. It has the more compact representation $\X:=B^{\top}\X+\varepsilon$, where  $\varepsilon = (\varepsilon_v : v\in V)$.
\end{example}
Linear SCMs have been studied extensively in the literature. For instance, it is known that the causal structure can be identified from the observational distribution if not all of the noise variables are normally distributed \citep{shimizu2006linear,shimizu2011directlingam}.
In the setting where $\varepsilon_v$ have heavy-tailed distributions, the extremes of the linear SCM~\eqref{eq:linSCM} are well understood and the causal order can be learned based on the largest observations only \citep{GMPE2019}. Another popular SCM class arises when replacing the sum with the maximum.
\begin{example}\label{ex:max_linear_SCM}
    Let $\G = (V, \E)$ be a DAG with and let $\X \in \mathbb{R}^{d}$ denote a random vector with components 
    \begin{align}\label{eq:maxSCM}
		X_v := \max_{i\in\pa(i)} b_{iv} X_{i} + \varepsilon_v, \quad v\in V.
\end{align}
where $\varepsilon_v$ are independent noise variables.
Then $\X$ is called a max-linear SCM.
\end{example}
Max-linear models have been studied in the field of extreme value theory \citep{gissibl2018max} and their probabilistic structure is well understood \citep{amendola2022conditional}.
For a more detailed review of SCMs and their causal interpretation we refer to \citet{PJS2017}.

\subsection{Multivariate Pareto distributions}\label{sec:MPD}

Let $\X = (X_v: v\in V)$ be a random vector with index set $V = \{1,\dots, d\}$ and assume that the marginal distribution functions $F_v$, $v\in V$, are continuous with generalized inverse $F_v^{-1}$.
For notational simplicity we normalize the marginal distributions of $\g X$ to standard exponential distributions by defining the random vector $\X^*$ with components 
\begin{align}\label{X_star}
	X_v^* = -\log\{1- F_v(X_v)\}, \quad v\in V.
\end{align}
In practice, this marginal normalization is typically done using empirical distribution functions.
In order to describe the extremal dependence between large values of the components of this vector, a standard assumption is that the rescaled threshold exceedances converge in distribution to a (generalized) multivariate Pareto distribution $ \Y $, that is,
\begin{align}\label{eq:MPD}
	\PP(\Y\le \y)=\lim_{u\to\infty}\PP(\X^*-u\mathbf{1}\le \y \mid \max_{i=1,\dots, d} X^*_i >u), \quad \y\in\mathcal{L}, 
\end{align}
where $\mathcal{L}=\{\x\in\RR^d: \max_{i=1,\dots, d} x_i >0\} $ and $ \mathbf{1}=(1,\ldots,1)$ is the vector of ones.
Multivariate Pareto distributions are the only possible limits for such threshold exceedances \citep{rootzen2006} and therefore natural objects for the study of extremal dependence.

Equivalent to the existence of the limit in~\eqref{eq:MPD} is the well-known assumption of multivariate regular variation \citep[Chapter 5]{res2008}. On exponential margins, it states vague convergence, i.e.,
\begin{align}\label{MRV}
	\lim_{t \to \infty} t \mathbb P(\X^* - t\einsfun \in A)  = \Lambda\left(A \right),
\end{align}
for all Borel subsets $A \subset \mathbb R^d$ bounded away from $\{-\infty\}^d$ with zero mass on boundary, $\Lambda(\partial A) = 0$. Here, $\Lambda$ is a Radon measure on $\mathbb R^d$ called the exponent measure. It is homogeneous in the sense that $\Lambda(tA) = e^{-t}\Lambda(A)$ for any Borel set $A\subset \mathbb R^d$ and $t \in \mathbb R$. The exponent measure fully describes the distribution of $\Y$ since $\PP(\Y\le \y) = \Lambda\left([-\mathbf \infty, \y]^C\cap \mathcal L  \right)/ \Lambda(\mathcal L)$.
If $\Lambda$ is absolutely continuous with respect to the Lebesgue measure, its derivative $\lambda$ satisfies the same homogeneity as $\Lambda$. The density of $\Y$ is then given by $f_{\Y}(\y) = \lambda(\y) / \Lambda(\mathcal L)$ for $\y\in\mathcal L$.
For any $I\subseteq V$, let $\lambda(\y_I) = \lambda(\y_I)$ be the marginal exponent measure density obtained by integrating out all variables in $V\setminus I$.
Note that here and throughout the paper we exclude $-\infty$ as possible values of $\Y$ since we are interested in multivariate Pareto distributions that admit densities on $\mathcal L$. 

We give three popular examples of parametric families of exponent measure densities, or equivalently, of multivariate Pareto distributions. They are used throughout the paper to illustrate the theory and to provide examples for extremal causal models.
\begin{example}\label{ex:logistic_prelim}
	The $d$-dimensional extremal logistic model with parameter $\theta \in (0,1)$ has density 
	\begin{align*}
		\lambda(\y) =  \left(\sum_{i=1}^d \exp\left\{ -\frac{y_i}{\theta} \right\}\right)^{\theta-d} \exp\left\{-\frac{1}{\theta}\sum_{i=1}^{d} y_i\right\}\prod_{i=1}^{d-1}\left(\frac{i}{\theta}-1\right), \quad  y \in \mathbb R^d.
	\end{align*}
  For any $I\subseteq V$, the marginal exponent measure density $\lambda(\y_I)$ is an extremal logistic model with the same parameter $\theta$.
\end{example}

\begin{example}\label{ex:dirichlet_prelim}
	The exponent measure density of a $d$-dimensional extremal Dirichlet distribution \citep{CT1991} 
	 is defined for parameters $\alpha_1,\ldots,\alpha_k>0$ as
\begin{equation} 
	\lambda(\y) =  \frac 1 d \frac{\Gamma(1+\sum_{i=1}^d \alpha_i) \exp\left\{\sum_{i=1}^{d} y_i\right\}}{(\sum_{i=1}^d \alpha_i \exp\{y_i\})^{d+1}} \prod_{i=1}^d \frac{\alpha_i}{\Gamma(\alpha_i)}  
			  \left(\frac{\alpha_i \exp\{y_i\}}{\sum_{j=1}^d \alpha_j \exp\{y_j\}}\right)^{\alpha_i-1}, \quad y \in \mathbb R^d.
   \end{equation}
  For any $I\subseteq V$, the marginal exponent measure density $\lambda(\y_I)$ is an extremal Dirichlet model with parameters $(\alpha_i:i\in I)$ \citep[Lemma 2.4]{corradini2024stochastic}.
\end{example}

\begin{example}\label{ex:HR_cond_prelim}
	Let $\mathbb{S}^d_0$ be the set of real symmetric $d\times d$ matrices with zero diagonal and let $\mathcal{D}^d:=\{\Gamma\in\mathbb{S}^d_0: \x^{\top}\Gamma \x<0 \text{ for all } \x \in \RR^d\setminus \mathbf{0} \text{ with } \x^{\top}\mathbf{1}=0 \}$.
	The exponent measure density of a $d$-dimensional \HR{} distribution
	\citep{HR1989} with parameter matrix $\Gamma\in \mathcal{D}^{d} $ has the representation
	\[\lambda(\y)=c_\Gamma \exp\left(-\frac{1}{2}\begin{pmatrix}
		\y^{\top},1\\
	\end{pmatrix}
	\begin{pmatrix}
		-\frac{1}{2}\Gamma & \mathbf{1}\\
		\mathbf{1}^\top&0\\
	\end{pmatrix}^{-1}\begin{pmatrix}
	\y\\
	1\\
	\end{pmatrix}
	 \right), \quad y \in \mathbb R^d,
	 \]
	 where $c_\Gamma>0$ is a constant depending on $\Gamma$. For $I\subseteq V$, the marginal exponent measure density $\lambda(\y_I)$ is \HR{} with parameter matrix~$\Gamma_{II}\in \mathcal{D}^{|I|}$. Note that the \HR{} distribution can also be parameterized through a signed Laplacian matrix $\theta(I)$, called the \HR{} precision matrix \citep{HES2022}, which we can obtain via the Fiedler--Bapat identity
   \begin{align}
     \begin{pmatrix}
       -\frac{1}{2}\Gamma_{II} & \bs 1 \\
       \bs 1^{\top}          & 0     \\
     \end{pmatrix}^{-1}=\begin{pmatrix}
       \theta(I) & \mathbf{p}(I)\\
       \mathbf{p}(I)^{\top}&\sigma^2(I)\\
     \end{pmatrix}, \label{eq:Fiedler-Bapat_margin}
   \end{align}
    where $\sigma^2(I)=\frac{1}{2}\left(\mathbf{1}^{\top}\Gamma_{II}^{-1}\mathbf{1}\right)^{-1}$ and $\mathbf{p}(I)=2\sigma^2(I) \Gamma_{II}^{-1}\mathbf{1} $ are the resistance radius and curvature, respectively \citep{devriendt2022a}. For the full model $I=V$, we denote the precision matrix by $\Theta = \theta(V)$. 
\end{example}

\subsection{Extremal conditional independence}\label{sec:extremal_CI}

The support of a multivariate Pareto distribution as defined in~\eqref{eq:MPD} is contained in the set $\mathcal L$, which is not of product form. This complicates statistical analysis and inhibits the use of traditional notions of conditional independence and graphical models. In extreme value theory, it is therefore often convenient to consider exceedances where the conditioning in~\eqref{eq:MPD} is replaced by one variable $X_k^* > u$ for some $k\in V$ \citep{heffernan2004,eng2014}. The resulting limit $\Y^k$ can be seen as an auxiliary vector derived from the initial multivariate Pareto distribution as $\Y \mid \Y_k > 0$. It has stochastic representation on the product space  $\LL^k:=\{\y\in\LL:y_k>0\}$ given by
\begin{align}
	\Y^k\stackrel{d}{=}R\mathbf{1}+\W^k,\label{eq:stoch_rep}
\end{align}
where $R$ is a standard exponential random variable and $\W^k$ is a random vector called the $k$th extremal function with $W^k_k=0$ almost surely and $\mathbb E \exp W^k_j = 1$, $j\in V$. Its probability density, if it exists, is given by $\lambda(\y)$ for $\y \in \LL^k$.
An important property of these auxiliary vectors is that through homogeneity, for any $k\in V$, the distribution of $\Y^k$ fully characterizes the distribution of $\Y$ \citep[e.g.,][Corollary 2]{segers_2020}.

One advantage of the vectors $\Y^k$, $k\in V$ is that their support has product form. \citet{EH2020} leverage this to define extremal conditional independence for the multivariate Pareto distribution $\Y$.
\begin{definition} \label{def:extremal_indep}
	Let $ A,B,C $ be disjoint subsets of $ V $. We say that the components $A$ and $B$ of the multivariate Pareto distribution $\Y$ are conditionally independent of the components in $C$, if 
	\begin{align}\label{CI_k}
		\forall k\in V:\;\Y_A^{k}\indep \Y_B^{k}\mid \Y_C^{k}.
	\end{align}
	In this case, we speak of extremal conditional independence and write $ \Y_A \perp_{e} \Y_B\mid \Y_{C} $. 
\end{definition}
If the exponent measure has Lebesgue density $\lambda$, then it suffices to require the existence of a $k\in V$ in~\eqref{CI_k} \citep[Lemma A.7]{eng_iva_kir}. Extremal conditional independence is then equivalent to the factorization 
\[\lambda(\y_{A\cup B\cup C}) = \lambda(\y_{A\cup C}) \lambda(\y_{B\cup C}) / \lambda(\y_C), \quad \y \in \mathbb R^d. \]
\cite{EH2020} define undirected extremal graphical models through local and global Markov properties and show that the resulting sparse models can be used for efficient statistical dependence modeling in high dimensions.
For the \HR{} model introduced in Example~\ref{ex:HR_cond_prelim}, similarly to the Gaussian case, extremal conditional independence can be be characterized through its precision matrix.

\begin{example}\label{ex:CI_test}
	Let $\Y$ be a \HR{} multivariate Pareto distribution with parameter matrix $\Gamma$. Then the extremal function $\g W^k_{\setminus k}$ is $(d-1)$-dimensional multivariate normal with mean $-\Gamma_{{\setminus k},k}/2$ and precision matrix $\Theta^{(k)} = \Theta_{\setminus k,\setminus k}$, that is, the \HR{} precision matrix with $k$th row and column removed.
	For two nodes $i,j\in V$, extremal conditional independence statement $ Y_i\perp_{e}Y_j\mid\Y_{\setminus ij} $ is equivalent to $ \Theta_{ij}=0 $ \citep{HES2022}. More generally, for some non-empty subset $S\subset V$ with $i,j\notin S$, by~\eqref{eq:Fiedler-Bapat_margin} we have 
	\begin{align}
			Y_i\perp_{e}Y_j\mid \Y_S\quad\Longleftrightarrow \quad \theta(\{i,j\}\cup S)_{ij}=0 
			 \label{eq:CIviaGamma}.
	\end{align}
\end{example}

Extremal conditional independence can also be shown to be natural from an axiomatic point of view, in the sense that it forms a semi-graphoid under weak assumptions \citep{REZ2021}.
\cite{eng_iva_kir} generalize the conditional independence notion in Definition~\ref{def:extremal_indep} to infinite measures $\Lambda$, allowing for mass on sub-faces of $\mathbb R^d$, asymptotic independence, and applications to L\'evy processes \citep{engelke2024levygraphicalmodels}.

\section{Extremes of structural causal models}\label{sec:eSCM}

\subsection{Extremal structural causal models}

Suppose that the random vector $\X = (X_v: v\in V)$ follows an SCM $\mathcal M := (G, \{f_j: j\in V\}, \mathbb P_{\varepsilon})$ on the DAG $G$ as in~\eqref{SCM0}, and assume that the marginal distribution functions $F_v$, $v\in V$, are continuous with generalized inverse $F_v^{-1}$.
We aim to study how the causal structure of $\X$ behaves for extreme observations, typically defined as an event where one or several components of $\X$ take very large values. 
We first focus on the observational distribution and discuss implications for interventional distributions in Section~\ref{sec:interventions_causality}. 
Extremes of SCMs have so far only been studied for heavy-tailed linear \citep{GMPE2019} and max-linear models \citep{gissibl2018max} as in Examples~\ref{ex:linear_SCM} and~\ref{ex:max_linear_SCM}, respectively. In both cases, the extremes are described by a multivariate Pareto distribution $\Y$ where the exponent measure $\Lambda$ does not possess a density. This is fairly restrictive and excludes most commonly used parametric models.

We take a broader approach for SCMs with general structure functions in~\eqref{SCM0} and concentrate on models where the multivariate Pareto limit distribution admits a density. 
This implies that $G$ can only have one root node, and we assume without losing generality that this is $X_1$; see Figures~\ref{subfig:b} and~\ref{subfig:c} for examples. Indeed, if there were at least two root nodes $X_1 := f_1(\varepsilon_1)$ and $X_2 := f_2(\varepsilon_2)$ as in Figure~\ref{subfig:a}, then, because of the independence $X_1 \indep X_2$, the extremes of these two variables would be asymptotically independent and the limiting multivariate Pareto distribution would not admit a density \citep[e.g.,]{EH2020}.

To characterize the extremal dependence and causal structures of $\X$, we consider the normalization in~\eqref{X_star} by $X^*_v = H^{-1}\circ F_v(X_v)$, where $H$ is the distribution function of a standard exponential variable.
The normalized vector $\X^*$ satisfies the structural assignments
\begin{align}\label{SCM1}
	X^*_v := f^*_v( \X^*_{\pa_G(v)}, \varepsilon_v), \quad v \in V,
\end{align}
on the same DAG $G$, with structure functions
\begin{align}\label{f_star}
	f^*_v(\x, e) = H^{-1}\circ F_v \circ f_v( F_{\pa_G(v)}^{-1} \circ H(\x), e), \quad (\x,e) \in \mathbb R^{|\pa_G(v)|+1},	
\end{align}
where $F_{\pa_G(v)}^{-1} \circ H$ denotes the componentwise transformation from exponential margins to original margins $F_j$, $j\in \pa_G(v)$.
The following assumption gives an easy-to-check regularity condition on the asymptotic behavior of the structure functions $f_v^*$. Similar assumptions are used in the study of univariate \citep{segers2007multivariate} and multivariate time series \citep{jan2014}
to determine the tail chain of a Markov chain defined through stochastic recursive equations.
\begin{assumption}\label{ass_main}
	Let $\X$ follow the SCM $\mathcal M := (G, \{f_j: j\in V\}, \mathbb P_{\varepsilon})$, where the DAG $G$ is rooted at node $1$.
  Assume that for any $j\in V \setminus \{1\}$ and any sequence $\x(t) \to \x \in \mathbb R^{|\pa_G(j)|} $, we have the convergence
	\begin{align}\label{limit_cond}
		\lim_{t\to \infty} f^*_j(\x(t) + t\einsfun, e) - t   = \Psi_j(\x, e), \quad  e \in \mathbb R,
	\end{align} 
	for some functions $\Psi_j: \mathbb R^{|\pa_G(j)|} \to \mathbb R$, called the extremal structure functions. 
\end{assumption}
The next theorem shows that this assumption guarantees the existence of a limiting model for the extremes of the SCM.
 
\begin{thm}\label{thm:SCM_limit}
	Let $\X$ satisfy Assumption~\ref{ass_main}.
	Then, for any $v\in V\setminus\{1\}$, the extremal structure function is homogeneous, that is, $\Psi_v( \x + s \einsfun, e) = \Psi_v( \x, e) + s$ for all $s\in\mathbb R$. Moreover, the distribution of $\X^*$ is multivariate regularly varying and 
    \begin{align}\label{MRV1}
		\lim_{t \to \infty} \mathbb P(\X^* - t\einsfun \in A \mid  X^*_1 > t ) = \mathbb P(\Y^{1} \in A),
	\end{align}
	for any Borel subset $A\subset \mathcal L^1$ with $ \mathbb P(\Y^{1} \in \partial A) = 0$.
	The limiting random vector has the form $\Y^{1} \stackrel{d}{=} R \mathbf{1} + \W^{1}$ for a standard exponential variable $R$ that is independent of an extremal function $\W^{1} = (W^{1}_v:v\in V)$. The latter admits the structural causal model
    \begin{align}\label{eSCM}
		\begin{split}
		W^{1}_1 &:= 0, \quad W^{1}_v := \Psi_v(\W^{1}_{\pa_G(v)}, \varepsilon_v), \quad v \in V.
		\end{split}
	\end{align}
Moreover, the extremal function then satisfies the moment condition $\mathbb E \exp\{ W^{1}_v \} = 1$ for all $v \in V$.
\end{thm}
\begin{remark}
	Since $X^*_1$ is a root node, we can see the conditioning in~\eqref{MRV1} also as an intervention that sets the value of this node to a large random value with the same distribution as $X^*_1 \mid X_1^* > t$; limits after interventions on other sets of variables are studied in Section~\ref{sec:interventions_causality}.
\end{remark}
Instead of writing the extremal function as a structural causal models as in~\eqref{eSCM}, we can equivalently write the limit $\Y^{1}$ as
\begin{align}\label{eSCM2}
		Y^{1}_1 &:= R, \quad  Y^{1}_v := \Psi_v(\Y^{1}_{\pa_G(v)}, \varepsilon_v), \quad v \in V.
\end{align}
The extremal structure functions $\Psi_v$ are \emph{a priori} defined on all variables corresponding to the parents of node $v$ in the graph $G$ of the original SCM. It can however be that edges from certain parents in this graph are not inherited in the extremes, that is, the function $\Psi_v$ is constant along such parent nodes. We can therefore define an extremal graph $G_e = (V, E_e)$ with edge set $E_e$ being the smallest set such that the distribution of the random vector $\Y^1$ in~\eqref{eSCM2} is not changed when replacing $\pa_G(v)$ with $\pa_{G_e}(v)$ and adapting $\Psi_v$ accordingly for all $v\in V$.
The extremal graph $G_e$ is therefore a subgraph of $G$ in the sense that $E_e \subseteq E$. We assume this minimal representation of the extremal SCM in the sequel.
Examples illustrating this behavior are presented in Section~\ref{sec:examples} below.

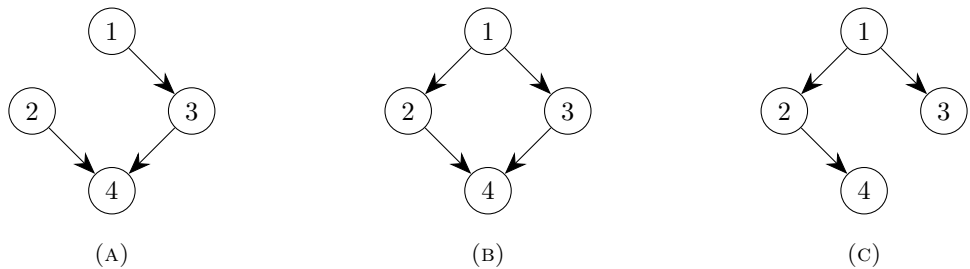
\begin{figure}[!t]
\begin{tikzpicture}[>={Stealth[length=3mm]}, node distance=1.5cm]
    \begin{scope}[local bounding box=graph1]
        \node[draw, circle] (A) {1};
        \node[draw, circle] (B) [below left of=A] {2};
        \node[draw, circle] (C) [below right of=A] {3};
        \node[draw, circle] (D) [below right of=B] {4};

        \draw[->] (A) -- (C);
        \draw[->] (B) -- (D);
        \draw[->] (C) -- (D);
        
        \node [below of=D, yshift=.75cm] {\parbox{0.3\linewidth}{\subcaption{}\label{subfig:a}}};
    \end{scope}

    \begin{scope}[shift={(5,0)}, local bounding box=graph2]
        \node[draw, circle] (1') {1};
        \node[draw, circle] (2') [below left of=1'] {2};
        \node[draw, circle] (3') [below right of=1'] {3};
        \node[draw, circle] (4') [below right of=2'] {4};

        \draw[->] (1') -- (2');
        \draw[->] (1') -- (3');
        \draw[->] (2') -- (4');
        \draw[->] (3') -- (4');
        
        \node [below of=4', yshift=.75cm] {\parbox{0.3\linewidth}{\subcaption{}\label{subfig:b}}};
    \end{scope}

	 \begin{scope}[shift={(10,0)}, local bounding box=graph2]
        \node[draw, circle] (1') {1};
        \node[draw, circle] (2') [below left of=1'] {2};
        \node[draw, circle] (3') [below right of=1'] {3};
        \node[draw, circle] (4') [below right of=2'] {4};

        \draw[->] (1') -- (2');
        \draw[->] (1') -- (3');
        \draw[->] (2') -- (4');
        
        \node [below of=4', yshift=.75cm] {\parbox{0.3\linewidth}{\subcaption{}\label{subfig:c}}};
    \end{scope}
\end{tikzpicture}
\caption{Examples of directed acyclic graphs on the nodes $V = \{1,\dots, 4\}$: a DAG with two root nodes (left); the diamond graph (center); a directed tree (right).}\label{graphs_ex}
\end{figure}
\begin{definition}\label{def_extremal_SCM}
	We call the random vector $\Y^1$ in~\eqref{eSCM2} arising as in Theorem~\ref{thm:SCM_limit} a realization from an extremal SCM $\mathcal M_e := (G_e, \{\Psi_j: j\in V\}, \mathbb P_{\varepsilon})$
	on the extremal DAG $G_e$ and with extremal structure functions  $\Psi_j: \mathbb R^{|\pa_{G_e}(j)|}\times\mathbb R \to \mathbb R$. This also uniquely induces a corresponding multivariate Pareto distribution $\mathbf Y$.
\end{definition}

If the set of extremal parents $\pa_{G_e}(v)$ consists of only one element, say $j\in V$, then the extremal structure function has the form 
$\Psi_v(x, e) = x + h(e),$
where $h:\mathbb R \to \mathbb R$ is a measurable function. Indeed, the homogeneity of $\Psi_v$ implies in this case for all $x \in \mathbb R$ that 
\[\Psi_v(x, e) = \Psi_v(0 + x, e) = x + \Psi_v(0, e).\]
We define $h(e) = \Psi_v(0, e)$ for all $e\in\mathbb R$, and note that this is a measurable function since $\Psi_v$ is defined as the limit of measurable functions in~\eqref{limit_cond}. Thus, whenever there is only one parent, we have an additive noise model with linear dependence on the parent.

\subsection{Examples}\label{sec:examples}
We illustrate Theorem~\ref{thm:SCM_limit} with examples for a random vector $\X$ that follows an SCM $\mathcal M := (G, \{f_j: j\in V\}, \mathbb P_{\varepsilon})$ on four nodes with $\G$ given by the diamond graph in Figure~\ref{subfig:b}. To simplify the examples, we assume the following structure for the first three variables
\begin{align*}
	X_1 :=f_1(\varepsilon_1)= \varepsilon_1, \qquad X_2 := f_2(X_1,\varepsilon_2)= X_1 + \varepsilon_2, \qquad X_3 := f_3(X_1,\varepsilon_3)= X_1 + \varepsilon_3,
\end{align*} 
where $\varepsilon_1$ is a standard exponential variable and $\varepsilon_j$, $j=2,3, 4$, are such that $\mathbb E \exp \{  \varepsilon_j \} = 1$ and $\mathbb E \exp \{ (1+\delta) \varepsilon_j \} < \infty$ for some $\delta > 0$. Note that with this definition, the tails of the three variables are all exponential and satisfy 
\begin{align}\label{exp_tail}
	 \lim_{q \to 1} H^{-1}(q) - F^{-1}_j(q) = 0, \quad j=2,3;
\end{align}
see \citet[Lemma 2]{zhang2023extremal}. By putting $q = F_j(u)$ and letting $u\to\infty$, this implies that $H^{-1} \circ  F_j$, $j=2,3$, are approximately the identity function, and therefore we have
\[ \lim_{t\to \infty} f^*_j(x_1 +t, e) - t   =  \lim_{t\to \infty} f_j(x_1 +t, e) - t = x_1 + e=:\Psi_j(x_1, e),\]
that is, the limiting extremal structure function is linear; as shown above, this is always the case for nodes with just one parent. We discuss two examples for the  structural assignment $X_4 := f_4(X_2, X_3, \varepsilon_4)$.

\begin{example}\label{ex:tail}
	Suppose that the structure function is given by
	\[ X_4 := f_4(X_2, X_3, \varepsilon_4) = \frac12\left\{X_2 + X_3 + \sqrt{ (X_2 - X_3)^2 + 1/(1+X_2^2 + X_3^2)}\right\} + \varepsilon_4. \]
	One can check that $X_4$ has an exponential tail up to a constant, that is, $\lim_{u\to \infty} H^{-1}\circ F_4(u) - u = c_4$ for $c_4\in\mathbb R$; see Supplementary Material~\ref{sec:app_proofs}. Therefore,
	\[ \lim_{t\to \infty} f^*_4(x_2 +t, x_3 + t, e) - t  = \max(x_2, x_3) +c_4 + e=:\Psi_v(x_2,x_3, e).\]
	In this case, the original and extremal DAGs coincide, but the structure functions differ. In particular, the original functions are not homogeneous, but the limiting $\max$ function is homogeneous.
\end{example}

\begin{example}\label{ex_different_eSCM}
	Suppose that the structure function is given by
	\[ X_4 := f_4(X_2, X_3, \varepsilon_4) = X_2 + \frac{1}{1+X_2^2 + X_3^2} + \varepsilon_4.\]	
	 With arguments very similar to Example~\ref{ex:tail},
	 we can check that $X_4$ has an exponential tail as in~\eqref{exp_tail}. Thus, $H^{-1}\circ F_4$ and $F_j^{-1}\circ H$, $j=2,3$, are approximately the identity and therefore
	\[ \lim_{t\to \infty} f^*_4(x_2 +t, x_3 + t, e) - t   =  \lim_{t\to \infty} f_4(x_2 +t, x_3 + t, e) - t = x_2 + e=:\Psi_v(x_2, e).\]
	The extremal structure function only depends on $x_2$ and the causal link from $X_3$ to $X_4$ is not present in the limit. Therefore, the minimal representation of the extremal SCM is defined on the DAG $G_e$ on the right-hand side of Figure~\ref{graphs_ex} and different from the original DAG $G$.
\end{example}

The examples above are just two simple instances of a rich class of SCMs in the domain of attraction of extremal SCMs. As extremal structure functions we already observe two different cases, namely a linear function and the maximum. In Section~\ref{sec:ESCM_connection} we will encounter several other such homogeneous functions that naturally arise from well-known models for multivariate Pareto distributions.  

\subsection{Extremal interventions}\label{sec:interventions_causality}

Theorem~\ref{thm:SCM_limit} describes extremes of the observational distribution of the SCM in~\eqref{SCM0}. The resulting extremal SCM also encodes information on the behavior of the system under extremal interventions. An intervention as defined in Definition~\ref{def:SCM_intervention} is called extremal if it is beyond the range of the available training data; see, e.g., Figure~\ref{fig:intro}. To make interventions comparable on the different marginal scales, we consider here perfect interventions $\intervene\{X_v := F_v^{-1}(q_v)\}$, where the variable $X_v$ is set to a high quantile with $q_v\to 1$. 
It will be convenient to parameterize $q_v$ through a common level $t > 0$ that governs how extreme the interventions are. We set 
$q_v = 1 - \exp(-t - \xi_v) = H(t + \xi_v)$, where $H$ is the standard exponential distribution function and the $\xi_v\in\mathbb R$, $v\in V$, quantify by how much the individual interventions deviate from the common level. 
For instance, if we have $n$ training observations, choosing $t=\log n$ yields interventions that are around or above the highest marginal observations of $X_v$, $v\in \mathcal I$.
We use here the exponential scale as baseline to define extremal interventions since additive changes correspond to changes of the extremal functions in the framework of multivariate regular variation; see Section~\ref{sec:extremal_CI}. We define extremal interventions and the corresponding SCMs.
\begin{definition}
	For a random vector $\X$ following an SCM $\mathcal M = (G, \{f_j: j\in V\}, \mathbb P_{\varepsilon})$, let $\mathcal I\subset V$ be the subset of nodes that is intervened on. We call 
	\begin{align}
		\label{extremal_do}
		\intervene\{X_v := F_v^{-1}\circ H(t + \xi_v)\} = \intervene\{X^*_v := t + \xi_v\}, \quad v\in\mathcal I,
	\end{align}
	an extremal intervention at level $t>0$ at values $\xi_v \in \mathbb R$, where $X^*_v = H^{-1}\circ F_v(\X_v)$ is defined with the distribution function $F_v$ of  $X_v$ in the original SCM. Following Definition~\ref{def:SCM_intervention}, the corresponding random vector $\tilde{\X} = \tilde{\X}^{(t)}$ follows the intervened SCM $\mathcal M^{\intervene(\X^*_\mathcal{I}:=t + \boldsymbol{\xi}_\mathcal{I})} := (G, \{\tilde f_j: j\in V\}, \mathbb P_{\varepsilon})$ with structure functions as in \eqref{f_intervened}.
\end{definition}

To introduce a notion of an extremal cause, we define the extremal limit $Y_j$ of a target variable $X_j$ under interventions on variables in $\mathcal I$ as 
\begin{align}\label{limit_Y}
	\mathbb P^{\intervene(\Y_{\mathcal I} := \boldsymbol{\xi}_\mathcal I)}_{Y_j}(A) := \lim_{t\to\infty}  \mathbb P^{\intervene(\X^*_{\mathcal I} := t + \boldsymbol{\xi}_\mathcal I)}(X^*_j - t \in A),
\end{align}
for any continuity set $A \subset [-\infty, \infty)$ of the limit,
and whenever the limit exists. Here it is important to allow also for values $-\infty$, since not all variables might be extreme after the intervention.
We call $\intervene(\Y_{\mathcal I} := \boldsymbol{\xi}_\mathcal I)$ a (limiting) extremal intervention.

\begin{definition}\label{def_extreme_cause}
	We say that $X_i$ is an extremal cause of $X_j$, if there exists a subset of indices $\mathcal I\subset V\setminus\{i,j\}$ and extremal interventions $\boldsymbol{\xi}_{\mathcal I\cup \{i\}}$ such that 
	\begin{align}\label{Y_intervened}
		 \mathbb P^{\intervene(\Y_{\mathcal I \cup\{i\}} := \boldsymbol{\xi}_{\mathcal I\cup\{i\}})}_{Y_j} \neq \mathbb P^{\intervene(\Y_{\mathcal I} := \boldsymbol{\xi}_\mathcal I)}_{Y_j}.
	\end{align}	
\end{definition}
Definition~\ref{def_extreme_cause} relates to the notion of causal relevance in \citet{galles1997axioms} and \citet{Pearl2009}, Section 7.3.3.
The rationale behind this definition is 
that for a variable $X_i$ to be an extremal cause of $X_j$, there must exist an intervention that changes the extremal limit $Y_j$. It is important to allow for a set of intervention variables $\mathcal I$ since some causes might only have an effect in combination with others; we will see examples of this later. For now, we consider two simple illustrative examples.

\begin{example}\label{ex:normal}
	Let $\varepsilon_1$ and $\varepsilon_2$ be independent, non-degenerate normal distributions with arbitrary means and variances. Define the structural causal model
	\[ X_1 := \varepsilon_1, \quad X_2 := X_1 + \varepsilon_2.\]
	It follows from~\cite{heffernan2004} that $\mathbb P^{\intervene(Y_1 := \xi_1)}_{Y_2} = \mathbb P^{\intervene(\emptyset)}_{Y_2} = \delta_{-\infty}$, where $\intervene(\emptyset)$ indicates the limit in~\eqref{Y_intervened} under the observational distribution.
	Therefore, there is no extremal causal effect from $X_1$ to $X_2$ (and neither in the other direction).
\end{example}

This example shows that the type of extremal causality that we consider here can be seen as a strong version of extremal effects. Indeed, we require that setting the cause $X_1$ to an extreme value results in a large value of the effect $X_2$ on its respective quantile scale.
For the normal distribution in the above example, this is not satisfied since the resulting value of $X^*_2$ on the exponential scale is much smaller than the intervention on $X_1^*$; see the left-hand side of Figure~\ref{fig:examples}.
In extreme value theory it is natural to compare variables on a common scale. In particular, two variables are termed asymptotically dependent if they have joint exceedances over their respective high quantiles with common probabilistic level; see the definition of extremal correlation in \cite{col1999}. In that sense, our definition is an extension of extremal dependence to causality.

\begin{figure}[tb!]
  \centering
  \begin{subfigure}{.4\textwidth}
    \centering
    \includegraphics[scale=.6]{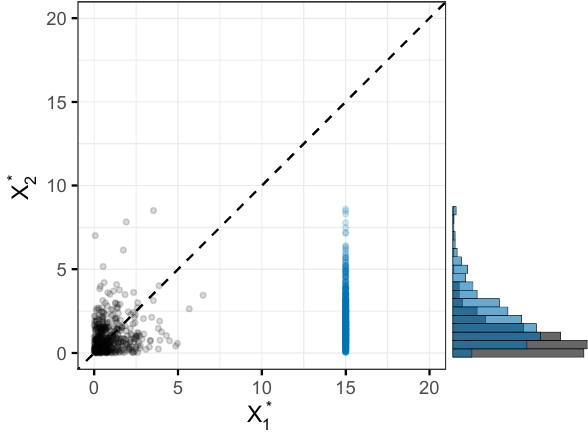}
  \end{subfigure}
  \begin{subfigure}{.4\textwidth}
    \centering
    \includegraphics[scale=.6]{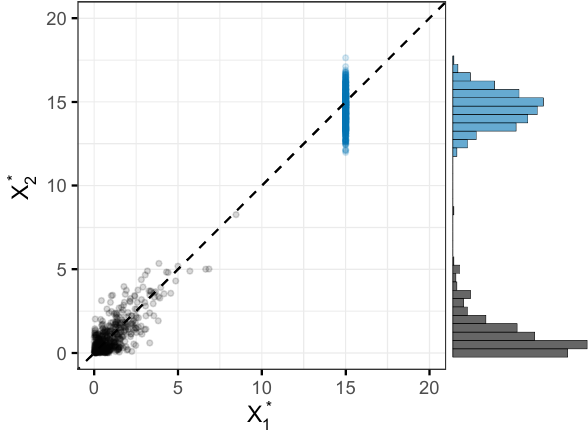}
  \end{subfigure}
  \caption{Data from from the models in Example~\ref{ex:normal} (left) and Example~\ref{ex:exp} (right) sampled from the observational distribution (black) and the distribution after an extremal intervention on $X^*_1$ (blue).
  }
  \label{fig:examples}
\end{figure}

\begin{example}\label{ex:exp}
	Let $\varepsilon_1$ follow a standard exponential distribution and, independently, let $\varepsilon_2$ be a centered normal distribution with variance $\sigma^2 > 0$. Define the structural causal model
	\[ X_1 := \varepsilon_1, \quad X_2 := X_1 + \varepsilon_2.\]
	Then $\mathbb P^{\intervene(Y_1 := \xi_1)}_{Y_2}$ equals a normal distribution $N(\xi_1 - \sigma^2/2, \sigma^2)$, which is different from $\mathbb P^{\intervene(\emptyset)}_{Y_2} = \delta_{-\infty}$.
	Therefore there is an extremal causal effect from $X_1$ to $X_2$.
\end{example}
In this example we have extremal causality since an extremal intervention on $X_1$ leads to an extreme value in $X_2$ on the same level on its quantile scale; see the right-hand side of Figure~\ref{fig:examples} and the derivation in Supplementary Material~\ref{app:ex-exp}. Note also that in contrast to Example~\ref{ex:normal}, the vector $(X_1, X_2)$ exhibits asymptotic dependence.

The next theorem characterizes the distribution under extremal interventions of general SCMs. 

\begin{thm}\label{thm:SCM_limit_intervened}
	Let $\mathbf X$ follow the SCM $\mathcal M = (G, \{f_v: v\in V\}, \mathbb P_\varepsilon)$ and suppose that Assumption~\ref{ass_main} also holds for sequences $\{\x(t)\} \subset [-\infty,\infty)^{|\pa_G(v)|}$ in the extended real line, and where the extremal functions are defined as $\Psi_v: [-\infty,\infty)^{|\pa_G(v)|} \to [-\infty,\infty)$.
	For extremal interventions at levels $t>0$ and values $\xi_v \in \mathbb R$, $v\in\mathcal I$ for some subset $\mathcal I\subset V$, we have
	\begin{align}\label{conv_dist}
		\lim_{t\to\infty}  \mathbb P^{\intervene(\X^*_{\mathcal I} := t + \boldsymbol{\xi}_\mathcal I)}(\X^* - t\einsfun \in A) = \mathbb P^{\intervene(\Y_{\mathcal I} := \boldsymbol{\xi}_\mathcal I)}(\Y \in A), 
	\end{align}
	for any continuity set $A\subset [-\infty,\infty)^d$ of the limit $\Y$, which has to be understood as the extremal SCM $\mathcal M_e^{\intervene(\Y_{\mathcal I} := \boldsymbol{\xi}_\mathcal I)} = (G_e, \{\tilde \Psi_v: v\in V\}, \mathbb P_\varepsilon)$ derived from the extremal SCM  of the observational distribution in Theorem~\ref{thm:SCM_limit}, but with structure functions 
	\begin{align}\label{tilde_Psi}
		\tilde \Psi_v(\x_{\pa(v)}, e) = 
		\begin{cases}
			\xi_v & v \in \mathcal I,\\
			-\infty & \{v \cup \an(v)\}\cap  \mathcal I = \emptyset,\\
			\Psi_v(\x_{\pa(v)}, e) & \text{otherwise.}\\
		\end{cases}
\end{align}
\end{thm}
The stronger version of Assumption~\ref{ass_main} in the above theorem can be easily checked to hold for Examples~\ref{ex:tail} and~\ref{ex_different_eSCM}, for instance. 
The intervened structure functions in~\eqref{tilde_Psi} are different from classical intervened SCMs, as we show in the next example.

\begin{figure}[tb!]
\begin{tikzpicture}[>={Stealth[length=3mm]}, node distance=1.5cm]
    \begin{scope}[local bounding box=graph1]
        \node[draw, circle] (1') {1};
        \node[draw, circle] (2') [below left of=1'] {2};
        \node[draw, circle] (3') [below right of=1'] {3};
        \node[draw, circle] (4') [below right of=2'] {4};

        \draw[->] (1') -- (2');
        \draw[->] (1') -- (3');
        \draw[->] (2') -- (4');
        \draw[->] (3') -- (4');

        \node[draw=none, fill=none, rectangle] at (1') [xshift=-.6cm, yshift=.25cm]  {\includegraphics[width=8mm]{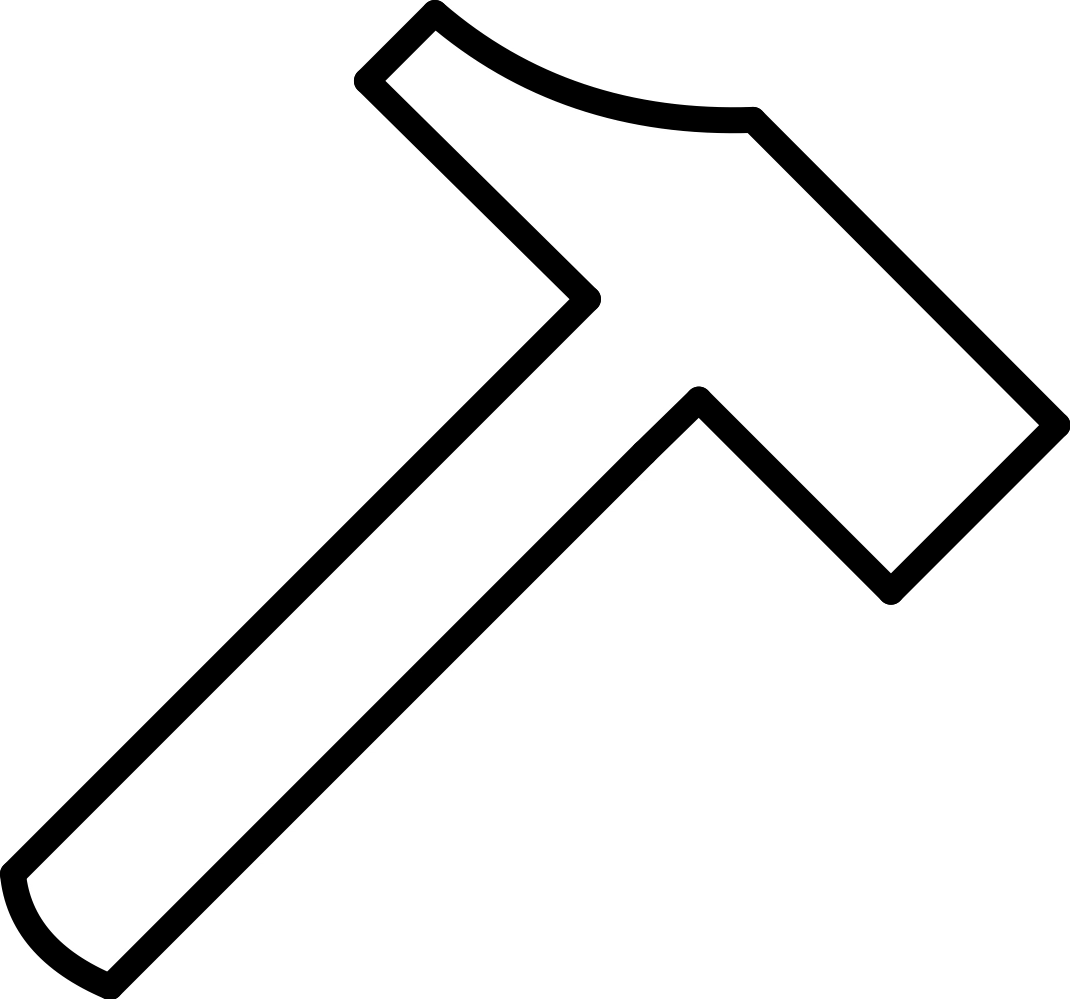}};

        \node [below of=4', yshift=.75cm] {\parbox{0.3\linewidth}{\subcaption{ $\intervene(Y_1 \coloneqq \xi_1)$}\label{subfig4:a}}};
    \end{scope}
    \begin{scope}[shift={(5,0)}, local bounding box=graph2]
        \node[draw, circle] (1') {1};
        \node[draw, circle] (2') [below left of=1'] {2};
        \node[draw, circle] (3') [below right of=1'] {3};
        \node[draw, circle] (4') [below right of=2'] {4};

        \draw[->] (1') -- (2');
        \draw[->] (1') -- (3');
        \draw[->] (2') -- (4');
        \draw[->] (3') -- (4');

        \node[draw=none, fill=none, rectangle] at (2') [xshift=-.6cm, yshift=.25cm]  {\includegraphics[width=8mm]{hammer.png}};

        \node[draw=none, fill=none, rectangle] at (3') [xshift=.6cm, yshift=.25cm]  {\includegraphics[width=8mm]{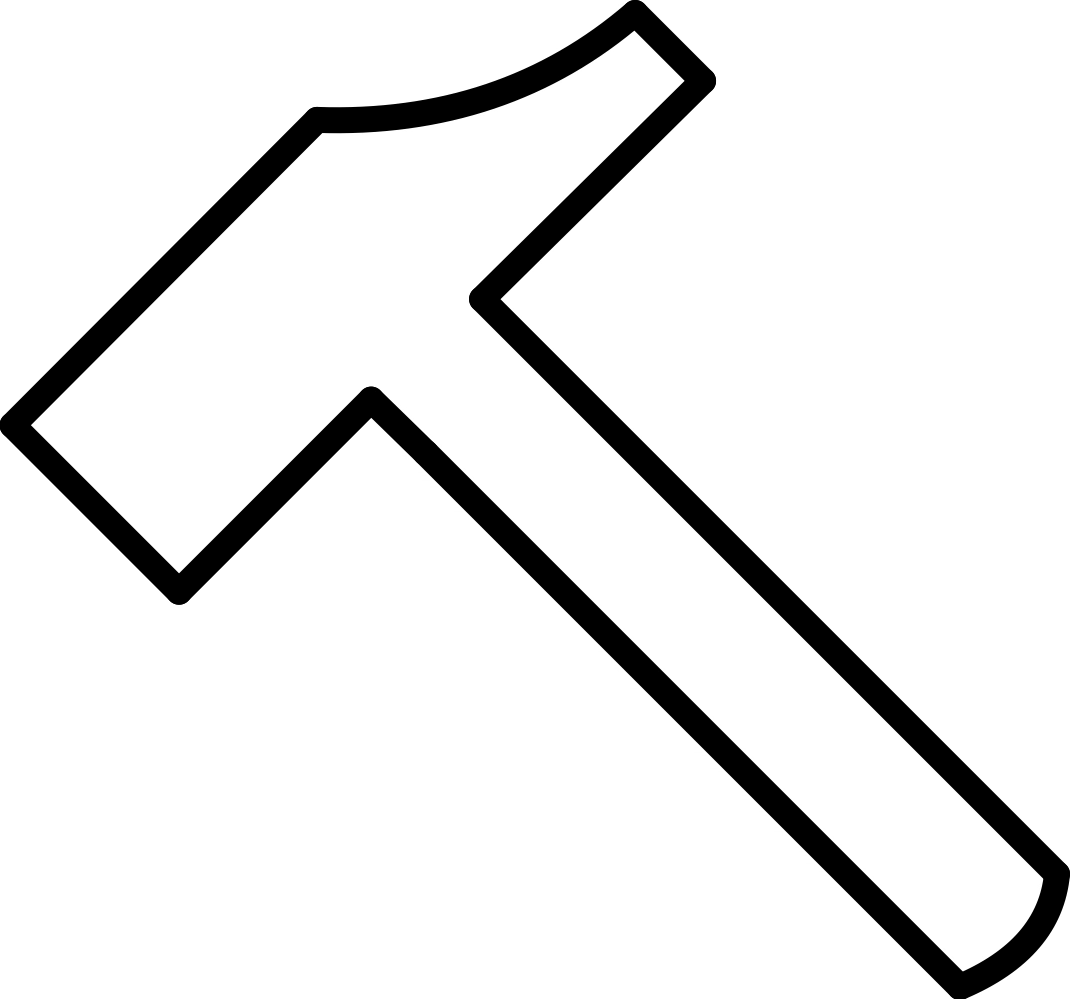}};
        
        \node [below of=4', yshift=.75cm] {\parbox{0.3\linewidth}{\subcaption{$\intervene(Y_2 \coloneqq \xi_2, Y_3 \coloneqq \xi_3)$}\label{subfig4:b}}};
    \end{scope}
	 \begin{scope}[shift={(10,0)}, local bounding box=graph2]
        \node[draw, circle] (1') {1};
        \node[draw, circle] (2') [below left of=1'] {2};
        \node[draw, circle] (3') [below right of=1'] {3};
        \node[draw, circle] (4') [below right of=2'] {4};

        \draw[->] (1') -- (2');
        \draw[->] (1') -- (3');
        \draw[->] (2') -- (4');
        \draw[->] (3') -- (4');

        \node[draw=none, fill=none, rectangle] at (3') [xshift=.6cm, yshift=.25cm]  {\includegraphics[width=8mm]{hammer-2.png}};
        
        \node [below of=4', yshift=.75cm] {\parbox{0.3\linewidth}{\subcaption{$\intervene(Y_3 \coloneqq \xi_3)$}\label{subfig4:c}}};
    \end{scope}
\end{tikzpicture}
\caption{Examples of three different configuration of extremal interventions on the extremal DAG $G_e$ in Figure~\ref{fig:examples}(b). The hammer indicates the variables on which interventions are performed.}
\end{figure}

\begin{example}
  Suppose the assumptions of Theorem~\ref{thm:SCM_limit_intervened} are satisfied for the SCM $\X$ with limit $\Y$ on the extremal DAG $G_e$ in Figure~\ref{fig:examples}(b). Recall that the first extremal structure functions are
  \begin{align*}
	    \Psi_1(e) = e, \quad 
	    \Psi_2(y_1, e) = \Psi_3(y_1, e) = y_1 + e,
	  \end{align*}
  and $\Psi_4(y_2, y_3, e)$ depends on the specific model.
  We consider different intervention scenarios.
  \begin{itemize}
	    \item[(i)] Consider the extremal intervention $\intervene(Y_1 \coloneqq \xi_1)$ in Figure~\ref{subfig4:a}, for $\xi_1 \in \RR$. The resulting intervened  structure functions are $ \tilde{\Psi}_1(e) = \xi_1$, and $\tilde\Psi_v = \Psi_v$ for $v=2,3,4$.
	    This agrees with the classical definition of intervened SCMs. All variables are extreme in this scenario as the intervention propagates from the root node.
	
	    \item[(ii)] Consider the extremal intervention $\intervene(Y_2 \coloneqq \xi_2, Y_3 \coloneqq \xi_3)$ in Figure~\ref{subfig4:b} for $\xi_2, \xi_3 \in \RR$. The resulting intervened  structure functions are
	    \begin{align*}
		      \tilde{\Psi}_1(e) = -\infty,\quad 
		      \tilde{\Psi}_2(y_1, e) = \xi_2, \quad
		      \tilde{\Psi}_3(y_1, e) = \xi_3, 
		    \end{align*}
	    and $\tilde{\Psi}_4 = {\Psi}_4$.
	    Unlike the classical definition of intervened SCMs, if neither $v$ or any ancestor $\an(v)$ have been intervened on, then $Y_v$ will take value $-\infty$; this indicates that in the extreme scenario produced by the intervention, these nodes are not extreme.
	
	    \item[(iii)] Under the extremal intervention $\intervene(Y_3 \coloneqq \xi_3)$ in Figure~\ref{subfig4:c}, for $\xi_3 \in \RR$, the intervened structure functions for the first three variables are
	    \begin{align*}
		      \tilde{\Psi}_1(e) = -\infty,\quad 
		      \tilde{\Psi}_2(y_1, e) = -\infty, \quad
		      \tilde{\Psi}_3(y_1, e) = \xi_3,\quad
		    \end{align*}
	    The explicit expression of $\tilde{\Psi}_4$ depends on the definition of $\Psi_4$.
	    For example, if $\Psi_4(y_2, y_3, e) = \min(y_2, y_3) + e$, then we have $\tilde{\Psi}_4(y_2, y_3, e) = -\infty$, or if $\Psi_4(y_2, y_3, e) = \max(y_1, y_2) + e$, then $\tilde{\Psi}_4(y_2, y_3, e) = \xi_3 + e$.
	    Like in classical intervened SCM, the structure functions of the descendants of an intervened node do not change; this follows from  the modularity principle \citep{Haavelmo1944}. However, evaluating such structure functions under extreme interventions might yield $-\infty$.
	  \end{itemize}
\end{example}

We can now connect the notion of direct extremal causes to the DAG  $G_e$ from the extremal SCM.
\begin{proposition}\label{prop:extr_cause_extr_scm}
  Consider an SCM $\mathcal M = (G, \{f_v: v\in V\}, \mathbb P_\varepsilon)$ and let $\X$ be the corresponding random vector. Suppose the assumptions of Theorem~\ref{thm:SCM_limit_intervened} are satisfied. Then, for $i,j\in V$, $X_i$ is a direct extremal cause  of $X_j$ as in Definition~\ref{def_extreme_cause} if and only if $i \in \pa_{G_e}(j)$, that is, $X_i$ is a parent of $X_j$ in the extremal DAG $G_e$.
\end{proposition}

\section{Directed extremal graphical models}\label{sec:DEGM}

Relying on a novel notion of extremal conditional independence, \cite{EH2020} introduced extremal graphical models on undirected graphs through the pairwise Markov property. Recent work has established various tools for inference and structure learning in this model class. In this section, we define directed extremal graphical models on a DAG $G=(V,E)$ that satisfy suitable directed Markov properties. We show that the limiting models arising from SCMs in Section~\ref{sec:eSCM} are directed extremal graphical models with respect to their extremal graph $G_e=(V,E_e)$ in Definition~\ref{def_extremal_SCM}.

To simplify notation in this section, we use a generic graph $G = (V,E)$ to define extremal graphical models and denote with $G_e$ only the graph arising from the extremes of an SCM. We also write $\pa(v)$ instead of $\pa_G(v)$ for the parents of a node $v\in V$ in the graph $G$.

\subsection{Directed extremal graphical models}\label{sec:dEGM}

The conditional independence notion $\perp_e$ for multivariate Pareto distributions in Definition~\ref{def:extremal_indep} can be used to define a directed version of extremal graphical models on a directed acyclic graph (DAG) $G=(V,E)$. We write $A\indep_{G} B \mid C$ when $C$ d-separates $A$ and $B$ in the DAG $G$; see Definition~\ref{def:d-sep} in the Supplementary Material.

\begin{definition}\label{dEGM}
	Let $G = (V,E)$ be a DAG. A multivariate Pareto distribution $\Y$ with positive and continuous exponent measure density is a directed extremal graphical model on $G$ if its distribution satisfies the extremal global Markov property relative to $G$, that is,
	\[ A\indep_{G} B \mid C \quad \Rightarrow \quad \Y_A\perp_e \Y_B \mid \Y_C, \]
	for all disjoint index sets $ A,B,C\subset V $.
\end{definition}

From the results in \cite{EH2020} we can directly derive a key property of the extremal graph structure $G$. Since a multivariate Pareto distribution $\Y$ with positive and continuous exponent measure density can not exhibit extremal independence $\Y_A \perp_e \Y_B$ for index sets $A,B\subset V$, the DAG $G$ must always be rooted, that is, there exists a node $v^*\in V$ such that for any other node $v\in V$ there is a directed path from $v^*$ to $v$. Indeed, otherwise there would be two nodes $v,w \in V$ that are d-separated by the empty set, that is, $v \indep_G w$, resulting in the contradiction $Y_v\perp_e Y_w$. 

In the classical theory of graphical models, under some conditions, the directed global Markov property is equivalent to other Markov properties such as the local Markov property or the Markov factorization property  \citep[e.g.,][Definition~6.21]{PJS2017}.
We introduce the corresponding extremal Markov properties and show that directed extremal graphical models can equivalently be defined through those conditions.

\begin{proposition}\label{prop:ext-markov-properties}
	Let $G=(V,E)$ be a DAG and let $\Y$ be a multivariate Pareto distribution with positive and continuous exponent measure density $\lambda$. Then the following are equivalent.
	\begin{enumerate}
		\item[(a)] The distribution of $\Y$ satisfies the extremal global Markov property relative to $G$.
		\item[(b)] The distribution of $\Y$ satisfies the extremal local Markov property relative to $G$, that is, 
		\[ Y_i \perp_e \Y_{\nd(i)\setminus \pa(i)} \mid \Y_{\pa(i)},\quad  \text{for all } i \in V.\]
		\item[(c)] The distribution of $\Y$ satisfies the extremal Markov factorization property relative to $G$, that is, the exponent measure density factorizes 
		\begin{align}\label{density_factorization}
			\lambda(\y) & = \prod_{v\in V} \lambda(y_v\mid \y_{\pa(v)}), \quad \forall \y \in \mathbb R^d,
		\end{align}
		where $\lambda(y_v\mid \y_{\pa(v)}):= \lambda(y_v,\y_{\pa(v)}) /\lambda(\y_{\pa(v)}) $ is the conditional exponent measure density.
	\end{enumerate}
\end{proposition}

The last equivalence is most relevant since it tells us that extremal directed graphical models have densities that factorize into lower-dimensional parts. This yields a sparse representation of the extremal dependence in these models. Interestingly, while $\lambda$ is the density of the infinite measure $\Lambda$, the conditional densities are probability densities.
This follows from the simple calculation
\[\int \lambda(y_v\mid \y_{\pa(v)}) \mathrm d y_v = \frac{1}{\lambda(\y_{\pa(v)})} \int \lambda(y_v,\y_{\pa(v)}) \mathrm d y_v  = 1,\]
where the last equation follows since the integration results in the marginal density $\lambda(\y_{\pa(v)})$.
Another important property of the conditional densities, which is inherited from the homogeneity of the exponent measure density, is the fact that $\lambda(y_v + t\mid \y_{\pa(v)}+ t\bs 1) = \lambda(y_v\mid \y_{\pa(v)})$ for all $\y\in \mathbb R^d$ and $t \in \mathbb R$.

The density factorization~\eqref{density_factorization} can be used as a principle for constructing a valid multivariate Pareto density in $d$ dimensions. Indeed, suppose that for all $v \in V$, the conditional density $\lambda(y_v\mid \y_{\pa(v)})$ arises from a valid multivariate Pareto density in dimension $|\pa(v)| + 1$. Then the density defined in~\eqref{density_factorization} is the density of a multivariate Pareto distribution. We give here a few examples of valid conditional densities; for derivations see the Supplementary Material.
\begin{example}\label{ex:logistic}
	Recall the exponent measure density of a $d$-dimensional extremal logistic distribution in Example~\ref{ex:logistic_prelim}.
	The corresponding logistic conditional density for node $v$ with parents $\pa(v)$ is
	\[ \lambda(y_v\mid \y_{\pa(v)}) = \left(1 + \frac{1}{\sum_{i \in \pa(v)} \exp\left\{ \frac{y_v - y_i}{\theta}\right\} } \right)^{\theta-|\pa(v)| - 1}	\frac{|\pa(v)|/ \theta-1}{\sum_{i \in \pa(v)} \exp\left\{ \frac{y_v-y_i}{\theta}\right\} }.   \]
\end{example}

\begin{example}\label{ex:dirichlet}
	Recall the exponent measure density the extremal Dirichlet distribution in Example~\ref{ex:dirichlet_prelim}. The corresponding Dirichlet  conditional density for node $v$ with parents $\pa(v)$ is
   \begin{align*} \lambda(y_v\mid \y_{\pa(v)}) &=  C\times   \left( 1 + \frac{1}{\sum_{j \in \pa(v)} \frac{\alpha_j}{\alpha_v} \exp\{y_j - y_v\}}\right)^{-1 -\sum_{i\in\{v\}\cup\pa(v) }\alpha_i} 
   \left( \sum_{j \in \pa(v)} \frac{\alpha_j}{\alpha_v} \exp\{y_j - y_v\} \right)^{-\alpha_v}
   \end{align*}
   for some constant $C$ depending on the $\alpha_i$. 
\end{example}

\begin{example}\label{ex:HR_cond}
	Recall the exponent measure density the \HR{} distribution in Example~\ref{ex:HR_cond_prelim}. The conditional density for node $v$ with parents $\pa(v)$ is
	 \begin{align*}
	 	\lambda(y_v\mid \y_{\pa(v)}) &=  \frac{1}{\sqrt{2\pi \vartheta_{vv}^{-1}}}\exp\left(-\frac{1}{2\vartheta_{vv}^{-1}} (y_v-\mu^*)^2 \right),
	 \end{align*}
	 where $\vartheta=\theta(\{v\}\cup \pa(v))$ and
	 \begin{align*}
	 	\mu^*=\begin{pmatrix}
	 		-\frac{1}{2}\Gamma_{v,\pa(v)}, 1\\
	 	\end{pmatrix}\begin{pmatrix}
	 	-\frac{1}{2}\Gamma_{\pa(v),\pa(v)} & \mathbf{1}\\
	 	\mathbf{1}&0\\
	 	\end{pmatrix}^{-1}\begin{pmatrix}
	 	\y_{\pa(v)}\\
	 	1\\
	 	\end{pmatrix}.
	 \end{align*}
\end{example}
Note that it is not required that the joint density is logistic or \HR{}, but one may use any combination of valid conditional densities for the different nodes $v \in V$ in~\eqref{density_factorization}.

\subsection{Connection to extremal SCMs}\label{sec:ESCM_connection}

In Section~\ref{sec:eSCM} we have shown that the extremes of structural causal models converge to a multivariate Pareto distribution with a, possibly different, extremal causal structure $G_e$. A natural question is whether this limit is a directed extremal graphical model in the sense of Definition~\ref{dEGM} on the DAG $G_e$. The answer is positive.
\begin{proposition}\label{prop:extr_scm_is_dgm}
	Let $\Y$ be a multivariate Pareto distribution, and let $G_e = (V,E_e)$ be a DAG. If $\Y$ is an extremal SCM on $G_e$ then it is also a directed extremal graphical model on $G_e$. 
\end{proposition}
To make this connection more explicit, we identify the extremal structure functions $\Psi_v$ corresponding to the examples of conditional exponent measure densities in Section~\ref{sec:dEGM}. This also provides additional examples for $\Psi_v$ that go beyond the linear and maximum functions encountered in Section~\ref{sec:eSCM}.
In fact, it will turn out that all of the examples below have the form of an additive noise models
\begin{align}\label{Phi_simple}
	  \Psi_v(\X_{\pa(v)}, \varepsilon_v) =  \Phi_v(\X_{\pa(v)}) +  \varepsilon_v, 
\end{align}
where $\Phi_v: \mathbb R^{|\pa(v)|} \to \mathbb R$ is a homogeneous structure function only acting on the parents of $v\in V$.
Note that here and for the remainder of the section, we abbreviate for convenience $\pa(v)=\pa_{G_e}(v)$.

\begin{example}
	Consider the model in Example~\ref{ex:logistic} where the conditional exponent measure density was constructed based on the logistic distribution. We then have an additive noise model as in~\eqref{Phi_simple} with
	\[ \Phi_v(\X_{\pa(v)}) = -\theta\log \sum_{i \in {\pa(v)}} \exp(-X_i / \theta), \]
	and the independent noise variable $\varepsilon_v$ has distribution function $F_{\varepsilon_v}(y) = (1 + e^{-y/\theta})^{\theta - |\pa(v)|}$, $y \in \mathbb R$.
\end{example}

\begin{example}
	Consider the model in Example~\ref{ex:dirichlet} where the conditional exponent measure density was constructed based on the Dirichlet distribution. We then have an additive noise model as in~\eqref{Phi_simple} with
	\[ \Phi_v(\X_{\pa(v)}) = \log \sum_{i \in {\pa(v)}}\frac{\alpha_i}{\alpha_v} \exp(X_i ), \]
	and the independent noise variable $\varepsilon_v$ has density function $f_{\varepsilon_v}(y) = C (1 + e^{y})^{-1 - \sum_{i\in\pa(v)\cup c} \alpha_i} e^{\alpha_v y}$, $y \in \mathbb R$.
\end{example}

\begin{example}\label{ex:HR_struct}
	Consider a H\"usler--Reiss vector $\Y$.
	From Example~\ref{ex:HR_cond} we obtain
	\begin{align*}
		\Phi_v(\X_{\pa(v)})&=\left(-\frac{1}{2}\Gamma_{v,\pa(v)}\theta(\pa(v))+\mathbf{p}^{\top}(\pa(v))\right)\X_{\pa(v)},
	\end{align*}
where the noise variables $\varepsilon_{v}$ are univariate Gaussians with mean $-\frac{1}{2}\Gamma_{v,\pa(v)}\mathbf{p}(\pa(v))+\sigma^2(\pa(v))$ and variance $\vartheta_{vv}^{-1}$. Note that the functions $\Phi_v$ are linear combinations of $\X_{\pa(v)}$.
\end{example}

While not all extremal SCMs have an additive noise representation~\eqref{Phi_simple}, it seems to cover the most popular examples used in the literature. We also use this representation as a starting point to define new extremal SCMs, and therefore also directed extremal graphical models. 

\begin{proposition}\label{prop:ESCM}
	Let $G_e = (V,E_e)$ be a DAG with one root node, say node 1. Define a $d$-dimensional random vector $\g X$ by 
	\begin{align}\label{additive_noise}
		X_1 := R,\qquad	 X_v := \Phi_v(\X_{\pa(v)}) +  \varepsilon_v, \quad v\in V\setminus \{1\},
	\end{align}	
	where $R$ is a standard exponential variable, $\varepsilon_v$ are independent noise variables that admit densities and $\Phi_v:\mathbb R^{|\pa(v)|} \to \mathbb R$ are homogeneous  functions satisfying $\Phi_v(\x + t\mathbf{1}) = \Phi_v(\x) + t$ for any $x \in \mathbb R^{|\pa(v)|}$ and $t\in\mathbb R$.	
	Assume the normalization condition
	\begin{align}\label{norm_constraint}
		\EE\{\exp(\varepsilon_v)\}=\frac{1}{\EE[\exp\{\Phi_v(\X_{\pa(v)}-X_1\mathbf{1})\}]}.
	\end{align}
	Then, we obtain a multivariate Pareto vector $\Y$ by choosing the auxiliary vector $\Y^1\stackrel{}{=}\X$, which is a directed extremal graphical models on $G_e$. 
\end{proposition}

\subsection{Linear extremal SCMs}

An important sub-class of SCMs are those corresponding to linear structural assignments introduced in Example~\ref{eq:linSCM}. In this section, we discuss the corresponding assignments for extremal SCMs. Let $G_e = (V,E_e)$ be a rooted DAG with root node $1$. Following the construction principle in Proposition~\ref{prop:ESCM}, let 
\begin{align}\label{linear_eSCM}
	X_1 &:= R, \quad  X_v :=\sum_{i \in {\pa(v)}} b_{iv} X_i + \varepsilon_v, \quad v\in V\setminus \{1\},
\end{align}
where $R$ is a standard exponential variable, $b_{iv}\in\mathbb R$ are the causal coefficients and the independent noise variables $\varepsilon_v$ satisfy normalization constraint~\eqref{norm_constraint}. Since the extremal structure function $\Phi_v(\x_{\pa(v)})=\sum_{i \in {\pa(v)}} b_{iv} x_i$ has to be homogeneous, it follows that 
\begin{align}\label{zerosum}
	\sum_{i \in {\pa(v)}}b_{iv}=1, \quad v\in V\setminus\{1\}.
\end{align}
By Proposition~\ref{prop:ESCM} it follows that for $\Y^1\stackrel{d}{=}\X$ the multivariate Pareto distribution $\Y$ is a directed extremal graphical model on $G_e$.
Let $B$ be the $d\times d$ coefficient matrix with entries $B_{iv} = b_{iv}$ for $v\in V$ and $i\in \pa(v)$, and zero otherwise. This allows a more compact representation of the structural assignments as
\begin{align}
	\Y^1=B^{\top} \Y^1 + \mathbf N^1,\label{eq:SEMin1}
\end{align}
where $\mathbf N^1 = (R,\varepsilon_2,\ldots,\varepsilon_d)^{\top}$ is the noise vector. 
With $\mathfrak{L}(G_e):=(I-B^{\top})_{\setminus 1,V}$ we can rewrite~\eqref{eq:SEMin1} as
\begin{align}
	\Y^1=\begin{pmatrix}
		e_1^{\top}\\
		\mathfrak{L}(G_e)\\
	\end{pmatrix}^{-1}\N^1,\label{eq:ESCM_Kirchhoff}
\end{align}
where $e_1$ is the first unit vector. Note that $\mathfrak{L}(G_e)\bs 1 = \bs 0$ by the constraint~\eqref{zerosum}.

In Example~\ref{ex:HR_struct}, we have seen that for the \HR{} model, $\Y^1$ can be constructed as a structural causal model with linear structure functions $\Phi_v$ and independent additive Gaussian noise. 
We find a construction method for \HR{} extremal SCMs in the following proposition.
\begin{proposition}\label{prop:HR_SCM_construction}
	Let $G_e=(V,E_e)$ be a DAG rooted in 1. We endow $G_e$ with real edge weights $b_{ij}$, where $b_{ij}=0$ when $ij\not\in E_e$. Let $B$ be the $d\times d$-matrix with $B_{ij}=b_{ij}$ and assume that it satisfies $\mathfrak{L}(G_e)\bs 1 = \bs 0$. For $\nu_2^2,\ldots,\nu_{d}^2 >0$, let 	$(\varepsilon_2,\ldots,\varepsilon_d)\sim N_{d-1}\left(\mu_\varepsilon,D_\varepsilon\right)$ with
	\begin{align*}
		\mu_\varepsilon&= -\frac{1}{2}\mathfrak{L}(G_e)\Gamma_{V,1}, \qquad D_\varepsilon=\diag\left(\nu_2^2,\ldots,\nu_d^2\right),
	\end{align*} 
	where the matrix $\Theta = \mathfrak{L}(G_e)^{\top} D_\varepsilon^{-1}\mathfrak{L}(G_e)$ is a positive semidefinite signed Laplacian matrix with corresponding variogram matrix $\Gamma$. Then, the structural equation model 
	\begin{align*}
		\Y^1:=B^{\top}\Y^1+\N^1
	\end{align*}
	with $\N^1=(R,\varepsilon_2,\ldots,\varepsilon_d)$ and $R\sim \text{Exp}(1)$, where $R$ is independent of $(\varepsilon_2,\ldots,\varepsilon_d)$, yields a \HR{} extremal SCM on $G_e$ with precision matrix~$\Theta$.
\end{proposition}

We give an explicit example to illustrate this construction in Supplementary Material~\ref{app:ex-HR-SEM}.

\section{Structure learning for extremal DAGs}\label{sec:structure_learning}
We propose two algorithms to learn the extremal causal structure from data, both relying on a test for extremal conditional independence presented in Section~\ref{sec:ex_cond_ind_test}. The first algorithm, shown in Section~\ref{sec:ext-pc}, adapts the PC-algorithm \citep{spirtes2000causation} to recover the extremal Markov equivalence class.
The second algorithm, called extremal pruning, recovers the extremal DAG from Definition~\ref{dEGM} and is presented in Section~\ref{sec:ext-pruning}.

\subsection{Extremal conditional independence test}\label{sec:ex_cond_ind_test}

Let $\Y$ be a multivariate Pareto vector following a \HR{} model  with parameter matrix $\Gamma$ and precision matrix $\Theta$.
For two indices $i, j \in V$ and a conditioning set $S \subseteq V \setminus \{i, j\}$, we would like to perform the hypothesis test
\begin{align}\label{eq:hyp-test}
H_0: Y_i\perp_{e}Y_j \mid \Y_S
\quad\text{versus}\quad
H_1: Y_i \not\perp_{e}Y_j \mid \Y_S.
\end{align}
Following~\eqref{eq:CIviaGamma} in Example~\ref{ex:CI_test} and Lemma~\ref{lem:precision_CI} in the Supplementary Material, for \HR{} distributions the null hypothesis $H_0$ can be equivalently expressed in terms of a vanishing extremal partial correlation coefficient, that is,
\begin{align}
	\rho_{ij \mid S} \coloneqq -\theta(\{i,j\}\cup S)_{ij} / \sqrt{\theta(\{i,j\}\cup S)_{ii} \theta(\{i,j\}\cup S)_{jj}}=0.\label{eq:partial_cor}
\end{align}\label{eq:hyp-test-2}
We can thus restate~\eqref{eq:hyp-test} as
\begin{align}\label{eq:hyp-test-3}
H_0 : \rho_{ij \mid S} = 0
\quad\text{versus}\quad
H_1: \rho_{ij \mid S} \neq 0.
\end{align}
Recall from Example~\ref{ex:CI_test} that the extremal function $\g W^k_{\setminus k} = (\Y^k-Y_k\einsfun)_{\setminus k}$ follows a $(d-1)$-dimensional normal distribution with precision matrix $\Theta^{(k)}$. For $k\in S$, Lemma~\ref{lem:partial_correlation} in the Supplementary Material shows that the corresponding partial correlation $\rho_{ij \mid S}^{(k)}$ is equal to the extremal partial correlation $\rho_{ij \mid S}$.
To obtain an estimator $\hat \rho_{ij \mid S}$ of the extremal partial correlation of $\Y$, we can therefore use an estimator $\hat \rho_{ij \mid S}^{(k)}$ of the partial correlation of the extremal function $\g W^k_{\setminus k}$ for some $k\in S$. For $n_k$ i.i.d.~observations of this normal random vector, we can then define the Fisher z-transform
\begin{align}\label{eq:z-transform}
  Z_{ij\mid S}^{(k)} = \frac{1}{2} \log\left(\frac{1 + \hat\rho_{ij \mid S}^{(k)}}{1 - \hat\rho_{ij \mid S}^{(k)}}\right),
\end{align}
which is approximately normally distributed with standard deviation $1/\sqrt{n_k - (|S|-1) - 3}$ \citep{kalisch2007estimating}. 
Note that the size of the conditioning set is $|S|-1$ as we have removed the $k$th component of the extremal function.

Using this z-transform yields a valid test for the null hypothesis in~\eqref{eq:hyp-test}. To improve the power of this test, in practice, we would like to use observations of $\Y$ on the whole support, rather than only from the half-space support of some $\Y^k$. To this end, we use the empirical estimator of the extremal variogram $\widehat \Gamma$ introduced in \cite{eng_vol_2022}. This estimator is consistent under mild assumptions, and the continuous mapping theorem implies that the plug-in estimator $\hat \rho_{ij \mid S}$ in~\eqref{eq:partial_cor} is consistent. 
While concentration bounds for $\widehat \Gamma$ are available \citep{engelke2022a}, its asymptotic distribution has not yet been established and is beyond the scope of this paper. We conjecture however that the z-transform $Z_{ij\mid S}$ based on $\hat \rho_{ij \mid S}$ is asymptotically normal with standard deviation $1/\sqrt{n - |S| - 3}$, where $n$ is the sample size of the random vector $\Y$. We provide empirical evidence for this conjecture with the following simulation study.

We fix the DAG $G = (V, E)$ and the extremal DAG $G_e = (V, E_e)$ as in Figure~\ref{subfig:b} and Figure~\ref{subfig:c}, respectively, and repeat 50 times the following experiment.
The first data set (a) are samples from an exact limiting multivariate Pareto distribution $\Y$ factorizing according to $G_e$. The other two data sets are samples from a random vector $\X$ in the domain of attraction of $\Y$ satisfying~\eqref{eq:MPD}. For (b) we sample from the corresponding max-stable distribution \citep[e.g.,]{EH2020} and for (c) from an SCM on $G$ and extremal limit $\Y$ on $G_e$. 
For (b) and (c) we use $n=1000$ samples of the respective models. To obtain approximate samples of the multivariate Pareto distribution, we empirically normalize the data to standard exponential margins as in~\eqref{X_star}, and threshold the data as in~\eqref{eq:MPD} using for $u$ the $\tau$th empirical quantile of $\max_{i=1,\dots d} X^*_i$ for $\tau \in \{0.9, 0.95, 0.975\}$, resulting in $m = \lfloor n(1-\tau)\rfloor $ effective samples. To have a fair comparison, we use $m$ samples of the exact multivariate Pareto distribution in (a). 

Given the generated data sets, for all  pairs $(i, j) \in V \times V$, with $i < j$, and for all subsets $S \subseteq V \setminus \{i, j\}$ we test~\eqref{eq:hyp-test-2} with the following two approaches:
\begin{itemize}
  \item random: the z-transform $Z_{ij\mid S}^{(k)}$ as defined in~\eqref{eq:z-transform}, by choosing one index $k \in S$ randomly;
  \item  average: the z-transform  $Z_{ij\mid S}$ based on $\hat \rho_{ij \mid S}$ obtained by the empirical estimator of the extremal variogram $\hat{\Gamma}$.
\end{itemize}
Figure~\ref{fig:ci-test} shows the proportion of p-values below a significance level $\alpha \in (0, 1)$ under the null (blue) and alternative (orange) hypothesis using the random (dashed line) and average (solid line) tests.
For data from the limiting model in (a), the level is always correct, but the power depends on the sample size.
For data in the domain of attraction in (b) and (c), the power 
improves if the threshold is large enough, since then the threshold exceedances converge to $\Y$ in distribution. Overall, the average approach that makes more efficient use of the data has higher power and still a good calibration of  Type-I error, especially for large thresholds~$\tau$.

\begin{figure}[!ht]
  \centering
  \includegraphics[scale=.75]{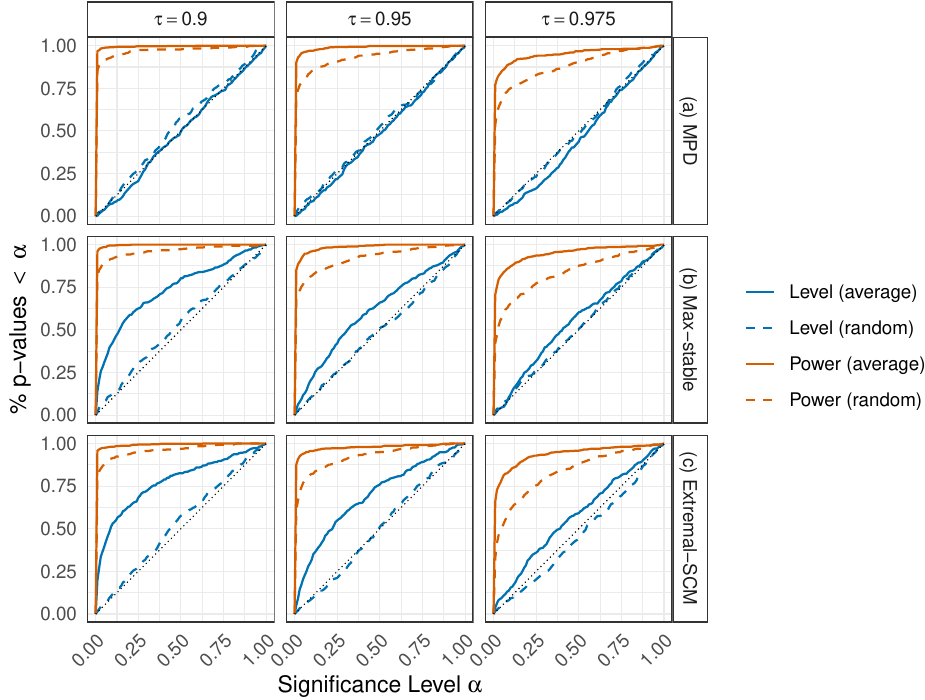}
  \caption{Percentage of estimated p-values below a significance level $\alpha \in (0, 1)$ under the null hypothesis $H_0 : \rho_{i, j \mid S} = 0$ (blue curves) and the alternative $H_1: \rho_{i, j \mid S} \neq 0$ (orange curves) for different methods and threshold level $\tau \in \{0.9, 0.95, 0.975\}$ using both the average (solid line) and random (dashed line) tests, as described in the main text.
  }
  \label{fig:ci-test}
\end{figure}

\subsection{Extremal PC-algorithm}\label{sec:ext-pc}

For classical directed graphical models, the PC-algorithm of \cite{spirtes2000causation} recovers the completed partial DAG (CPDAG) from conditional independence statements.
In this section we show that most of the concepts from the classical PC-algorithm literature translate to extremal settings.

Let $ \mathcal{P} $ denote the set of multivariate Pareto distributions and let $G_e$ be a rooted DAG and define the set
\begin{align}
  \mathcal{M}(G_e):=\{P\in\mathcal{P}:P\text{ is extremal global Markov to }G_e  \}.\label{eq:emec}
\end{align}
Two rooted DAGs $ G_e$ and $H_e $ are called extremal Markov equivalent if $\mathcal{M}(G_e)=\mathcal{M}(H_e)$. The extremal Markov equivalence class of $ G_e$ is the set of all rooted DAGs that are extremal Markov equivalent to $G_e$.
In classical graphical models, deciding on whether two DAGs are Markov equivalent is greatly simplified by the result from \citet{verma1990causal} which states that two DAGs are Markov equivalent if and only if they share the same skeleton and v-structures; for definition of skeleton and v-structures see Supplementary Material~\ref{app:d-separation}.  We extend this result to directed extremal graphical models.
\begin{corollary}\label{cor:skel_v-struct}
	Rooted DAGs that have the same skeleton graph and the same v-structures encode the same extremal conditional independence statements, that is, they are extremal Markov equivalent.
\end{corollary}
As in the standard graphical model case, we can represent the extremal Markov equivalence class of $G_e$ as a completed partial DAG (CPDAG).
The CPDAG of $\mathcal{M}(G_e)$ contains a directed edge between nodes $i$ and $j$ when there exists a DAG $H_e$ with $\mathcal{M}(G_e)=\mathcal{M}(H_e)$ such that $(i, j) \in E(H_e)$.
When the CPDAG contains directed edges from $i$ to $j$ and from $j$ to $i$, this usually is represented by an undirected edge between $i$ and $j$.
Based on the extremal Markov properties, we can deduce a set of conditional independence statements in $\Y$ by reading off d-separation statements in the graph $G_e$. In structure learning we usually need the converse implication, which corresponds to a notion extremal faithfulness.
\begin{definition}
	A multivariate Pareto vector $ \Y $ is extremal faithful with respect to a DAG $ G_e=(V,E_e) $ if
	\[\Y_A\perp_e \Y_B\mid \Y_C \Rightarrow \Y_A\indep_{G_e} \Y_B\mid \Y_C\]
	for all disjoint $ A,B,C\subset V $.\label{def:faithful}
\end{definition}
By construction, for a faithful extremal directed graphical model there is a one-to-one correspondence between conditional independence and d-separation statements. This motivates our extremal PC-algorithm to learn the CPDAG of $G_e$ in two steps. Step (i) modifies the classical PC-algorithm \citep{spirtes2000causation} to learn the skeleton of $G_e$ and is detailed in Algorithm~\ref{alg:ext-pc}. The key idea is to leverage the fact that the resulting graph must be rooted. Given this additional information, we avoid testing unconditional independencies; in other words we consider candidate separating sets $S$ with cardinality $|S| = \ell$ with $\ell \geq 1$. 
Step (ii) follows the same procedure as the non-extremal version of the algorithm: it extends the skeleton to a CPDAG by first orienting v-structures and then applying Meek's rules [\citealp[p. 50]{Pearl2009}, \citealp{meek1995}] to orient as many remaining edges as possible.
\begin{algorithm}[tb!]
  \caption{Extremal PC-skeleton learning (oracle version)}
  \label{alg:ext-pc}

  \begin{algorithmic}[1]
    \Require{Vertex nodes $V$ and extremal conditional independence information among all variables in~$\Y$.} 
    \Ensure{Learned skeleton $C$ and collection of separation sets $\mathcal{S}$.}

  \State Initialize the complete undirected graph $C \gets (V, E)$ and collection of separating sets $\mathcal{S} \gets \varnothing$.

  \For{each size of separating set $\ell \in \{1, \dots, d\}$}
  \For{each edge $(i, j) \in \E$}\label{alg:pcalg-line-2}
  \If {the edge cannot be removed, i.e., $|\mathrm{adj}_{C}(i)\setminus \{j\}|< \ell$,}
    \State \textbf{continue} the \textbf{for loop} at line~\ref{alg:pcalg-line-2}
  \EndIf
  
  \For{each separating set $S \subseteq \mathrm{adj}_C(i) \setminus \{j\}$ with $|S| = \ell$}\label{alg:pcalg-line-3}

  \If{$Y_i \perp_e Y_j \mid \Y_S$,}\label{alg:ci-statement-pc}
    \State Update the edge set $\E \gets \E \setminus \{(i, j)\}$ and undirected graph $C \gets (V, E)$.
    \State Append the separating set and edge pair $\mathcal{S} \gets \mathrm{append}([S, (i, j)], \mathcal{S})$.
    \State \textbf{continue} the \textbf{for loop} at line~\ref{alg:pcalg-line-2}
  \EndIf

\EndFor
\EndFor
\EndFor
\end{algorithmic}
\end{algorithm}
Building on results from the non-extremal setting, we can show that the extremal PC-algorithm correctly recovers the extremal CPDAG from a faithful directed extremal graphical model.
\begin{corollary}\label{cor:cpdag-recovery}
  Let $\Y$ be a faithful directed extremal graphical model on the rooted DAG $G_e = (V, \E_e)$.
  If the input for the extremal PC-algorithm is the oracle of true extremal conditional independence relations for $\Y$, then it outputs the unique CPDAG that represents the extremal Markov equivalence class of $G_e$. 
\end{corollary}
The proof of Corollary~\ref{cor:cpdag-recovery} rests upon the  assumption of extremal faithfulness and follows the same arguments as Theorem 5.1 in \citet{spirtes2000causation} for the skeleton and \citet{meek1995} for the orientation. Similar arguments have been used in the high-dimensional non-extremal setting by \citet{kalisch2007estimating}. 
In applications, we do not have perfect knowledge of extremal conditional independencies among the variables $\Y$. For $n$ independent copies of $\Y$ we obtain a sample version of Algorithm~\ref{alg:ext-pc} by replacing the conditional independence statement on Line~\ref{alg:ci-statement-pc} with $\sqrt{n - |S| - 3}\ |Z_{ij\mid S}| \leq \Phi^{-1}(1 - \alpha / 2)$, as discussed in Section~\ref{sec:ex_cond_ind_test}.

In a numerical experiment, we evaluate the performance of the sample version of the extremal PC-algorithm, where we set the significance level $\alpha = 0.01$ to the default value for the classical PC-algorithm (see also the discussion in \citealt[][Section 4.2]{kalisch2007estimating}).
We consider the same three types of data generating processes as in the simulation of Section~\ref{sec:ex_cond_ind_test}, where  (a) corresponds to multivariate Pareto, (b) to max-stable, and (c) to extremal-SCM. 
For each combination of dimension $d \in \{5, 10, 15\}$ and expected neighborhood size $E_N \in \{2, 3.5, 5\}$, we generate 20 independent random rooted DAGs. For each DAG we randomly prune up to $d$ edges from $G$ such that the resulting extremal DAG $G_e$ remains rooted. Given $G$ and $G_e$, we generate $n=10000$ samples from the three models as in Section~\ref{sec:ex_cond_ind_test} and use different empirical quantile levels $\tau \in \{0.5, 0.7, 0.9, \dots, 0.995\}$ resulting in effective sample sizes $m = \lfloor n(1-\tau)\rfloor$; for the multivariate Pareto distribution in (a) we only generate $m$ samples. The results are shown in Figure~\ref{fig:pc-simul} in terms of the structural Hamming distance (SHD) (defined in Section~\ref{app:d-separation} of the Supplementary Material) between estimated and true graphs \citep{tsamardinos2006max}.
In all combinations of dimension and neighbor size, the data from the multivariate Pareto distribution does not have any pre-asymptotic bias and therefore lower thresholds are always better. For the data (b) and (c) in the domain of attraction we see that there is a bias-variance trade-off, and the threshold has to be sufficiently high to make sure the models are close enough to the limit~$\Y$. 
As expected, the learning difficulty scales with the graph dimensions $d$  and the expected neighbor size $E_N$.
The range of the SHD is driven by $d$ while $E_N$ controls the minimal SHD achieved by the data (b) and (c) in the domain of attraction.
For small to moderate $d$ and $E_N$, on  the data from (a), the extremal PC-algorithm maintains relatively stable performance across threshold levels, only degrading at extreme thresholds. 
The highest SHD occurs on data from (c) at low thresholds, suggesting the presence of additional pre-asymptotic bias in this setting compared to max-stable data.

\begin{figure}[tb]
  \centering
  \includegraphics[scale=0.65]{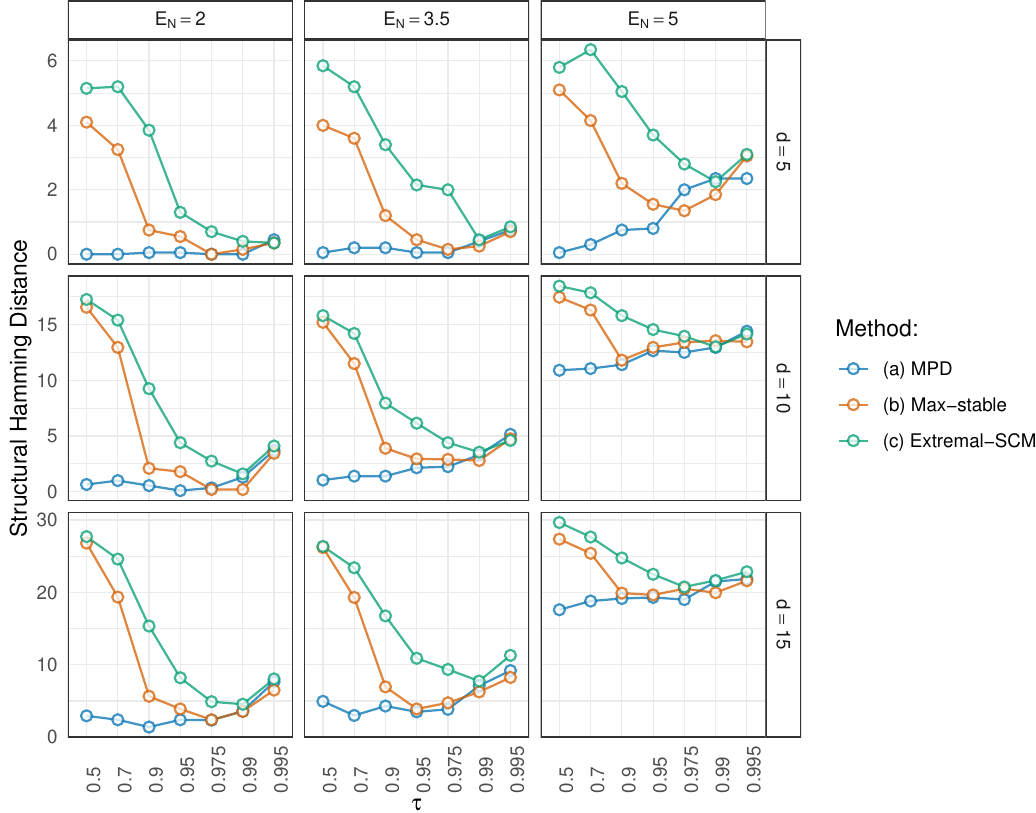}
  \caption{Performance of the extremal PC-algorithm (Algorithm~\ref{alg:ext-pc}) in terms of the average structural Hamming distance over 20 repetitions for different dimensions $d \in \{5, 10, 15\}$, expected neighborhood sizes $E_N \in \{2, 3.5, 5\}$ and empirical quantile levels $\tau \in \{0.5, 0.7, 0.9, \dots, 0.995\}$.
  }
  \label{fig:pc-simul}
\end{figure}

\subsection{Extremal pruning}\label{sec:ext-pruning}
The idea of extremal pruning is that the causal structure of the whole distribution might differ from the one in the tail. 
Since the non-extremal DAG $G$ corresponding to the whole distribution is always a supergraph of the extremal DAG $G_e$, we propose an algorithm that uses the extremal conditional independence information in $\Y$ to prune the edges of $G$ to get $G_e$.
The strength of this approach is that we can use classical structure learning algorithm from causal discovery to obtain an estimate of the graph $G$. Learning of the extremal graph $G_e$ is then reduced to the much simpler task of removing suitable edges from this given graph. In this way, we can use the whole sample size $n$ of some data set for the hard problem, and the extreme events (exceedances) with effective sample size $m<n$ for the simpler sub-task. 

The extremal pruning algorithm, detailed in Algorithm~\ref{alg:ext_pruning_v2},  begins with a rooted DAG $G = (V, \E)$ and extremal conditional independence information for the variables in $\Y$. It then iterates over the set of prunable edges $P$, which are edges that can be removed from $G$ while maintaining the rootedness of the resulting DAG. For each prunable edge $(i, j)$, the algorithm temporarily removes the edge and examines whether there exists a separating set $S$ such that $Y_i$ and $Y_j$ are extremally conditionally independent given $\Y_S$. If such a separating set is found, the edge is permanently removed from the graph. This process is repeated until no more edges can be pruned, yielding the final pruned DAG $G^*$.
The extremal pruning Algorithm~\ref{alg:ext_pruning_v2}, by construction, always returns a rooted DAG $G^*$, since it keeps track of the set of edges which can be pruned at each iteration (see Line~\ref{alg:line-3}). 
Moreover, if the conditional independence information corresponds to the extremal DAG $G_e$ of the input DAG $G$, then Algorithm~\ref{alg:ext_pruning_v2} returns $G_e$. 

\begin{proposition}\label{prop:algo-oracle}
  Let $\Y$ be a faithful directed extremal graphical model on the rooted DAG $G_e = (V, \E_e)$. For any rooted DAG $G = (V, \E)$ with $\E_e \subseteq \E$, Algorithm~\ref{alg:ext_pruning_v2} with input $G$ and the extremal conditional independence information of $\Y$ returns~$G_e$.

\end{proposition}

Similarly as for the PC-algorithm, we can obtain a sample version of the extremal pruning algorithm, and show that it is consistent.

\begin{algorithm}[tb!]
  \caption{Extremal pruning (oracle version)}
  \label{alg:ext_pruning_v2}

  \begin{algorithmic}[1]
    \Require{A rooted DAG $G = (V, \E)$ and extremal conditional independence information among all variables in $\Y$.} 
    \Ensure{A pruned rooted DAG $\Gstar = (V, \Estar)$.}

  \State Initialize the set of prunable edges as $P \gets \{(i, j) \in \E \colon |pa_{G}(j)| \geq 2\}$.

  \State Initialize the edge set $\Estar \gets \E$ and DAG $\Gstar \gets (V, \Estar)$.

\For{each prunable edge $(i, j) \in P$}\label{alg:line-2}
  \If {the edge cannot be pruned, i.e., $|pa_{\Gstar}(j)|< 2$,} \label{alg:line-3}
    \State \textbf{continue} the \textbf{for loop} at line~\ref{alg:line-2}
  \EndIf

  \State Define the DAG $\Gstar_{i, j} \gets (V, \Estar \setminus \{(i, j)\})$ where we drop directed edge $(i, j)$.
  \label{alg:pruned-dag}

  \State Compute the collection of separating sets as $\mathcal{S}_{i, j} = \{S \subseteq V \setminus\{i, j\} \colon i \indep_{\Gstar_{i, j}} j \mid S\}$. 
  \label{alg:bottleneck}

  \If{there exists a separating set $S \in \mathcal{S}_{i, j}$ such that $Y_i \perp_e Y_j \mid \Y_S$,}\label{alg:ci-statement}
    \State Update $\Gstar \gets \Gstar_{i, j}$.
  \EndIf
\EndFor
\end{algorithmic}
\end{algorithm}

\begin{proposition}\label{prop:algo-consistency}
  Denote by $\hat{G}_e(\alpha_n)$ the estimate from the sample version of Algorithm~\ref{alg:ext_pruning_v2}, which uses $\sqrt{n - |S| - 3}\ |Z_{ij\mid S}| \leq \Phi^{-1}(1 - \alpha_n / 2)$ in place of the oracle independence test $Y_i \perpe Y_j \mid \Y_S$.
  Let $G_e$ denote the true extremal DAG. 
  Assume that for all $i, j \in V$ and $S \subseteq V \setminus \{i, j\}$ the estimator $\hat{\rho}_{ij \mid S}$ of the partial correlation converges in probability to $\rho_{ij \mid S}$.
  Then, as $n \to \infty$, there exists a sequence of significance levels as  $\alpha_n \to 0$ such that
  \begin{align*}
    \PP\left( \hat{G}_e(\alpha_n) \neq G_e \right) \to 0.
  \end{align*}
\end{proposition}

Algorithm~\ref{alg:ext_pruning_v2} has a computational bottleneck when it comes to computing the collection of separating sets $\mathcal{S}_{i, j}$ for the pair $\{i, j\}$ (see Line~\ref{alg:bottleneck}). This computation requires evaluating whether each of the $2^{d -2}$ subsets $S \subseteq V \setminus\{i, j\}$ separates $i$ from $j$ in the pruned graph $\Gstar_{i, j}$.
To speed up the computation, in place of $\mathcal{S}_{i, j}$, we consider the single separating set 
$S^* \coloneqq S^{\mathrm{MB}, \Gstar_{i, j}}_j \cup \{v^*\}$, where $S^{\mathrm{MB}, \Gstar_{i, j}}_j$ denotes the Markov blanket of $j$ in $\Gstar_{i, j}$ (see Section~\ref{app:d-separation} in the Supplementary Material for a definition) and 
$v^*$ denotes the root node of $\Gstar_{i, j}$.
The idea is that the union of the Markov blanket of $j$ with the root node (which is an ancestor of all nodes in the graph) is a set that guarantees d-separation while being `local' with respect to $j$. 

In the following numerical experiment, we evaluate the performance of the sample version of the extremal pruning algorithm where we set the significance level to $\alpha = 0.01$, and consider the same setup as in Section~\ref{sec:ext-pc}.
Figure~\ref{fig:sim-d-aggregated} shows the average SHD over 20 independent draws of random DAGs. The threshold $\tau$ affects the extremal pruning algorithm via the power and level of the underlying  extremal conditional independence test introduced in Section~\ref{sec:ex_cond_ind_test}. As shown in Figure~\ref{fig:ci-test}, at low $\tau$, the errors entail under-pruning: p-values under the null hypothesis can be too low causing excess retention of edges. At high $\tau$, the test controls Type I error well, but the smaller sample size reduces the test power leading to over-pruning errors. 
For data from (b) and (c) (domain of attraction cases), we observe both under and over-pruning errors across $\tau$, highlighting a tradeoff between the pre-asymptotic bias and the variance.
In contrast, for (a), errors predominantly appear at high $\tau$ where the reduced sample size reduces the power of the underlying test.
Pruning accuracy degrades, unsurprisingly, with increasing dimensionality $d$, leading to higher SHD across all settings. While $d$ controls the overall magnitude of SHD, the graph density $E_N$ affects the minimum SHD achievable in the domain of attraction data (b) and (c). 
For (a), when sufficient data is available, i.e., low $\tau$, pruning is nearly perfect across all $d$ and $E_N$, with errors only at high $\tau$ due to increased variance. For (b) and (c), there is an optimal determined by a bias-variance tradeoff. 
\begin{figure}[tb!]
	\centering
	\includegraphics[scale=.65]{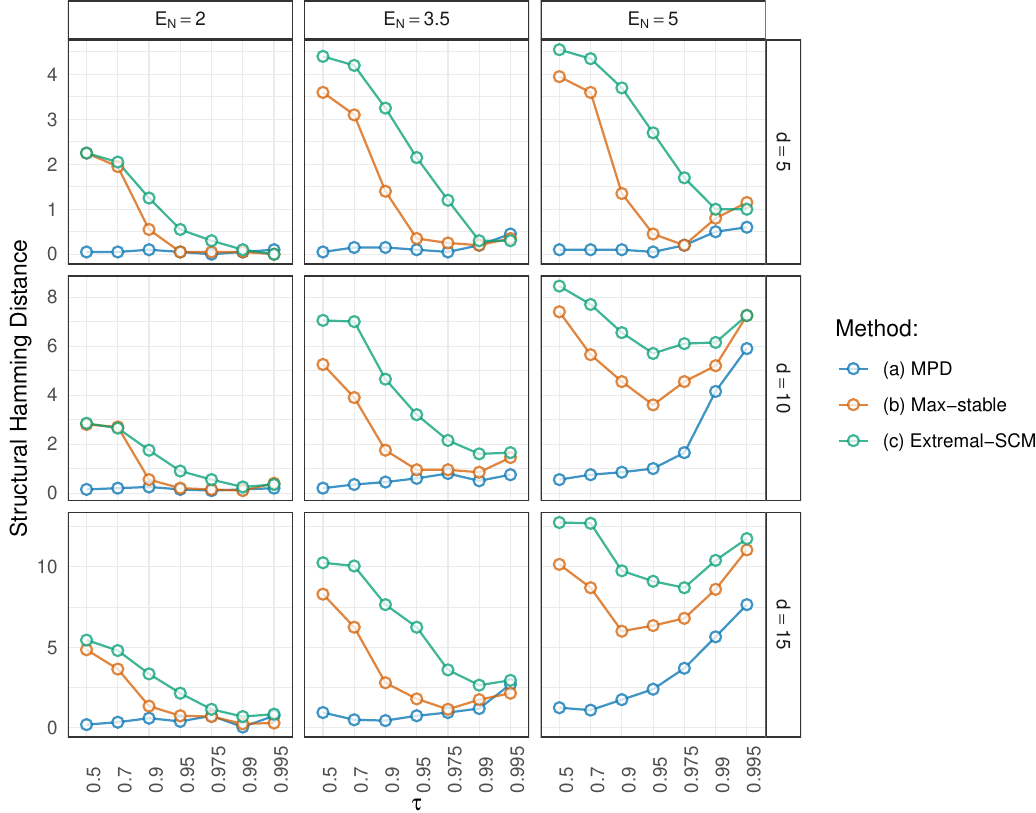}
	\caption{Performance of the extremal pruning algorithm (Algorithm~\ref{alg:ext_pruning_v2}) 
		in terms of the average structural Hamming distance over 20 repetitions for different dimensions $d \in \{5, 10, 15\}$, expected neighborhood sizes $E_N \in \{2, 3.5, 5\}$ and empirical quantile levels $\tau \in \{0.5, 0.7, 0.9, \dots, 0.995\}$. 
	}
	\label{fig:sim-d-aggregated}
\end{figure}
For small to moderate density of the graph, i.e., $E_N < 5$, the performance of extremal pruning under (b) and (c) approaches that under (a).

\section{Application}\label{sec:application}

We apply the extremal pruning algorithm to a river network dataset originally introduced by \citet{asadi2015extremes} and made publicly available by the Bavarian Environmental Agency (\url{https://www.gkd.bayern.de}). The dataset comprises average daily discharge measurements collected across the upper Danube basin between 1960 and 2009. Following the preprocessing pipeline established by \cite{GMPE2019}, we retain $n = 4,600$ observations for the summer months June, July, and August.

This dataset presents an ideal testbed for our method due to its well-defined causal structure, which is rare in real-world data. Specifically, we consider
six different branches of the river basin with up to 12 gauging stations. 
We assume that the ground truth of causal connections is given by the natural flow of water through these branches. Each branch can be represented as a rooted DAG $G_e$ satisfying the theoretical assumptions supporting our methods. Our goal is to assess the ability of the extremal pruning algorithm to recover a sparse, interpretable causal structure by progressively pruning edges from an initial larger graph $G$. This initial, non-extremal graph can be obtained by any structure learning method from the causal discovery literature using all observations. As we focus on learning causal structure in the extremes, we assume that the causal order is known and provided by a previous analysis. We assume the fully connected graph based on this causal order as the initial DAG $G = (V, E)$ as input for extremal pruning.
For instance, if a river branch consists of 12 stations with a known causal order $\pi = (12, 11, \dots , 1)$, then $(i, j) \in E$ if $\pi(i) < \pi(j)$. The graphical representation of the six river branches is shown in Figure~\ref{fig:river-maps} in the Supplementary Material.

We apply our extremal pruning algorithm presented in Algorithm~\ref{alg:ext_pruning_v2} with the computational speed-up discussed in Section~\ref{sec:ext-pruning} and the extremal conditional independence test introduced in Section~\ref{sec:ex_cond_ind_test}.
To compare the estimated extremal graph to estimates of non-extremal graphs, we adapt Algorithm~\ref{alg:ext_pruning_v2} by replacing the test in Line~\ref{alg:bottleneck} by different conditional independence tests that consider the whole distribution rather than only the extremes. In this way, we make sure that we can indeed attribute differences in the estimated causal structure to whether the focus is on extremes or not, rather than differences in the algorithm. 

The first conditional independence test is a kernel-based joint independence test using the Hilbert--Schmidt criterion, denoted as dHSIC \citep{pfister2018kernel}. The second is the projected covariance measure (PCM) \citep{lundborg2024}. While our proposed extremal conditional independence test can be applied directly to the extremal pruning algorithm, adapting these alternative methods requires additional preprocessing. 
By following a similar approach as in \citet{peters2014causal}, given a pair of nodes $(i, j) \in E$ and a subset
$S \subseteq V \setminus \{i, j\}$, we first estimate the conditional expectation $\mathbb{E}[Y_j \mid \Y_S]$ by regressing $Y_j$ onto $\Y_S$. We then define residuals $r = Y_j - \mathbb{E}[Y_j \mid \Y_S]$ and we test (i)  $r \indep (\Y_S, Y_i)$ with dHSIC and (ii) $\mathbb{E}[r \mid \Y_S, Y_i] \equiv 0$ with PCM.
All methods are evaluated at a significance level of $\alpha = 0.05$. For the regression step, we use random forests \citep{breiman2001random} with $B = 100$ trees and by setting the parameter \texttt{mtry}$=d$,
as implemented in \citet{wright2017ranger}. The HSIC test is conducted using a bootstrap procedure with $B' = 100$ repetitions and a Gaussian kernel, with the median heuristic for bandwidth selection \citep{pfister2017dhsic}.

The evaluation proceeds as follows. Each river branch is treated as an independent dataset, denoted \texttt{dataset-i} for $i = 1, \dots, 6$. We conduct $n_{\text{rep}} = 50$ repeated experiments, where in each run, we randomly subsample 25\% of the available observations ($n_{\text{sub}} = 1,150$). Given a fully connected DAG $G_i$ 
derived from the ground-truth order $\pi_i$, we apply extremal pruning with our proposed method using quantile thresholds $\tau \in \{0.9, 0.95, 0.975\}$, yielding estimated DAGs $\hat{G}_{\text{EP}, \tau}$. We also run extremal pruning using dHSIC and PCM, producing estimated graphs $\hat{G}_{\text{HSIC}}$ and $\hat{G}_{\text{PCM}}$. Finally, we compute the structural Hamming distance between each estimated DAG and the ground-truth graph given by the flow connections.

Figure~\ref{fig:river-results} presents the boxplots of the structural Hamming distance across 50 experimental repetitions for each of the six datasets; see also Figure~\ref{fig:river-maps} in the Supplementary Material for structure estimates of the different models for a fixed repetition.  The results demonstrate that employing extremal conditional independence tests improves the recovery of the river network compared to alternative methods, with relatively stable performance across different threshold levels $\tau$. The most significant advantage is observed relative to dHSIC, which tends to underperform in this setting. This can be attributed to the stringent requirement that dHSIC prunes an edge $(i, j)$ only if the residuals from regression are fully independent of the random vector $(Y_i, Y_S)$. This condition is challenging to satisfy, as nearby stations often exhibit dependence due to hidden confounders such as regional rainfall events.
\begin{figure}[tb]
	\centering
	\includegraphics[width=\textwidth]{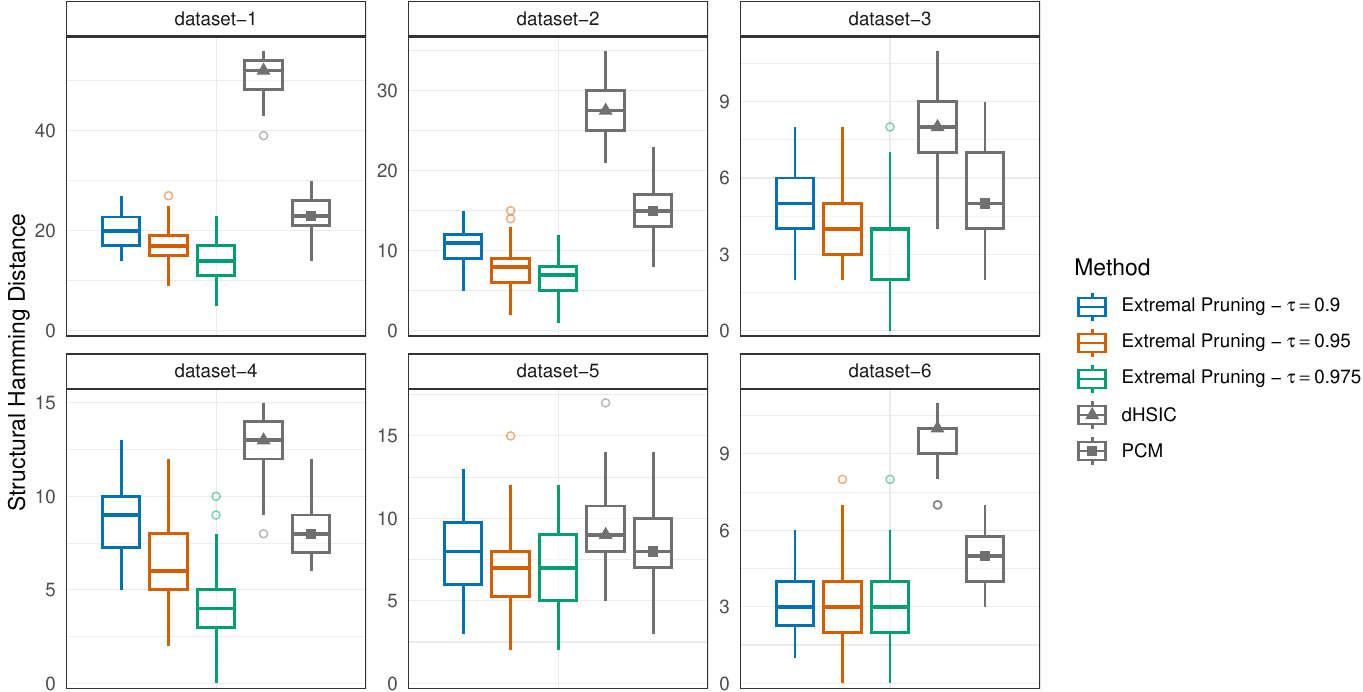}
	\caption{Structural Hamming distances over $n_{\text{reps}}=50$ subsample repetitions and six datasets for the three methods extremal pruning, HSIC, and PCM. Each dataset represents a branch of the upper Danube basin river network. 
	}
	\label{fig:river-results}
\end{figure}
In contrast, PCM demonstrates an improvement over dHSIC by testing the milder condition of mean independence of the residuals with respect to $(Y_i, Y_S)$. While our extremal conditional independence test and PCM maintain a consistent performance across river branches with respect to each other, dHSIC exhibits greater variability. One possible explanation is that confounding effects are more pronounced in the bulk of the distribution, and that the underlying causal structure becomes more apparent in extreme scenarios. Additionally, due to the extreme nature of river discharge data, nonparametric methods such as random forests may struggle to learn dependencies in the tails, leaving residual dependence unaccounted for.

\section*{Acknowledgments}
Sebastian Engelke and Frank R\"ottger were supported by the Swiss National Science Foundation (Grant 186858). 
Nicola Gnecco was supported by the Swiss National Science Foundation (Grant 210976). 
The authors thank Leonard Henckel and Jonas Peters for helpful discussions.

\bibliographystyle{chicago-custom}
\bibliography{bibliography}

\begin{thebibliography}{}

\bibitem[\protect\citeauthoryear{Am\'{e}ndola, Kl\"{u}ppelberg, Lauritzen, and
  Tran}{Am\'{e}ndola et~al.}{2022}]{amendola2022conditional}
Am\'{e}ndola, C., C.~Kl\"{u}ppelberg, S.~Lauritzen, and N.~M. Tran (2022).
\newblock Conditional independence in max-linear {B}ayesian networks.
\newblock {\em The Annals of Applied Probability\/}~{\em 32\/}(1), 1--45.

\bibitem[\protect\citeauthoryear{Asadi, Davison, and Engelke}{Asadi
  et~al.}{2015}]{asadi2015extremes}
Asadi, P., A.~C. Davison, and S.~Engelke (2015).
\newblock Extremes on river networks.
\newblock {\em The Annals of Applied Statistics\/}~{\em 9\/}(4), 2023--2050.

\bibitem[\protect\citeauthoryear{Asenova, Mazo, and Segers}{Asenova
  et~al.}{2021}]{asenova2021inference}
Asenova, S., G.~Mazo, and J.~Segers (2021).
\newblock Inference on extremal dependence in the domain of attraction of a
  structured {Hüsler--Reiss} distribution motivated by a {Markov} tree with
  latent variables.
\newblock {\em Extremes\/}~{\em 24}, 461--500.

\bibitem[\protect\citeauthoryear{Asenova and Segers}{Asenova and
  Segers}{2023}]{asenova2023extremes}
Asenova, S. and J.~Segers (2023).
\newblock Extremes of {M}arkov random fields on block graphs: max-stable limits
  and structured {H}\"{u}sler-{R}eiss distributions.
\newblock {\em Extremes\/}~{\em 26\/}(3), 433--468.

\bibitem[\protect\citeauthoryear{Bodik}{Bodik}{2024}]{bodik2024extremetreatmenteffectextrapolating}
Bodik, J. (2024).
\newblock Extreme treatment effect: Extrapolating dose-response function into
  extreme treatment domain.
\newblock URL \url{https://arxiv.org/abs/2403.11003}.

\bibitem[\protect\citeauthoryear{Bongers, Forr{\'e}, Peters, and Mooij}{Bongers
  et~al.}{2021}]{bon2021}
Bongers, S., P.~Forr{\'e}, J.~Peters, and J.~M. Mooij (2021).
\newblock {Foundations of structural causal models with cycles and latent
  variables}.
\newblock {\em The Annals of Statistics\/}~{\em 49\/}(5), 2885 -- 2915.

\bibitem[\protect\citeauthoryear{Breiman}{Breiman}{2001}]{breiman2001random}
Breiman, L. (2001).
\newblock Random forests.
\newblock {\em Machine learning\/}~{\em 45}, 5--32.

\bibitem[\protect\citeauthoryear{Buriticá and Engelke}{Buriticá and
  Engelke}{2024}]{buritica024progressionextrapolationprincipleregression}
Buriticá, G. and S.~Engelke (2024).
\newblock Progression: an extrapolation principle for regression.
\newblock URL \url{https://arxiv.org/abs/2410.23246}.

\bibitem[\protect\citeauthoryear{Chavez-Demoulin and Mhalla}{Chavez-Demoulin
  and Mhalla}{2024}]{chavezdemoulin2024causalityextremes}
Chavez-Demoulin, V. and L.~Mhalla (2024).
\newblock Causality and extremes.
\newblock URL \url{https://arxiv.org/abs/2403.05331}.

\bibitem[\protect\citeauthoryear{Chernozhukov}{Chernozhukov}{2005}]{chernozhukov2005}
Chernozhukov, V. (2005).
\newblock {Extremal quantile regression}.
\newblock {\em {Annals of Statistics}\/}~{\em 33\/}(2), 806 -- 839.

\bibitem[\protect\citeauthoryear{Coles, Heffernan, and Tawn}{Coles
  et~al.}{1999}]{col1999}
Coles, S., J.~Heffernan, and J.~Tawn (1999).
\newblock Dependence measures for extreme value analyses.
\newblock {\em Extremes\/}~{\em 2}, 339--365.

\bibitem[\protect\citeauthoryear{Coles and Tawn}{Coles and Tawn}{1991}]{CT1991}
Coles, S.~G. and J.~A. Tawn (1991).
\newblock Modelling extreme multivariate events.
\newblock {\em Journal of the Royal Statistical Society. Series B.
  Methodological\/}~{\em 53\/}(2), 377--392.

\bibitem[\protect\citeauthoryear{Corradini and Strokorb}{Corradini and
  Strokorb}{2024}]{corradini2024stochastic}
Corradini, M. and K.~Strokorb (2024).
\newblock Stochastic ordering in multivariate extremes.
\newblock {\em Extremes\/}~{\em 27}, 357--396.

\bibitem[\protect\citeauthoryear{Deidda, Engelke, and De~Michele}{Deidda
  et~al.}{2023}]{dei2023}
Deidda, C., S.~Engelke, and C.~De~Michele (2023).
\newblock Asymmetric dependence in hydrological extremes.
\newblock {\em Water Resources Research\/}~{\em 59\/}(12).

\bibitem[\protect\citeauthoryear{Deuber, Li, Engelke, and Maathuis}{Deuber
  et~al.}{2024}]{deu2021}
Deuber, D., J.~Li, S.~Engelke, and M.~H. Maathuis (2024).
\newblock Estimation and inference of extremal quantile treatment effects for
  heavy-tailed distributions.
\newblock {\em Journal of the American Statistical Association\/}~{\em
  119\/}(547), 2206--2216.

\bibitem[\protect\citeauthoryear{Devriendt}{Devriendt}{2022}]{devriendt2022a}
Devriendt, K. (2022).
\newblock {\em Graph geometry from effective resistances}.
\newblock Ph.\ D. thesis, University of Oxford.

\bibitem[\protect\citeauthoryear{Engelke, Hentschel, Lalancette, and
  Röttger}{Engelke
  et~al.}{2024}]{engelke2024graphicalmodelsmultivariateextremes}
Engelke, S., M.~Hentschel, M.~Lalancette, and F.~Röttger (2024).
\newblock Graphical models for multivariate extremes.
\newblock URL \url{https://arxiv.org/abs/2402.02187}.

\bibitem[\protect\citeauthoryear{Engelke and Hitz}{Engelke and
  Hitz}{2020}]{EH2020}
Engelke, S. and A.~S. Hitz (2020).
\newblock Graphical models for extremes.
\newblock {\em Journal of the Royal Statistical Society. Series B. Statistical
  Methodology\/}~{\em 82\/}(4), 871--932.
\newblock With discussions.

\bibitem[\protect\citeauthoryear{Engelke and Ivanovs}{Engelke and
  Ivanovs}{2021}]{engelke2021a}
Engelke, S. and J.~Ivanovs (2021).
\newblock Sparse structures for multivariate extremes.
\newblock {\em Annual Review of Statistics and Its Application\/}~{\em 8},
  241--270.

\bibitem[\protect\citeauthoryear{Engelke, Ivanovs, and Strokorb}{Engelke
  et~al.}{2024}]{eng_iva_kir}
Engelke, S., J.~Ivanovs, and K.~Strokorb (2024).
\newblock Graphical models for infinite measures with applications to extremes.
\newblock URL \url{https://arxiv.org/abs/2211.15769}.

\bibitem[\protect\citeauthoryear{Engelke, Ivanovs, and Thøstesen}{Engelke
  et~al.}{2024}]{engelke2024levygraphicalmodels}
Engelke, S., J.~Ivanovs, and J.~D. Thøstesen (2024).
\newblock L\'evy graphical models.
\newblock URL \url{https://arxiv.org/abs/2410.19952}.

\bibitem[\protect\citeauthoryear{Engelke, Lalancette, and Volgushev}{Engelke
  et~al.}{2024}]{engelke2022a}
Engelke, S., M.~Lalancette, and S.~Volgushev (2024).
\newblock Learning extremal graphical structures in high dimensions.
\newblock URL \url{https://arxiv.org/abs/2111.00840}.

\bibitem[\protect\citeauthoryear{Engelke, Malinowski, Oesting, and
  Schlather}{Engelke et~al.}{2014}]{eng2014}
Engelke, S., A.~Malinowski, M.~Oesting, and M.~Schlather (2014).
\newblock Statistical inference for max-stable processes by conditioning on
  extreme events.
\newblock {\em Advances in Applied Probability\/}~{\em 46\/}(2), 478--495.

\bibitem[\protect\citeauthoryear{Engelke and Taeb}{Engelke and
  Taeb}{2024}]{engelke2024extremalgraphicalmodelinglatent}
Engelke, S. and A.~Taeb (2024).
\newblock Extremal graphical modeling with latent variables via convex
  optimization.
\newblock URL \url{https://arxiv.org/abs/2403.09604}.

\bibitem[\protect\citeauthoryear{Engelke and Volgushev}{Engelke and
  Volgushev}{2022}]{eng_vol_2022}
Engelke, S. and S.~Volgushev (2022).
\newblock Structure learning for extremal tree models.
\newblock {\em Journal of the Royal Statistical Society. Series B. Statistical
  Methodology\/}~{\em 84\/}(5), 2055--2087.

\bibitem[\protect\citeauthoryear{Galles and Pearl}{Galles and
  Pearl}{1997}]{galles1997axioms}
Galles, D. and J.~Pearl (1997).
\newblock Axioms of causal relevance.
\newblock {\em Artificial Intelligence\/}~{\em 97\/}(1-2), 9--43.

\bibitem[\protect\citeauthoryear{Gissibl and Kl{\"u}ppelberg}{Gissibl and
  Kl{\"u}ppelberg}{2018}]{gissibl2018max}
Gissibl, N. and C.~Kl{\"u}ppelberg (2018).
\newblock Max-linear models on directed acyclic graphs.
\newblock {\em Bernoulli\/}~{\em 24\/}(4A), 2693--2720.

\bibitem[\protect\citeauthoryear{Gnecco, Meinshausen, Peters, and
  Engelke}{Gnecco et~al.}{2021}]{GMPE2019}
Gnecco, N., N.~Meinshausen, J.~Peters, and S.~Engelke (2021).
\newblock Causal discovery in heavy-tailed models.
\newblock {\em The Annals of Statistics\/}~{\em 49\/}(3), 1755--1778.

\bibitem[\protect\citeauthoryear{Haavelmo}{Haavelmo}{1944}]{Haavelmo1944}
Haavelmo, T. (1944).
\newblock The probability approach in econometrics.
\newblock {\em Econometrica\/}~{\em 12}, S1--S115 (supplement).

\bibitem[\protect\citeauthoryear{Heffernan and Tawn}{Heffernan and
  Tawn}{2004}]{heffernan2004}
Heffernan, J.~E. and J.~A. Tawn (2004).
\newblock A conditional approach for multivariate extreme values.
\newblock {\em Journal of the Royal Statistical Society. Series B. Statistical
  Methodology\/}~{\em 66\/}(3), 497--546.
\newblock With discussions and reply by the authors.

\bibitem[\protect\citeauthoryear{Hentschel, Engelke, and Segers}{Hentschel
  et~al.}{2024}]{HES2022}
Hentschel, M., S.~Engelke, and J.~Segers (2024).
\newblock Statistical inference for {H}\"usler--{R}eiss graphical models
  through matrix completions.
\newblock {\em Journal of the American Statistical Association\/}~{\em 0\/}(0),
  1--13.

\bibitem[\protect\citeauthoryear{Hu, Peng, and Segers}{Hu
  et~al.}{2024}]{hu2022}
Hu, S., Z.~Peng, and J.~Segers (2024).
\newblock Modeling multivariate extreme value distributions via {M}arkov trees.
\newblock {\em Scandinavian Journal of Statistics. Theory and
  Applications\/}~{\em 51\/}(2), 760--800.

\bibitem[\protect\citeauthoryear{H{\"u}sler and Reiss}{H{\"u}sler and
  Reiss}{1989}]{HR1989}
H{\"u}sler, J. and R.-D. Reiss (1989).
\newblock Maxima of normal random vectors: between independence and complete
  dependence.
\newblock {\em Statistics \& Probability Letters\/}~{\em 7\/}(4), 283--286.

\bibitem[\protect\citeauthoryear{Janssen and Segers}{Janssen and
  Segers}{2014}]{jan2014}
Janssen, A. and J.~Segers (2014).
\newblock Markov tail chains.
\newblock {\em Journal of Applied Probability\/}~{\em 51\/}(4), 1133--1153.

\bibitem[\protect\citeauthoryear{Kabluchko, Schlather, and de~Haan}{Kabluchko
  et~al.}{2009}]{KSdH2009}
Kabluchko, Z., M.~Schlather, and L.~de~Haan (2009).
\newblock {Stationary max-stable fields associated to negative definite
  functions}.
\newblock {\em The Annals of Probability\/}~{\em 37\/}(5), 2042 -- 2065.

\bibitem[\protect\citeauthoryear{Kalisch and B{\"u}hlmann}{Kalisch and
  B{\"u}hlmann}{2007}]{kalisch2007estimating}
Kalisch, M. and P.~B{\"u}hlmann (2007).
\newblock Estimating high-dimensional directed acyclic graphs with the
  pc-algorithm.
\newblock {\em Journal of Machine Learning Research\/}~{\em 8\/}(3).

\bibitem[\protect\citeauthoryear{Lauritzen}{Lauritzen}{1996}]{lau1996}
Lauritzen, S.~L. (1996).
\newblock {\em Graphical models}, Volume~17 of {\em Oxford Statistical Science
  Series}.
\newblock The Clarendon Press, Oxford University Press, New York.

\bibitem[\protect\citeauthoryear{Lauritzen, Dawid, Larsen, and
  Leimer}{Lauritzen et~al.}{1990}]{Lauritzen1990}
Lauritzen, S.~L., A.~P. Dawid, B.~N. Larsen, and H.~G. Leimer (1990).
\newblock Independence properties of directed {M}arkov fields.
\newblock {\em Networks\/}~{\em 20}, 491--505.

\bibitem[\protect\citeauthoryear{Leadbetter, Lindgren, and
  Rootz{\'e}n}{Leadbetter et~al.}{2012}]{leadbetter2012extremes}
Leadbetter, M.~R., G.~Lindgren, and H.~Rootz{\'e}n (2012).
\newblock {\em Extremes and related properties of random sequences and
  processes}.
\newblock Springer New York, NY.

\bibitem[\protect\citeauthoryear{Lederer and Oesting}{Lederer and
  Oesting}{2024}]{lederer2023extremes}
Lederer, J. and M.~Oesting (2024).
\newblock Extremes in high dimensions: Methods and scalable algorithms.
\newblock URL \url{https://arxiv.org/abs/2303.04258}.

\bibitem[\protect\citeauthoryear{Lundborg, Kim, Shah, and Samworth}{Lundborg
  et~al.}{2024}]{lundborg2024}
Lundborg, A.~R., I.~Kim, R.~D. Shah, and R.~J. Samworth (2024).
\newblock {The projected covariance measure for assumption-lean variable
  significance testing}.
\newblock {\em The Annals of Statistics\/}~{\em 52\/}(6), 2851 -- 2878.

\bibitem[\protect\citeauthoryear{Meek}{Meek}{1995}]{meek1995}
Meek, C. (1995).
\newblock Causal inference and causal explanation with background knowledge.
\newblock In {\em Proceedings of the Eleventh Conference on Uncertainty in
  Artificial Intelligence}, UAI'95, San Francisco, CA, USA, pp.\  403–410.
  Morgan Kaufmann Publishers Inc.

\bibitem[\protect\citeauthoryear{Mhalla, Chavez-Demoulin, and Dupuis}{Mhalla
  et~al.}{2020}]{mha2020}
Mhalla, L., V.~Chavez-Demoulin, and D.~J. Dupuis (2020, 05).
\newblock Causal mechanism of extreme river discharges in the upper danube
  basin network.
\newblock {\em Journal of the Royal Statistical Society Series C: Applied
  Statistics\/}~{\em 69\/}(4), 741--764.

\bibitem[\protect\citeauthoryear{Mhalla, Chavez-Demoulin, and Naveau}{Mhalla
  et~al.}{2024}]{mhalla2024causaldiscoverymultivariateextremes}
Mhalla, L., V.~Chavez-Demoulin, and P.~Naveau (2024).
\newblock Causal discovery in multivariate extremes with a hydrological
  analysis of {S}wiss river discharges.
\newblock URL \url{https://arxiv.org/abs/2405.10371}.

\bibitem[\protect\citeauthoryear{Pearl}{Pearl}{2009}]{Pearl2009}
Pearl, J. (2009).
\newblock {\em Causality: Models, Reasoning, and Inference\/} (2nd ed.).
\newblock New York, USA: Cambridge University Press.

\bibitem[\protect\citeauthoryear{Peters and B{\"u}hlmann}{Peters and
  B{\"u}hlmann}{2015}]{peters2015structural}
Peters, J. and P.~B{\"u}hlmann (2015).
\newblock Structural intervention distance for evaluating causal graphs.
\newblock {\em Neural computation\/}~{\em 27\/}(3), 771--799.

\bibitem[\protect\citeauthoryear{Peters, Janzing, and Sch\"olkopf}{Peters
  et~al.}{2017}]{PJS2017}
Peters, J., D.~Janzing, and B.~Sch\"olkopf (2017).
\newblock {\em Elements of Causal Inference: Foundations and Learning
  Algorithms}.
\newblock Cambridge, MA, USA: MIT Press.

\bibitem[\protect\citeauthoryear{Peters, Mooij, Janzing, and
  Sch{\"o}lkopf}{Peters et~al.}{2014}]{peters2014causal}
Peters, J., J.~M. Mooij, D.~Janzing, and B.~Sch{\"o}lkopf (2014).
\newblock Causal discovery with continuous additive noise models.
\newblock {\em The Journal of Machine Learning Research\/}~{\em 15\/}(1),
  2009--2053.

\bibitem[\protect\citeauthoryear{Pfister, B{\"u}hlmann, Sch{\"o}lkopf, and
  Peters}{Pfister et~al.}{2018}]{pfister2018kernel}
Pfister, N., P.~B{\"u}hlmann, B.~Sch{\"o}lkopf, and J.~Peters (2018).
\newblock Kernel-based tests for joint independence.
\newblock {\em Journal of the Royal Statistical Society Series B: Statistical
  Methodology\/}~{\em 80\/}(1), 5--31.

\bibitem[\protect\citeauthoryear{Pfister and Peters}{Pfister and
  Peters}{2017}]{pfister2017dhsic}
Pfister, N. and J.~Peters (2017).
\newblock dhsic: Independence testing via {H}ilbert {S}chmidt independence
  criterion.
\newblock {\em R Package version\/}~{\em 2}.

\bibitem[\protect\citeauthoryear{Resnick}{Resnick}{2008}]{res2008}
Resnick, S.~I. (2008).
\newblock {\em Extreme Values, Regular Variation and Point Processes}.
\newblock New York: Springer.

\bibitem[\protect\citeauthoryear{Rootz\'en and Tajvidi}{Rootz\'en and
  Tajvidi}{2006}]{rootzen2006}
Rootz\'en, H. and N.~Tajvidi (2006).
\newblock Multivariate generalized {P}areto distributions.
\newblock {\em Bernoulli\/}~{\em 12\/}(5), 917--930.

\bibitem[\protect\citeauthoryear{R\"{o}ttger, Engelke, and
  Zwiernik}{R\"{o}ttger et~al.}{2023}]{REZ2021}
R\"{o}ttger, F., S.~Engelke, and P.~Zwiernik (2023).
\newblock Total positivity in multivariate extremes.
\newblock {\em The Annals of Statistics\/}~{\em 51\/}(3), 962--1004.

\bibitem[\protect\citeauthoryear{Röttger, Coons, and Grosdos}{Röttger
  et~al.}{2023}]{roettger2023parametric}
Röttger, F., J.~I. Coons, and A.~Grosdos (2023).
\newblock Parametric and nonparametric symmetries in graphical models for
  extremes.
\newblock URL \url{https://arxiv.org/abs/2306.00703}.

\bibitem[\protect\citeauthoryear{Segers}{Segers}{2007}]{segers2007multivariate}
Segers, J. (2007).
\newblock Multivariate regular variation of heavy-tailed markov chains.
\newblock URL \url{https://arxiv.org/abs/math/0701411}.

\bibitem[\protect\citeauthoryear{Segers}{Segers}{2020}]{segers_2020}
Segers, J. (2020).
\newblock One- versus multi-component regular variation and extremes of
  {M}arkov trees.
\newblock {\em Advances in Applied Probability\/}~{\em 52\/}(3), 855–878.

\bibitem[\protect\citeauthoryear{Shen and Meinshausen}{Shen and
  Meinshausen}{2024}]{shen2024engressionextrapolationlensdistributional}
Shen, X. and N.~Meinshausen (2024).
\newblock Engression: Extrapolation through the lens of distributional
  regression.
\newblock URL \url{https://arxiv.org/abs/2307.00835}.

\bibitem[\protect\citeauthoryear{Shimizu, Hoyer, Hyv{\"a}rinen, and
  Kerminen}{Shimizu et~al.}{2006}]{shimizu2006linear}
Shimizu, S., P.~O. Hoyer, A.~Hyv{\"a}rinen, and A.~Kerminen (2006).
\newblock A linear non-{G}aussian acyclic model for causal discovery.
\newblock {\em Journal of Machine Learning Research\/}~{\em 7\/}(Oct),
  2003--2030.

\bibitem[\protect\citeauthoryear{Shimizu, Inazumi, Sogawa, Hyv{\"a}rinen,
  Kawahara, Washio, Hoyer, and Bollen}{Shimizu
  et~al.}{2011}]{shimizu2011directlingam}
Shimizu, S., T.~Inazumi, Y.~Sogawa, A.~Hyv{\"a}rinen, Y.~Kawahara, T.~Washio,
  P.~O. Hoyer, and K.~Bollen (2011).
\newblock Directlingam: A direct method for learning a linear non-{G}aussian
  structural equation model.
\newblock {\em Journal of Machine Learning Research\/}~{\em 12\/}(Apr),
  1225--1248.

\bibitem[\protect\citeauthoryear{Spirtes, Glymour, and Scheines}{Spirtes
  et~al.}{2000}]{spirtes2000causation}
Spirtes, P., C.~Glymour, and R.~Scheines (2000).
\newblock {\em Causation, prediction, and search\/} (Second ed.).
\newblock Adaptive Computation and Machine Learning. MIT Press, Cambridge, MA.
\newblock With additional material by David Heckerman, Christopher Meek,
  Gregory F. Cooper and Thomas Richardson, A Bradford Book.

\bibitem[\protect\citeauthoryear{Tsamardinos, Brown, and Aliferis}{Tsamardinos
  et~al.}{2006}]{tsamardinos2006max}
Tsamardinos, I., L.~E. Brown, and C.~F. Aliferis (2006).
\newblock The max-min hill-climbing {B}ayesian network structure learning
  algorithm.
\newblock {\em Machine learning\/}~{\em 65}, 31--78.

\bibitem[\protect\citeauthoryear{van~der Vaart}{van~der
  Vaart}{1998}]{vaart_1998}
van~der Vaart, A.~W. (1998).
\newblock {\em Asymptotic Statistics}.
\newblock Cambridge Series in Statistical and Probabilistic Mathematics.
  Cambridge University Press.

\bibitem[\protect\citeauthoryear{Velthoen, Cai, Jongbloed, and
  Schmeits}{Velthoen et~al.}{2019}]{Velthoenetal2019}
Velthoen, J., J.-J. Cai, G.~Jongbloed, and M.~Schmeits (2019).
\newblock Improving precipitation forecasts using extreme quantile regression.
\newblock {\em Extremes\/}~{\em 22\/}(4), 599--622.

\bibitem[\protect\citeauthoryear{Verma and Pearl}{Verma and
  Pearl}{1990}]{verma1990causal}
Verma, T. and J.~Pearl (1990).
\newblock Causal networks: Semantics and expressiveness.
\newblock In {\em Machine intelligence and pattern recognition}, Volume~9, pp.\
   69--76. Elsevier.

\bibitem[\protect\citeauthoryear{Wan and Zhou}{Wan and
  Zhou}{2025}]{wan2023graphical}
Wan, P. and C.~Zhou (2025).
\newblock Graphical lasso for extremes.
\newblock URL \url{https://arxiv.org/abs/2307.15004}.

\bibitem[\protect\citeauthoryear{Wang and Miao}{Wang and
  Miao}{2024}]{wang2024extremebasedcausaleffectlearning}
Wang, R. and W.~Miao (2024).
\newblock Extreme-based causal effect learning with endogenous exposures and a
  light-tailed error.
\newblock URL \url{https://arxiv.org/abs/2408.06211}.

\bibitem[\protect\citeauthoryear{Wright and Ziegler}{Wright and
  Ziegler}{2017}]{wright2017ranger}
Wright, M.~N. and A.~Ziegler (2017).
\newblock ranger: A fast implementation of random forests for high dimensional
  data in {C}++ and {R}.
\newblock {\em Journal of statistical software\/}~{\em 77}, 1--17.

\bibitem[\protect\citeauthoryear{Youngman}{Youngman}{2019}]{ExGAM2}
Youngman, B.~D. (2019).
\newblock {Generalized Additive Models for Exceedances of High Thresholds With
  an Application to Return Level Estimation for U.S. Wind Gusts}.
\newblock {\em {Journal of the American Statistical Association}\/}~{\em
  114\/}(528), 1865--1879.

\bibitem[\protect\citeauthoryear{Zhang et~al.}{Zhang
  et~al.}{2018}]{zhang2018extremal}
Zhang, Y. et~al. (2018).
\newblock Extremal quantile treatment effects.
\newblock {\em The Annals of Statistics\/}~{\em 46\/}(6B), 3707--3740.

\bibitem[\protect\citeauthoryear{Zhang, Bolin, Engelke, and Huser}{Zhang
  et~al.}{2023}]{zhang2023extremal}
Zhang, Z., D.~Bolin, S.~Engelke, and R.~Huser (2023).
\newblock Extremal dependence of moving average processes driven by
  exponential-tailed {L}\'evy noise.
\newblock URL \url{https://arxiv.org/abs/2307.15796}.

\end{thebibliography}



\newpage
\begin{appendix}

  \section{Background on directed graphical models}\label{app:d-separation}

  Let $ \G=(V,\E) $ be a directed acyclic graph (DAG) with edge weight matrix $ B $.
  We define $d$-separation as in \cite{PJS2017}.
  \begin{definition}[d-separation]\label{def:d-sep}
      A path between nodes $ i_1 $ and $ i_m $ is \textit{blocked by a set $ S $}, whenever there is a node $ i_k $ such that one of the following holds:
      \begin{enumerate}
          \item\, $ i_k\in S $ and
          \begin{align*}
              i_{k-1}             & \to i_k \to i_{k+1}               \\
              \text{or}~~i_{k-1}  & \leftarrow i_k \leftarrow i_{k+1} \\
              \text{or}~~ i_{k-1} & \leftarrow i_k \to i_{k+1}
            \end{align*}
          \item \, neither $ i_k $ or any of its descendants is in $ S $ and
          \begin{align*}
              i_{k-1}\to i_k \leftarrow i_{k+1}.
            \end{align*}
        \end{enumerate}
      Two disjoint sets $ A,B\subset V $ are \textit{d-separated} by a disjoint set $ S\subset V $ when each path between nodes in $ A $ and $ B $ is blocked by $ S $. This is denoted as
      \[A\indep_{G}B\mid S.\]
  \end{definition}
  A collection of three nodes $i,j,k$ that satisfies $ i\to j \leftarrow k $ is called a collider.
  If the nodes $i,k$ are not connected by an edge, a collider is called a v-structure.
  For some DAG $G$, the undirected graph obtained from ignoring directions in $G$ is denoted as the skeleton graph of $G$.
  The Markov blanket of some node $i\in V$ is the smallest set $M\subset V$ such that $i\indep_{G}V\setminus(\{i\}\cup M) \mid M$.
   
  Let $\X$ be a $d$-variate random vector and $\G=(V,\E)$ some DAG with vertex set $V=[d]$ and edge set $\E\subset V\times V$.
  We say that $\X$ satisfies the global Markov property with respect to $\G$ when
  \[A\indep_{G}B\mid S \Rightarrow \X_A\indep \X_B\mid \X_S\]
  for any disjoint subsets $A,B,S\subset V$.
  The global Markov property connects $d$-separation with conditional independence, and we call any random vector $\X$ that satisfies the global Markov property with respect to $\G$ a directed graphical model with respect to $\G$.
  
  We now define the structural Hamming distance as in~\citet{peters2015structural}.
  \begin{definition}[Structural Hamming distance]\label{def:shd}
    Given two DAGs $G = (V, E_G)$ and $H = (V, E_H)$, the structural Hamming distance is defined as:
    \begin{align*}
      \text{SHD}(G, H) = \# \{ (i, j) \in V \times V \mid G \text{ and } H \text{ do not have the same edge type between } i \text{ and } j \}.
    \end{align*}
    Here, an edge type refers to the presence, absence, or orientation of an edge between two nodes. In particular,
    \begin{itemize}
        \item If $(i, j) \in E_G$ and $(i, j) \notin E_H$, this counts as a discrepancy.
        \item If $(i, j) \in E_H$ but $(i, j) \notin E_G$, this counts as a discrepancy.
        \item If $(i, j) \in E_G$ and $(j, i) \in E_H$ (or vice versa), this counts as a discrepancy.
    \end{itemize}
    \end{definition}
  The SHD measures the number of edge modifications (addition, deletion, or reversal) required to transform $G$ into $H$.

  \section{Proofs}\label{sec:app_proofs}
  \subsection{Proof of Theorem \ref{thm:SCM_limit}}
  \begin{proof}
   The idea of the proof is inspired by the proof of \cite[Theorem 2.3]{segers2007multivariate} for times series.	
     Let us first establish the homogeneity of $\Psi_v$. For any $s\in \mathbb R$, this follows from the definition of $\Psi_v$ that 
     \begin{align*}
       \Psi_v(\x + s \einsfun, e) &= \lim_{t\to \infty} f^*_v(\x + (s+t)\einsfun, e) - (s+t) + s\\
       & = \lim_{u\to \infty} f^*_v(\x + u\einsfun, e) - u + s\\
       & = \Psi_v(\x , e) +s,
     \end{align*}
    where we replaced $s+t$ by $u$.
     
    Instead of statement~\eqref{MRV1}, we prove the equivalent statement	
    \begin{align}\label{MRV2}
      \lim_{t \to \infty} \mathbb P(X^*_1 - t > y, \X^* - X^*_1 \einsfun \in A \mid  X^*_1 > t ) = \exp(-y) \mathbb P(\W^{1} \in A),
    \end{align}
    for any Borel subset $A\subset \mathbb R^d$,
    where the only difference is that we now normalize with the first component of $\X^*$ instead of the deterministic threshold $t$. 
  
    Let $\pi:\{1,\dots, d\} \to \{1,\dots, d\}$ be a causal (or topological) ordering of the DAG $\G$, that is, for all $i \in \an(j)$ it holds that $\pi(i) < \pi(j)$, $i,j\in V$. Since the (only) root is node $X^*_1$, we have $\pi(1) = 1$.
    We prove~\eqref{eSCM} and~\eqref{MRV2} for the set of variables $(\pi(1), \dots, \pi(m))$, $m\in \{1,\dots ,d\}$, and then use induction over $m$. For $m=1$, the statement follows from the assumption that $X^*_1$ is normalized to exponential margins, and therefore $X_1^* - t \mid  X^*_1 > t$ is again an exponential  distribution because of the memorylessness property. The second part of the probability on the left-hand side of~\eqref{MRV2} becomes $X^*_{\pi(1)}- X^*_1 = 0 \in A$, which yields that $W^{1}_1 = 0$ as stated in~\eqref{eSCM}.
    
    Now suppose that the induction hypothesis holds for $m-1$, and recall that 
    \[ \pa(\pi(m)) \subseteq \{ \pi(1), \dots, \pi(m-1) \},\]
    by the definition of a causal ordering.	
    Using Portmanteau theorem, we can equivalently show that for any bounded and continuous function $g: \mathbb R^{m+1} \to \mathbb R$ we have 
    \begin{align}\label{claim1}
      \lim_{t \to \infty}  \mathbb E\left[ g\left(X^*_1 - t, X^*_{\pi(1)} - X^*_1, \dots, X^*_{\pi(m)} - X^*_1 \right) \middle| X^*_1 > t \right] = \mathbb E\left[ g\left( R, W^{1}_{\pi(1)}, \dots, W^{1}_{\pi(m)} \right) \right] 
    \end{align}
    For this we rewrite the left-hand side as
    \[   \mathbb E\left[ \tilde g_t \left(X^*_1 - t, X^*_{\pi(1)} - X^*_1, \dots, X^*_{\pi(m-1)} - X^*_1 \right) \middle| X^*_1 > t \right]  \]
    with a function $\tilde g_t: \mathbb R^{m} \to \mathbb R$ defined as
    \[ \tilde g_t(y, x_{\pi(1)}, \dots, x_{\pi(m-1)}) = \mathbb E_{\varepsilon_{\pi(m)}} \left[  g\left(y, x_{\pi(1)}, \dots, x_{\pi(m-1)}, f^*_{\pi(m)}( \x_{\text{pa}(\pi(m))} + (y + t)\mathbf{1}, \varepsilon_{\pi(m)}) - y - t \right) \right],  \]
    where the expectation is taken over the noise variable as indicated.
    Moreover, define 
    \[ \tilde g(y, x_{\pi(1)}, \dots, x_{\pi(m-1)}) = \mathbb E_{\varepsilon_{\pi(m)}} \left[  g\left(y,x_{\pi(1)}, \dots, x_{\pi(m-1)}, \Psi_{\pi(m)}( \x_{\text{pa}(\pi(m))}, \varepsilon_{\pi(m)})  \right) \right].  \]	
    With these definitions, claim~\eqref{claim1} is equivalent to showing that 
    \begin{align}\label{eq:weak}
      \lim_{t \to \infty}  \mathbb E\left[ \tilde g_t (X^*_1 - t, X^*_{\pi(1)} - X^*_1, \dots, X^*_{\pi(m-1)} - X^*_1 ) \mid X^*_1 > t \right] = \mathbb E\left[ \tilde g( R, W^{(1)}_{\pi(1)}, \dots, W^{(1)}_{\pi(m-1)} ) \right].  
    \end{align}
    Since $\tilde g_t$ and $\tilde g$ are bounded functions, it suffices to show that the random variables inside the expectation converge weakly. 
    To this end, note that for any sequences $y(t) \to y$ and $x_{\pi(i)}(t) \to x_{\pi(i)}$, $i = 1,\dots , m-1$, we have 
    \[ \lim_{t \to \infty} \tilde g_t(y, x_{\pi(1)}(t), \dots, x_{\pi(m-1)}(t)) = \tilde g(y, x_{\pi(1)}, \dots, x_{\pi(m-1)}),\]
    by the dominated convergence theorem and since 
    \[ \lim_{t \to \infty} f^*_{\pi(m)}( \x_{\pa(\pi(m))}(t) + y+ t, e) - y - t  = \Psi_{\pi(m)}( \x_{\pa(\pi(m))},e),\]
    where the latter follows from the definition of $\Psi_{\pi(m)}$ in~\eqref{limit_cond} and Assumption~\ref{ass_main}. We can now apply the extended continuous mapping theorem \citep[Theorem 18.11]{vaart_1998} to conclude weak convergence of the random variables inside the expectation in~\eqref{eq:weak}, where we use the induction assumption for $m-1$. This establishes~\eqref{MRV2} and~\eqref{eSCM}.
     
    From \citet[Theorem 2]{segers_2020} it follows that~\eqref{MRV2} implies~\eqref{MRV1} and, equivalently, the multivariate regular variation of $\X^*$. Moreover, we note that since $\Psi_v$ is a function with values in $\mathbb R$, we have that $W_v^{(1)} > -\infty$ almost surely for all $v \in V$. Therefore, \citet[Corollary 3]{segers_2020} implies that $\mathbb E \exp\{ W^{(1)}_v \} = 1$; for details see also Proposition 1 in the Supplementary Material of \cite{eng_vol_2022}.
  \end{proof}
  
  \subsection{Proof of Example \ref{ex:tail}}
  \begin{proof}
    We need to show that the first term has approximately an exponential distribution. For this we define the approximate maximum as 
    \[ \widetilde{\max}(x_2,x_3) = \frac12\left\{x_2 + x_3 + \sqrt{ (x_2 - x_3)^2 + 1/(1+x_2^2 + x_3^2)}\right\} \] 
    and note that 
    \[ 2\{\widetilde{\max}(x_2,x_3) - \max(x_2,x_3)\}  = \sqrt{ (x_2 - x_3)^2 + 1/(1+x_2^2 + x_3^2)} - \sqrt{ (x_2 - x_3)^2} \leq C |1+x_2^2 + x_3^2|^{-1/2},\]
    because of the H\"older continuity of the square root and where $C>0$ is a constant.
    Since $ \widetilde{\max}(x_2,x_3) \leq \max(x_2,x_3) + 2C$ for any $x_2,x_3 \in \mathbb R$, we obtain  
    \begin{align*}
      \mathbb P(\widetilde{\max}(X_2,X_3) > x ) &= \mathbb P(\widetilde{\max}(X_2,X_3) > x, \max(X_2,X_3) > x - 2C)\\
      & \leq \mathbb P(\max(X_2,X_3) + 2C| X_2^2 + X_3^2|^{-1/2}  > x, \max(X_2,X_3) > x - 2C) \\
      & \leq \mathbb P(\max(X_2,X_3) + 2C(x-2C)^{-1}  > x, \max(X_2,X_3) > x - 2C) \\
      & = \mathbb P(\max(X_2,X_3) > x - 2C(x-2C)^{-1}) \\
      & \leq \mathbb P(\max(X_2,X_3) + \eta > x) 
    \end{align*}
    where for a fixed but arbitrary $\eta > 0$ we have chosen $x>0$ large enough such that $2C(x-2C)^{-1} < \eta$. 
    It follows that 
    \begin{align*}
      0 \leq \lim_{q \to 1} F^{-1}_{\widetilde{\max}} (q) -  F^{-1}_{\max} (q) \leq \eta, 
    \end{align*}
    where $F^{-1}_{\max}$ and $ F^{-1}_{\widetilde{\max}}$ denote the the quantile functions of the random variables $\max(X_2,X_3)$ and $\widetilde{\max}(X_2,X_3)$, respectively. Note that the first inequality above follows trivially from the fact that  $ \max(x_2,x_3) \leq \widetilde{\max}(x_2,x_3)$. Since $\eta>0$ was arbitrary, we can replace it by $0$ in the above inequality. 
    
    We now consider the quantile function of the maximum.
    Note that $\max(X_2,X_3) =  \varepsilon_1 +  \max(\varepsilon_2, \varepsilon_3)$.
    Since $\max(\varepsilon_2, \varepsilon_3) \leq \max(0,\varepsilon_2) + \max(0,\varepsilon_3)$ we have that $\mathbb E \exp \{ (1+\delta) \max(\varepsilon_2, \varepsilon_3) \} < \infty$. Consequently, $\mathbb E \exp \{ \max(\varepsilon_2, \varepsilon_3)\} = K < \infty $ and since $\varepsilon_1$ is exponentially distributed, we have
    \begin{align*}
      \lim_{q \to 1} F^{-1}_{\max} (q) - G^{-1}(q) = \log K,
    \end{align*}
    by \citet[Lemma 2]{zhang2023extremal}. Putting things together, we find that 
    \begin{align*}
      \lim_{q \to 1} F^{-1}_{\widetilde{\max}} (q) - G^{-1}(q) = \log K,
    \end{align*}
    that is, $\widetilde{\max}(X_2,X_3)$ has an exponential tail up to a constant shift.
    This implies that in the limit, we can replace the quantile function $F^{-1}_{\widetilde{\max}}$ by $F_4^{-1}$, since the latter is defined as $X_4 = \widetilde{\max}(X_2,X_3) + \varepsilon_4$ and thus 
    \begin{align}\label{F4_tail}
      \lim_{q \to 1} G^{-1}(q) - F^{-1}_4(q) = 0,
    \end{align}
    with the same argument as in~\eqref{exp_tail}.
    Consequently, the quantile function of $X_4$ satisfies
    \[ \lim_{q\to 1} F_4^{-1}(q) - G^{-1}(q)  = \lim_{q\to 1} F_4^{-1}(q) - F^{-1}_{\widetilde{\max}} (q) + F^{-1}_{\widetilde{\max}} (q) - G^{-1}(q) = \log K, \]  
    and therefore
    \[ \lim_{u\to \infty} G^{-1}\circ F_4(u) - u  = -\log K. \]  
  \end{proof}

  \subsection{Proof of Example~\ref{ex:exp}}\label{app:ex-exp}
  \begin{proof}
    We first observe that under no intervention, i.e., the observational distribution, we have
    \begin{align*}
       \PP^{\intervene(\emptyset)}(Y_2\le x)&= \lim_{t \to \infty}\PP(X_
    2^*-t \leq x) = 1, \quad \forall x\in(-\infty,\infty),
    \end{align*}
    since $X_2^* = H^{-1}\circ F_2(X_2)$ is standard exponential. Therefore, $\PP_{Y_2}^{\intervene(\emptyset)}=\delta_{-\infty}$.
  
    For intervention on $Y_1$, we note that the survival function of $X_2$ satisfies
    \[\bar F_2(x) =  \PP(X_1 + \varepsilon_2 > x) = \PP(X_1 + \varepsilon_2 - \sigma^2/2 > x -\sigma^2/2) = \bar F_1(x - \sigma^2/2)\{1+ o(1)\}, \quad x\to\infty,\]
    since $\mathbb E \exp(\varepsilon_2 - \sigma^2/2) = 1$ and therefore by \citet[Lemma 2]{zhang2023extremal}, $X_1 + \varepsilon_2 - \sigma^2/2$ has approximately a standard exponential tail.
    Thus, we have 
    \[H^{-1}\circ F_2(x) = -\log\{\bar F_2(x)\}  = -\log[ \bar F_1(x - \sigma^2/2)\{1+ o(1)\} ]  = x - \sigma^2/2 + o(1),\quad x\to\infty,\] 
    and consequently
    \begin{align*}
        \PP_{Y_2}^{\intervene(Y_1:=\xi_1)}(A)&=\lim_{t \to \infty}\PP^{\intervene(X_1:=t+\xi_1)}(H^{-1}\circ F_2(X_2)-t\in A)\\
        &=\lim_{t \to \infty}\PP(H^{-1}\circ F_2(t+\xi_1+\varepsilon_2)-t\in A)\\
        &=\lim_{t \to \infty}\PP(\varepsilon_2 + \xi_1- \sigma^2/2+ o(1)\in A)\\
        &=\PP(\varepsilon_2 + \xi_1- \sigma^2/2\in A),
    \end{align*}
    where the last equation follows from dominated convergence.
    Therefore, we have that $\PP_{Y_2}^{\intervene(Y_1:=\xi_1)} = N(\xi_1-\sigma^2/2, \sigma^2)$.
  \end{proof}
    
  \subsection{Proof of Theorem~\ref{thm:SCM_limit_intervened}}\label{proof:thm-SCM_limit_intervened}
  \begin{proof}
     As in the proof of Theorem~\ref{thm:SCM_limit}, let $\pi:\{1,\dots, d\} \to \{1,\dots, d\}$ be a causal ordering of the DAG $\G$, that is, for all $i \in \an(j)$ it holds that $\pi(i) < \pi(j)$, $i,j\in V$. Since the only root node is $X^*_1$, we have $\pi(1) = 1$.
     We prove the convergence in~\eqref{conv_dist} for the set of variables $(\pi(1), \dots, \pi(m))$, $m\in \{1,\dots ,d\}$, and then use induction over $m$. 	
     For $m=1$, we note that $X_1 = f_1(\varepsilon_1)$ since there are no parents. If $1 \in \mathcal I$, that is, the root has been intervened on, then  we have
     \[\lim_{t\to\infty} \mathbb P^{\intervene(X^*_1:= t+\xi_1)}(X^*_1 - t \in A) = \mathbb P(\xi_1 \in A)  = \mathbb P(\tilde \Psi_1(\varepsilon_1) \in A),\]
     by the definition of $\tilde\Psi_1$ in~\eqref{tilde_Psi}, and for any continuity set $A\subset [-\infty, \infty)$ of $\tilde \Psi_1(\varepsilon_1)$.
     Otherwise, if $1 \notin \mathcal I$, then we have the convergence 
     \[\lim_{t\to\infty} \mathbb P(X^*_1 - t \in A) = \lim_{t\to\infty} \mathbb P(H^{-1}\circ F_1 \circ f_1(\varepsilon_1) -t \in A)  = \mathbb P(-\infty \in A) = \mathbb P(\tilde \Psi_1(\varepsilon_1) \in A),\]
     since $H^{-1}\circ F_1 \circ f_1(\varepsilon_1)$ is a random variable with values in $\mathbb R$.
     
     Now suppose that the induction hypothesis holds for $m-1$, and recall that 
     \[ \pa(\pi(m)) \subseteq \{ \pi(1), \dots, \pi(m-1) \},\]
     by the definition of a causal ordering. 
     First, consider the case where $\pi(m) \in\mathcal I$, that is, it is a node that has been intervened on. In this case, the definition of the intervention and the transformation to the exponential scale imply the $\mathbb P^{\intervene(\X^*_{\mathcal I} := t+\boldsymbol{\xi}_{\mathcal I})}$-almost sure equality
     \[X^*_{\pi(m)} - t= H^{-1}\circ F_{\pi(m)} \circ F_{\pi(m)}^{-1}\circ H(t + \xi_{\pi(m)}) - t = \xi_{\pi(m)} = \tilde \Psi_{\pi(m)}({\Y}_{\pa(\pi(m))}, \varepsilon_{\pi(m)}),\]
     and therefore the converge in distribution~\eqref{conv_dist} follows.
      
     Now suppose that $\pi(m)\notin \mathcal I$.
     Using Portmanteau theorem, our claim is equivalent to showing that for any bounded and continuous function $g: [-\infty,\infty)^{m} \to \mathbb R$ we have 
     \begin{align}\label{claim1_intervened}
       \notag\lim_{t \to \infty}  &\mathbb E^{\intervene(\X^*_{\mathcal I} := t + \boldsymbol{\xi}_{\mathcal I})}\left[ g\left(X^*_{\pi(1)} - t, \dots, X^*_{\pi(m)} - t \right) \right] \\
       &= \mathbb E\left[ g\left( \tilde \Psi_{\pi(1)}({\Y}_{\pa({\pi(1)})}, \varepsilon_{\pi(1)}), \dots, \tilde \Psi_{\pi(m)}({\Y}_{\pa({\pi(m)})}, \varepsilon_{\pi(m)}) \right) \right]. 
     \end{align}
     For this we rewrite the left-hand side as
     \[  \mathbb E^{\intervene(\X^*_{\mathcal I} := t + \boldsymbol{\xi}_{\mathcal I})}\left[ \tilde g_t \left(X^*_{\pi(1)} - t, \dots, X^*_{\pi(m-1)} - t \right) \right]  \]
     with a function $\tilde g_t: [-\infty,\infty)^{m-1} \to \mathbb R$ defined as
     \[ \tilde g_t(x_{\pi(1)}, \dots, x_{\pi(m-1)}) = \mathbb E_{\varepsilon_{\pi(m)}} \left[  g\left(x_{\pi(1)}, \dots, x_{\pi(m-1)}, f^*_{\pi(m)}( \x_{\text{pa}(\pi(m))} + t\mathbf{1}, \varepsilon_{\pi(m)}) - t \right) \right],  \]
     where the expectation is taken over the noise variable as indicated.
     Moreover, define 
     \[ \tilde g(x_{\pi(1)}, \dots, x_{\pi(m-1)}) = \mathbb E_{\varepsilon_{\pi(m)}} \left[  g\left(x_{\pi(1)}, \dots, x_{\pi(m-1)}, \Psi_{\pi(m)}( \x_{\text{pa}(\pi(m))}, \varepsilon_{\pi(m)})  \right) \right].  \]	
     With these definitions, claim~\eqref{claim1_intervened} is equivalent to showing that 
     \begin{align}\label{eq:weak_intervened}
       \notag \lim_{t \to \infty}  &\mathbb E^{\intervene(\X^*_{\mathcal I} := t + \boldsymbol{\xi}_{\mathcal I})}\left[ \tilde g_t (X^*_{\pi(1)} - t, \dots, X^*_{\pi(m-1)} - t)\right] \\
       & = \mathbb E \left[ \tilde g\left(  \tilde \Psi_{\pi(1)}({\Y}_{\pa({\pi(1)})}, \varepsilon_{\pi(1)}), \dots, \tilde \Psi_{\pi(m-1)}({\Y}_{\pa({\pi(m-1)})}, \varepsilon_{\pi(m-1)})  \right) \right].   
     \end{align}
     Since $\tilde g_t$ and $\tilde g$ are bounded functions, it suffices to show that the random variables inside the expectation converge weakly. 
     To this end, note that for any sequences $x_{\pi(i)}(t) \to x_{\pi(i)}\in[-\infty,\infty)$, $i = 1,\dots , m-1$, we have 
     \[ \lim_{t \to \infty} \tilde g_t(x_{\pi(1)}(t), \dots, x_{\pi(m-1)}(t)) = \tilde g(x_{\pi(1)}, \dots, x_{\pi(m-1)}),\]
     by the dominated convergence theorem and since 
     \[ \lim_{t \to \infty} f^*_{\pi(m)}( \x_{\pa(\pi(m))}(t) + t, e) - t  = \Psi_{\pi(m)}( \x_{\pa(\pi(m))},e),\]
     where the latter follows from the stronger version of Assumption~\ref{ass_main} from the theorem statement that also applies to sequences with possible values and limits $-\infty$. We can now apply the extended continuous mapping theorem \citep[Theorem 18.11]{vaart_1998} to conclude weak convergence of the random variables inside the expectation in~\eqref{eq:weak_intervened}, where we use the induction assumption for $m-1$.
     Note that this also includes the case where $\{\pi(m) \cup \an(\pi(m))\}\cap  \mathcal I = \emptyset$, that is, neither $\pi(m)$ nor any node upstream of $\pi(m)$ has been intervened on. Intuitively, this means that node $\pi(m)$ is in normal state since no extreme intervention is propagated to this node. In this case, according to the induction hypothesis, 
     \[ \tilde \Psi_{i}({\Y}_{\pa(i)}, \varepsilon_{i}) = -\infty, \quad \forall i \in \pa(\pi(m)), \]
     and therefore $\Psi_{\pi(m)}( \Y_{\text{pa}(\pi(m))}, \varepsilon_{\pi(m)}) = -\infty$ because of homogeneity of the function $\Psi_{\pi(m)}$.
     This establishes~\eqref{conv_dist}.
   \end{proof}
   
  \subsection{Proof of Proposition~\ref{prop:extr_cause_extr_scm}}\label{proof:prop-extr_cause_extr_scm}
  \begin{proof}		
  We first note that by Theorem~\ref{thm:SCM_limit_intervened}, for any extremal interventions $\boldsymbol{\xi}_{\mathcal V\cup \{i,j\}}$, the limit in~\eqref{limit_Y} under extremal interventions exist. Moreover,
  \[\mathbb P^{\intervene(\Y_{ V\setminus\{j\}} := \boldsymbol{\xi}_{V\cup\{j\}})}_{Y_j}(A) = \mathbb P(\Psi_j(\boldsymbol{\xi}_{\pa_{G_e}(j)}, \varepsilon_j) \in A)\]
  for any Borel set $A\subset \mathbb R$.
  If $i\notin \pa_{G_e}(j)$, then we are still intervening on all of the parents of $j$ in the extremal graph $G_e$ and therefore 
  \[ \mathbb P^{\intervene(\Y_{V \cup\{i,j\}} := \boldsymbol{\xi}_{V\cup\{i,j\}})}_{Y_j}(A) = \mathbb P(\Psi_j(\boldsymbol{\xi}_{\pa_{G_e}(j)}, \varepsilon_j) \in A) = \mathbb P^{\intervene(\Y_{ V\setminus\{j\}} := \boldsymbol{\xi}_{V\cup\{j\}})}_{Y_j}(A).\]
  Therefore, in this case, $X_i$ is not a direct extremal cause of $X_j$, or, equivalently, if $X_i$ is a direct extremal cause of $X_j$, then $i\in \pa_{G_e}(j)$.

  Let now $i\in \pa_{G_e}(j)$ and suppose that $X_i$ is not a direct extremal cause of $X_j$. This means that for all interventions $\boldsymbol{\xi}_{\mathcal V\cup \{i,j\}}$ we have
  \begin{align}\label{intervention_equality}
    \mathbb P^{\intervene(\Y_{V \cup\{i,j\}} := \boldsymbol{\xi}_{V\cup\{i,j\}})}_{Y_j}(A) =  \mathbb P^{\intervene(\Y_{ V\setminus\{j\}} := \boldsymbol{\xi}_{V\cup\{j\}})}_{Y_j}(A),
  \end{align}
  for all Borel subsets $A\subset \mathbb R$. First assume that $X_i$ is not the root node. We then have that 
  \[  \mathbb P^{\intervene(\Y_{V \cup\{i,j\}} := \boldsymbol{\xi}_{V\cup\{i,j\}})}_{Y_j}(A) =
  \mathbb P(\Psi_j(\boldsymbol{\xi}_{\pa_{G_e}(j)\setminus\{i\}}, \Psi_i(\boldsymbol{\xi}_{\pa_{G_e}(i)}, \varepsilon_i), \varepsilon_j) \in A).\]
  Similarly, we obtain
  \[ \mathbb P^{\intervene(\Y_{ V\setminus\{j\}} := \boldsymbol{\xi}_{V\cup\{j\}})}_{Y_j}(A) = \mathbb P(\Psi_j(\boldsymbol{\xi}_{\pa_{G_e}(j)}, \varepsilon_j) \in A) .\]
  We now fix a $\xi_i \in \mathbb R$ and integrate over the observational distribution of the extremal SCM over the nodes $V\setminus \{i,j\}$, that is,
  \begin{align*}
    \int \mathbb P^{\intervene(\Y_{ V\setminus\{j\}} := \boldsymbol{\xi}_{V\cup\{j\}})}_{Y_j}(A) \mathbb P_{\Y_{V\setminus \{i,j\}}^1}(\mathrm d\boldsymbol{\xi}_{V\setminus \{i,j\}}) &= \int \mathbb P(\Psi_j(\boldsymbol{\xi}_{\pa_{G_e}(j)}, \varepsilon_j) \in A) \mathbb P_{\Y_{V\setminus \{i,j\}}^1}(\mathrm d\boldsymbol{\xi}_{V\setminus \{i,j\}})\\
    & = \mathbb P(\Psi_j(\Y^1_{\pa_{G_e}(j)\setminus\{i\}}, {\xi}_{i}, \varepsilon_j) \in A)
  \end{align*}
  On the other hand, we have 
  \begin{align*}
    \int \mathbb P^{\intervene(\Y_{V \cup\{i,j\}} := \boldsymbol{\xi}_{V\cup\{i,j\}})}_{Y_j}(A)  \mathbb P_{\Y_{V\setminus \{i,j\}}^1}(\mathrm d\boldsymbol{\xi}_{V\setminus \{i,j\}}) 
    &= \int  \mathbb P( \Psi_j(\boldsymbol{\xi}_{\pa_{G_e}(j)\setminus\{i\}}, \Psi_i(\boldsymbol{\xi}_{\pa_{G_e}(i)}, \varepsilon_i))\in A) \mathbb P_{\Y_{V\setminus \{i,j\}}^1}(\mathrm d\boldsymbol{\xi}_{V\setminus \{i,j\}})\\
    & = \mathbb P(\Psi_j(\Y^1_{\pa_{G_e}(j)}, \varepsilon_j) \in A)
  \end{align*}
  By the equality in~\eqref{intervention_equality}, this would mean that we have the equality in distribution
  \[ \Psi_j(\Y^1_{\pa_{G_e}(j)}, \varepsilon_j) \stackrel{(d)}{=} \Psi_j(\Y^1_{\pa_{G_e}(j)\setminus\{i\}}, {\xi}_{i}, \varepsilon_j),\]
  which contradicts the minimality of the extremal SCM as discussed after~\eqref{eSCM2}. Indeed, otherwise we could choose $\tilde \Psi_j : \mathbb R^{|\pa_{G_e}(j)| -1}\times \mathbb R \to \mathbb R$ as $\tilde \Psi(x, e) = \Psi(x_{\pa_{G_e}(j)\setminus\{i\}}, \xi_i, e)$ without changing the distribution of the extremal SCM. Therefore, $X_i$ is a direct extremal cause of $X_j$. 
  
  If $i\in\pa_{G_e}(j)$ is the root node, then by Theorem~\ref{thm:SCM_limit_intervened} we have 
  \[  \mathbb P^{\intervene(\Y_{V \cup\{i,j\}} := \boldsymbol{\xi}_{V\cup\{i,j\}})}_{Y_j}(A) =
  \mathbb P(\Psi_j(\boldsymbol{\xi}_{\pa_{G_e}(j)\setminus\{i\}}, -\infty, \varepsilon_j) \in A).\]
  If $X_i$ was not a direct extremal cause of $X_j$, with similar arguments as above but now integrating over nodes $V\setminus\{j\}$, we could then conclude that 
  \[ \Psi_j(\Y^1_{\pa_{G_e}(j)}, \varepsilon_j) \stackrel{(d)}{=} \Psi_j(\Y^1_{\pa_{G_e}(j)\setminus\{i\}}, -\infty, \varepsilon_j),\]
  which contradicts again the minimality assumption of the function $\Psi_j$. 	
  \end{proof}
   
  \subsection{Proof of Proposition~\ref{prop:ext-markov-properties}}\label{proof:prop-ext-markov-properties}
  \begin{proof}
  Since the extremal conditional independence forms a semi-graphoid \citep[Theorem 5.3]{eng_iva_kir}, it follows along the lines of \citet[Proposition 4]{Lauritzen1990} that the extremal local and global Markov properties are equivalent.
  
  Consider a terminal node $i\in V$ of the DAG $G$, that is, $i$ has no descendants in $G$. Let $k\in\pa(i)$ and consider the auxiliary random vector $\Y^k$ defined in Section~\ref{sec:MPD} and denote by $f^k = \lambda$ its density on $\mathcal L^k$. By the definition of extremal conditional independence, the extremal local Markov property implies the classical local Markov property of $\Y^k$ on the same graph $G$. Therefore,
  \begin{align*}
    \lambda(\y)&= f^k(\y) \\
    &= f^k(y_i \mid \y_{\pa(i)}) f^k(\y_{\setminus i})\\
    &=\lambda(y_i\mid \y_{\pa(i)})\lambda(\y_{\setminus i}),\quad \y \in \mathcal L^k,
  \end{align*}
  where the second equation follow by the local 
  Markov property since $\nd(i) = V \setminus\{i\}$, and the last equation follows from the fact that the $I$-th marginal density of $\Y^k$ is equal to $\lambda(\y_I)$ whenever $k\in I$. 
  To conclude that this factorization also hold on the whole space $\mathbb R^d$, we use homogeneity of $\lambda$ and its marginals. Indeed, let $\y \in \mathbb R^d$ and choose $t \in\mathbb R$ such that $t\einsfun + \y \in \mathcal L^k$. Then
  \[ \lambda(\y) = e^{\top}\lambda(t\einsfun + \y) = e^{\top}\lambda(t + y_i\mid t\einsfun + \y_{\pa(i)})\lambda( t\einsfun + \y_{\setminus i}) = \lambda(y_i\mid \y_{\pa(i)})\lambda(\y_{\setminus i}),\]
  where we used that $\lambda(t + y_i\mid t\einsfun + \y_{\pa(i)})\lambda( t\einsfun + \y_{\setminus i}) =  \lambda(y_i\mid t\einsfun + \y_{\pa(i)})\lambda(\y_{\setminus i})$, again by homogeneity of the marginal exponent measure densities.
  
  We can now consider the reduced exponent measure density $\lambda(\y_{\setminus i})$ on the rooted DAG $G_{\setminus i} = (V\setminus \{i\} , E_{\setminus i})$, where $E_{\setminus i}$ are all edges in $E$ not involving node $i$. This is again a directed extremal graphical model, and we can iteratively apply the same arguments as before to obtain~\eqref{density_factorization}.
  
  On the other hand, if we start from the factorization~\eqref{density_factorization}, for any $i\in V$, we first integrate out the descendants of $i$ to obtain the density $\lambda(\y_{\nd(i)\cup i})$ on the sub-graph corresponding to $G$ restricted to the nodes $\nd(i)\cup \{i\}$. Because of the factorization, we can write this density as 
  \begin{align}\label{proof_fact}	
    \lambda(\y_{\nd(i)\cup \{i\}}) = \prod_{v\in \nd(i)\cup \{i\}} \lambda(y_v\mid \y_{\pa(v)})  = h(\y_{\pa(i)\cup \{i\}})g(\y_{\nd(i)}), \quad \y \in \mathcal{E},
  \end{align}
  for suitable functions $h$ and $g$. In order to check the extremal local Markov property for $\Y$, we need to check that for any $i,k\in V$ we have  
  \[ Y^k_i \indep \Y^k_{\nd(i)\setminus \pa(i)} \mid \Y^k_{\pa(i)}.\]
  Since the the density of $\Y^k$ is equal to $\lambda$ on $\mathcal L^k$, this follow immediately from~\eqref{proof_fact} and classical equivalent conditions for conditional independence \citep[Chapter 3.1]{lau1996}.
  \end{proof}
   
   \subsection{Proof of Example~\ref{ex:logistic}}
   For $m=|\pa(v)|$ it is
   \begin{align*}
     \lambda(y_v \mid \y_{\pa(v)})&=\frac{ \left(\sum_{i\in \{v\}\cup\pa(v)} \exp\left\{ -\frac{y_i}{\theta} \right\}\right)^{\theta-(m+1)} \exp\left\{-\frac{1}{\theta}\sum_{i\in \{v\}\cup\pa(v)} y_i\right\}\prod_{i=1}^{m}\left(\frac{i}{\theta}-1\right)}{ \left(\sum_{i\in\pa(v)} \exp\left\{ -\frac{y_i}{\theta} \right\}\right)^{\theta-m} \exp\left\{-\frac{1}{\theta}\sum_{i\in\pa(v)} y_i\right\}\prod_{i=1}^{m-1}\left(\frac{i}{\theta}-1\right)}\\
     &=\frac{\left(\exp\left\{ -\frac{y_v}{\theta} \right\}+\sum_{i\in \pa(v)} \exp\left\{ -\frac{y_i}{\theta} \right\}\right)^{\theta-m-1} \left(\frac{m}{\theta}-1\right)}{ \left(\sum_{i\in\pa(v)} \exp\left\{ -\frac{y_i}{\theta} \right\}\right)^{\theta-m} \exp\left\{\frac{y_m}{\theta} \right\}}\\
     &=\left(1 + \frac{1}{\sum_{i \in \pa(v)} \exp\left\{ \frac{y_v - y_i}{\theta}\right\} } \right)^{\theta-|\pa(v)| - 1}	\frac{|\pa(v)|/ \theta-1}{\sum_{i \in \pa(v)} \exp\left\{ \frac{y_v-y_i}{\theta}\right\} }.
   \end{align*}
   
   \subsection{Proof of Example~\ref{ex:dirichlet}}
    For $m=|\pa(v)|$ we have
  \begin{align*}
    \lambda(y_v \mid \y_{\pa(v)})&=\frac{\frac{1}{m+1} \frac{\Gamma(1+\sum_{i\in \{v\}\cup\pa(v)} \alpha_i) \exp\left\{\sum_{i\in \{v\}\cup\pa(v)} y_i\right\}}{(\sum_{i\in \{v\}\cup\pa(v)} \alpha_i \exp\{y_i\})^{m+2}} \prod_{i\in \{v\}\cup\pa(v)} \frac{\alpha_i}{\Gamma(\alpha_i)}  
      \left(\frac{\alpha_i \exp\{y_i\}}{\sum_{j\in \{v\}\cup\pa(v)} \alpha_j \exp\{y_j\}}\right)^{\alpha_i-1}}{ \frac 1 m \frac{\Gamma(1+\sum_{i\in\pa(v)} \alpha_i) \exp\left\{\sum_{i\in\pa(v)} y_i\right\}}{(\sum_{i\in\pa(v)} \alpha_i \exp\{y_i\})^{m+1}} \prod_{i\in\pa(v)} \frac{\alpha_i}{\Gamma(\alpha_i)}  
      \left(\frac{\alpha_i \exp\{y_i\}}{\sum_{j\in\pa(v)} \alpha_j \exp\{y_j\}}\right)^{\alpha_i-1}}\\
    &=\frac{m\Gamma(1+\sum_{i\in \{v\}\cup\pa(v)} \alpha_i)}{(m+1)\Gamma(1+\sum_{i\in\pa(v)} \alpha_i)\Gamma(\alpha_v)}\times\\ 
    &\quad \frac{\exp(y_v)(\sum_{i\in\pa(v)} \alpha_i \exp\{y_i\})^{m+1}}{(\sum_{i\in \{v\}\cup\pa(v)} \alpha_i \exp\{y_i\})^{m+2}}\frac{(\alpha_v \exp(y_v))^{\alpha_v-1}\left(\sum_{j\in\pa(v)} \alpha_j \exp\{y_j\}\right)^{\sum_{i \in \pa(v)}\alpha_i-m}}{\left(\sum_{j\in \{v\}\cup\pa(v)} \alpha_j \exp\{y_j\}\right)^{\sum_{i\in \{v\}\cup\pa(v)} \alpha_i -(m+1)}}\\
    &=C\times \frac{(\alpha_v \exp(y_v))^{\alpha_v}\left(\sum_{j\in\pa(v)} \alpha_j \exp\{y_j\}\right)^{1+\sum_{i \in \pa(v)}\alpha_i}}{\left(\sum_{j\in \{v\}\cup\pa(v)} \alpha_j \exp\{y_j\}\right)^{1+\sum_{i\in \{v\}\cup\pa(v)} \alpha_i}}\\
    &=C\times (\alpha_v \exp(y_v))^{\alpha_v}\left(\sum_{j\in\pa(v)} \alpha_j \exp\{y_j\}\right)^{-\alpha_v}\left(\frac{\sum_{j\in\pa(v)} \alpha_j \exp\{y_j\}}{\sum_{j\in \{v\}\cup\pa(v)} \alpha_j \exp\{y_j\}}\right)^{1+\sum_{i\in \{v\}\cup\pa(v)} \alpha_i}\\
    &= C\times   \left( 1 + \frac{1}{\sum_{j \in \pa(v)} \frac{\alpha_j}{\alpha_v} \exp\{y_j - y_v\}}\right)^{-1 -\sum_{i\in\{v\}\cup\pa(v) }\alpha_i} 
    \left( \sum_{j \in \pa(v)} \frac{\alpha_j}{\alpha_v} \exp\{y_j - y_v\} \right)^{-\alpha_v},
  \end{align*}
  for some constant $C>0$.  
   
   \subsection{Proof of Proposition~\ref{prop:extr_scm_is_dgm}}\label{proof:prop-extr_scm_is_dgm}
   \begin{proof}
     First note that the DAG $G_e$ is necessarily rooted.
     Suppose that $\Y$ is an extremal SCM on $G_e$. It is enough to check the local Markov property for $\Y^1$, which satisfies the structural causal model~\eqref{eSCM2} by assumption. For an arbitrary fixed $i\in V$, we observe that the only randomness in the conditional distribution $(Y^{1}_i  \mid \Y^1_{\pa(i)} = \y_{\pa(i)}) = \Psi_i(\y_{\pa(i)}, \varepsilon_i)$ comes from the independent noise variable $\varepsilon_i$, and therefore 
     $Y^1_i \perp_e \Y^1_{\nd(i)\setminus \pa(i)} \mid \Y^1_{\pa(i)}$.
   \end{proof}
   
   \subsection{Proof of Proposition~\ref{prop:ESCM}}\label{proof:prop-ESCM}
   \begin{proof}
     Denote by $f$ the density of $\X$, and by $f_{\varepsilon_v}$ the density of $\varepsilon_v$. 
     Thus,
     \[f(\x)=e^{-x_1}\prod_{v=2}^{|V|} f(x_v\mid \x_{\pa(v)})=e^{-x_1}\prod_{v=2}^{|V|} f_{\varepsilon_v}(x_v - \Phi_v(\x_{\pa(v)})). \]
     If we select $\Y^1\stackrel{d}{=}\X$, we have $\lambda(\y)=f(\y)$ and
     \begin{align*}
       \lambda(y_v\mid \y_{\pa(v)}) = f_{\varepsilon_v}(y_v - \Phi_v(\y_{\pa(v)})).
     \end{align*}
     We need to check that the sufficient conditions for valid conditional exponent measure densities are satisfied.
     First, we observe that $\lambda( \cdot \mid \y_{\pa(v)})$ is a probability density and satisfies
     \begin{align*}
       \lambda(y_v + t\mid \y_{\pa(v)}+ t\einsfun) = f_{\varepsilon_v}(y_v + t - \Phi_v(\y_{\pa(v)} + t\einsfun)) = f_{\varepsilon_v}(y_v - \Phi_v(\y_{\pa(v)})) = \lambda(y_v\mid \y_{\pa(v)}).
     \end{align*}
     Therefore, $\lambda$ satisfies homogeneity.
     Concerning the marginal constraints
     \begin{align*}
       \int_{y_v>0}\lambda(\y)d\y=1,
     \end{align*}
     we observe that this is given for $v=1$. Now, the transformation $\varphi(\y)=\y+(y_1-y_v)\mathbf{1}$ maps the set $\{y_1>0\}$ to $\{y_v>0\}$. By substitution we obtain  
     \begin{align*}
       \int_{y_v>0}\lambda(\y)d\y&= \int_{y_1>0}\lambda(\y+(y_1-y_v)\mathbf{1})d\y\\
       &= \int_{y_1>0} e^{y_v - y_1}\lambda(\y)d\y\\
       &=\EE\left(\exp(\Phi_v(\X_{\pa(v)}-X_1\mathbf{1})+\varepsilon_v)\right),
     \end{align*}
     such that by independence of the noise variables, the marginal constraints follow from the moment assumptions.
   \end{proof}
   
   \subsection{Proof of Proposition~\ref{prop:HR_SCM_construction}}
   \begin{lemma}\label{lem:LDL_psdLaplacian}
     Let $\nu_2^2,\ldots,\nu_{d}^2$ be positive scalars and $G_e$ a DAG rooted in 1 with some weighted adjacency matrix $B$ which satisfies $B^{\top}\einsfun=\mathbf{0}$.
     Then, the matrix
     \begin{align*}
       \Theta = \mathfrak{L}(G_e)^{\top} \diag(\nu_2^2,\ldots,\nu_d^2)^{-1}\mathfrak{L}(G_e).
     \end{align*}
     is a positive semidefinite signed Laplacian matrix. 
   \end{lemma}
   \begin{proof}
     $\Theta$ is positive semidefinite, as for any $\x\in\RR^d\setminus\{\mathbf{0}\}$ it is $\x^{\top}\Theta\x=\tilde{\x}^{\top}\diag(\nu_2^2,\ldots,\nu_d^2)^{-1}\tilde{\x}\ge 0$, where $\tilde{\x}=\mathfrak{L}(G_e)\x$. With $\mathfrak{L}(G_e)\einsfun=\mathbf{0}$, this yields that $\Theta$ is a signed Laplacian matrix.
   \end{proof}
   
   \begin{proof}[Proof of Proposition~\ref{prop:HR_SCM_construction}]
     According to Lemma~\ref{lem:LDL_psdLaplacian}, the matrix $\Theta$ is a positive semidefinite signed Laplacian.
     The structural assignments $\Phi_v(\x_{\pa(v)})=\sum_{i \in {\pa(v)}} b_{iv} x_i$ are homogeneous.
     Thus, by Proposition~\ref{prop:ESCM}, $\X$ gives rise to an extremal SCM when the moment condition
     \begin{align*}
       \EE\{\exp(\varepsilon_v)\}=\frac{1}{\EE[\exp\{\Phi_v(\X_{\pa(v)}-X_1\mathbf{1})\}]}.
     \end{align*}
     is satisfied. Let $\X=(I-B^{\top})^{-1}\N^1$. As $G_e$ is rooted in 1 and $\mathfrak{L}(G_e)\bs 1 = \bs 0$, the first column of $(I-B^{\top})^{-1}$ is equal to $\einsfun$.  Let $A=(I-B^{\top}_{\setminus 1,\setminus 1})^{-1}$ and observe
     \begin{align*}
       \begin{pmatrix}
         1&\mathbf{0}^{\top}\\
         -B_{1,V}^{\top}& (I-B^{\top}_{\setminus 1,\setminus 1})\\
       \end{pmatrix}
       \begin{pmatrix}
         1&\mathbf{0}^{\top}\\
         \einsfun &A\\
       \end{pmatrix}=I,
     \end{align*}
     such that the matrix on the right is $(I-B^{\top})^{-1}$.
     Hence,
     \[\X=R\einsfun + (I-B^{\top})^{-1} (0,\varepsilon_2,\ldots,\varepsilon_d)^{\top}, \]
     and in particular $\X-X_1\einsfun = (0, A(\varepsilon_2,\ldots,\varepsilon_d)^{\top})$.
     The moment condition for $v\neq 1$ simplifies to 
     \begin{align*}
       1=\EE(\exp(\varepsilon_v+\phi_v(\X_{\pa(i)})-X_1))&=\EE(\exp(X_v-X_1))
     \end{align*}
     As $(\varepsilon_2,\ldots,\varepsilon_d)^{\top}$ is multivariate Gaussian with mean $-\frac{1}{2}(I-B_{\setminus 1, \setminus 1}^{\top})\Gamma_{\setminus 1 ,1}$ and covariance matrix $D_{\varepsilon}$, it follows that $A (\varepsilon_2,\ldots,\varepsilon_d)^{\top}$ is multivariate Gaussian with mean $-\frac{1}{2}\Gamma_{V,1}$ and covariance matrix $AD_{\varepsilon}A^{\top}$.
     From the shape of $\Theta$ it follows that $\Theta^{(1)}=\Theta_{\setminus 1,\setminus 1}= (I-B_{\setminus 1,\setminus 1})D_\varepsilon^{-1}(I-B^{\top}_{\setminus 1,\setminus 1}).$
     Thus, $\Sigma^{(1)}=(\Theta^{(1)})^{-1}=AD_\varepsilon A^{\top}$.
     With $\Gamma_{v,1}=\Sigma^{(1)}_{vv}$, the moment condition holds as
     \[\EE(\exp(A_{v,V\setminus \{1\}}(\varepsilon_2,\ldots,\varepsilon_d)^{\top}))=\exp\left(-\frac{1}{2}\Sigma^{(1)}_{vv}+\frac{1}{2}\Sigma^{(1)}_{vv}\right)=1.\]
     If we identify $\Y^1\stackrel{d}{=}\X$, the extremal function $\W^1_{\setminus 1}=A(\varepsilon_2,\ldots,\varepsilon_d)^{\top}$ is Gaussian. This gives rise to a \HR{} extremal SCM with respect to $G_e$.
   \end{proof}
   
  \subsection{Proof of Corollary~\ref{cor:skel_v-struct}}
  
  \begin{proof}
    Let $ G_e $ and $ H_e $ be two rooted DAGs that have the same skeleton graph and the same v-structures.
    Let $\mathbb P\in \mathcal M(G_e)$ and let $\Y\sim \mathbb P$ be the corresponding multivariate Pareto random vector. We need to show that $\Y$ also satisfies the extremal global Markov property with respect to $H_e$. To this end, note that by the definition of extremal conditional independence, we have that the auxiliary vector $ \Y^k $ is Markov to $ G_e $ in the usual sense for all $k\in V$. Since $G_e$ and $H_e$ 
    are Markov equivalent in the classical sense as they have the same skeleton and v-structures~\citep[Theorem~1.2.8]{Pearl2009}, we conclude that $\Y^k$ is Markov to $H_e$ for all $k\in V$. This implies that $\Y$ satisfies the extremal global Markov property with respect to $H_e$. 
  \end{proof}
   
  \subsection{Proof of Proposition~\ref{prop:algo-oracle}}
  \begin{proof}\label{proof:prop-algo-oracle}
    For each edge $(i, j) \in \E$, define the collection of sets separating $i$ and $j$ in $G_e$ as
    $$\mathcal{S}^e_{i, j} \coloneqq \{S \subseteq V \setminus\{i, j\} \colon i \indep_{G_e} j \mid S\} = \{S \subseteq V \setminus\{i, j\} \colon Y_i \perp_e Y_j \mid \Y_S\},$$
    where the second equality follows since the distribution of $\Y$ is faithful to $G_e$. Moreover, we have 
  \begin{align}\label{eq:oracle-implication}
    (i, j) \in \E_e  \quad \Leftrightarrow \quad \mathcal S_{i,j}^e = \emptyset.
  \end{align}
  Define the set of prunable edges as $P \coloneqq \{(i, j) \in \E \colon |pa_{G}(j)| \geq 2\}$ 
  and enumerate them as $P = \{e_1, \dots, e_{|P|}\}$. 
  We will prove by strong induction that for all $ \ell \in \{1, \dots, |P|\}$, the following holds: 
  \begin{align}\label{eq:no-mistake}
  \text{Algorithm~\ref{alg:ext_pruning_v2} removes edge } e_\ell  \iff e_\ell \notin \E_e.
  \end{align} 
  
  We start with the base case ($\ell = 1$). Let $e_1 \coloneqq (i, j) \in P$. 
  Define the DAG $G_{i, j} \coloneqq (V, \E \setminus \{(i, j)\})$, and the collection of sets separating $i$ and $j$ in $G_{i, j}$ as $\mathcal{S}^{(1)}_{i, j} \coloneqq \{S \subseteq V \setminus\{i, j\} \colon i \indep_{G_{i, j}} j \mid S\}$.
  First, we show that $\mathcal{S}^{(1)}_{i, j}$ is non-empty.
  Define $S \coloneqq pa_{G_{i, j}}(j) \setminus \{i\}$. Since $|pa_G(j)| \geq 2$, after removing $(i, j)$, $|pa_{G_{i, j}}(j)| \geq 1$.
  Any path from $i$ to $j$ in $G_{i, j}$ would have to go through $S$ or via some collider. Therefore, $S$ blocks all paths from $i$ to $j$, and thus $S \in \mathcal{S}^{(1)}_{i, j}$.
  
  We now consider two cases.
  
  \begin{enumerate}
    \item\, Suppose $(i, j) \notin \E_e$. Then, by the bi-implication in~\eqref{eq:oracle-implication}, for all $S \in \mathcal{S}^e_{i, j}$ it holds $Y_i \perp_{e} Y_j \mid \Y_S$.
    Moreover, since $\E_e \subseteq \E \setminus \{(i, j)\}$, it holds that $\mathcal{S}^{(1)}_{i, j} \subseteq \mathcal{S}^{e}_{i, j}$.
    Since $\varnothing \neq \mathcal{S}^{(1)}_{i, j} \subseteq \mathcal{S}^{e}_{i, j}$, then there exists a separating set $S \in \mathcal{S}^{(1)}_{i, j}$ such that $Y_i \perp_{e} Y_j \mid \Y_S$. 
    Therefore, the conditional statement on Line~\ref{alg:ci-statement} is fulfilled, and we remove the edge $e_1 \in P$. 
    Define the DAG $G_1 \coloneqq (V, \E_1)$ with $\E_1 = \E  \setminus \{e_1\}$ and note that $G_1$ is still rooted since $e_1 \in P$. 
  
    \item\, Suppose $(i, j) \in \E_e$. 
    Then, by the bi-implication in~\eqref{eq:oracle-implication}, for all $S \in V \setminus \{i, j\}$ it holds $Y_i \not\perp_{e} Y_j \mid \Y_S$. 
    Therefore, the conditional statement on Line~\ref{alg:ci-statement} is not fulfilled, and we do not remove $e_1 \in P$.
    Define the DAG $G_1 \coloneqq (V, \E_1)$ with $\E_1 = E$.
  \end{enumerate}
  Putting the cases (1) and~(2) together, we have shown that
  \begin{align*}
    \text{Algorithm~\ref{alg:ext_pruning_v2} removes edge } e_1  \iff e_1\notin \E_e.
  \end{align*}
  
  We now show the {induction step $(\ell = m + 1)$}. Fix $1 \leq m \leq |P|$
  and suppose that the following implication holds
  \begin{align}\label{eq:no-mistake-induction}
    \forall \ell \in \{1, \dots, m\}\ 
    (\text{Algorithm~\ref{alg:ext_pruning_v2} removes edge } e_\ell  \iff e_\ell \notin \E_e\big).
  \end{align}
  By the induction step~\eqref{eq:no-mistake-induction}, we have that the DAG $G_m = (V, \E_m)$ satisfies  $\E_e \subseteq \E_m \subseteq E$. 
  Let $e_{m + 1} \coloneqq (i, j) \in P$. 
  First, note that $j$ is not the unique root node of $G$, and thus neither of $G_e$.
  
  Consider now two cases.
  
  \begin{enumerate}   
    \item\, If $|pa_{G_m}(j)| < 2$, 
    Algorithm~\ref{alg:ext_pruning_v2} skips $e_{m+1}$
    (Line~\ref{alg:line-3}), so $e_{m + 1}$ remains in $\E_m$.
    Since $E_e \subseteq E_m$, and $j$ must have at least one parent in $G_e$ (since $j$ is not the unique root node of $G_e$), $e_{m + 1} \in \E_e$. 
    Therefore, we have shown the implication $
      e_{m + 1} \notin \E_e \implies \text{Algorithm~\ref{alg:ext_pruning_v2} removes edge } e_{m+1}$.
    The converse holds vacuously because of $|pa_{G_m}(j)| = 1$ and Line~\ref{alg:line-3} of Algorithm~\ref{alg:ext_pruning_v2}. Therefore we have shown that 
    \begin{align*}
      e_{m + 1} \notin \E_e \iff \text{Algorithm~\ref{alg:ext_pruning_v2} removes edge } e_{m+1}.
    \end{align*} 
  
    \item\, If $|pa_{G_m}(j)| \geq 2$, define the DAG $G_{i, j} \coloneqq (V, \E_m \setminus \{(i, j)\})$ and the collection of sets separating $i$ and $j$ in $G_{i, j}$ as $\mathcal{S}^{(m+1)}_{i, j} \coloneqq \{S \subseteq V \setminus\{i, j\} \colon i \indep_{G_{i, j}} j \mid S\}$. 
    By defining $S \coloneqq pa_{G_{i, j}}(j) \setminus \{i\}$, and since $|pa_{G_m}(j)| \geq 2$, using the same argument as in the base case, the collection of separating sets $\mathcal{S}^{(m+1)}_{i, j} \neq \varnothing$.
  
    We now consider two cases.
  
    \begin{enumerate}
      \item\, Suppose $(i, j) \notin \E_e$. Then, $\E_e \subseteq \E_m \setminus \{(i, j)\}$ and therefore $\mathcal{S}^{(m+1)}_{i, j} \subseteq \mathcal{S}^{e}_{i, j}$. By the same argument as in the base case, the conditional statement on Line~\ref{alg:ci-statement} is fulfilled, and we remove the edge $e_{m+1} \in P$. 
      \item\, Suppose $(i, j) \in \E_e$. Then, by the bi-implication in~\eqref{eq:oracle-implication}, the conditional statement on Line~\ref{alg:ci-statement} is not fulfilled, and we do not remove $e_{m+1} \in P$. 
    \end{enumerate}
  
    Putting the cases (a) and~(b) together, we have shown that
    
    \begin{align*}
    \text{Algorithm~\ref{alg:ext_pruning_v2} removes edge } e_{m + 1}  \iff e_{m + 1}\notin \E_e.
    \end{align*}
  \end{enumerate}
  \end{proof}
   
  \subsection{Proof of Proposition~\ref{prop:algo-consistency}}
  \begin{proof}\label{proof:prop-algo-consistency}
    Some ideas of this proof come from~\citet{kalisch2007estimating}.
    For all $i, j \in V$ and $S \subseteq V \setminus \{i, j\}$, define the population and sample Fisher z-transforms
    \begin{align}
      z_{ij\mid S} \coloneqq \frac{1}{2} \log\left(\frac{1 + \rho_{ij \mid S}}{1 - \rho_{ij \mid S}}\right),\quad
      Z_{ij\mid S} \coloneqq \ \frac{1}{2} \log\left(\frac{1 + \hat\rho_{ij \mid S}}{1 - \hat\rho_{ij \mid S}}\right).
    \end{align}
    For all prunable edges $(i, j) \in P$ and $S \subseteq V \setminus\{i, j\}$, define the type I and type II errors
    \begin{align*}
      M_{ij\mid S}^{I} \coloneqq&\  \left\{ \sqrt{n - |S| - 3}\ |Z_{ij\mid S}| > \Phi^{-1}(1 - \alpha_n / 2) \right\}\quad \text{if}\quad z_{ij\mid S} = 0,\\
      M_{ij\mid S}^{II} \coloneqq&\ \left\{ \sqrt{n - |S| - 3}\ |Z_{ij\mid S}| \leq \Phi^{-1}(1 - \alpha_n / 2)  \right\}\quad  \text{if} \quad z_{ij\mid S} \neq 0.
    \end{align*}
    We have 
    \begin{align*}
      \left(\forall (i, j) \in P\ \forall S \subseteq V \setminus \{i, j\} \colon (M^{I}_{ij\mid S} \cup M^{II}_{ij \mid S})^c\right) \implies \hat{G}_e(\alpha_n) = G_e.
    \end{align*}
    Therefore,
    \begin{align}\label{eq:algo-mistake-union-bound}
    \begin{split}
      \PP\left(\hat{G}_e(\alpha_n) \neq G_e\right) 
      \leq &\
      \PP\left(\bigcup_{(i, j) \in P} \bigcup_{S \subseteq V \setminus \{i, j\}} (M_{ij\mid S}^{I} \cup M_{ij\mid S}^{II})
      \right) \\
        \leq &\
        2^{d} |P|
        \left(
        \max_{\substack{(i, j) \in P,\\ S \in V \setminus \{i, j\}}}
          \PP(M_{ij\mid S}^{I})
        +
        \max_{\substack{(i, j) \in P,\\ S \in V \setminus \{i, j\}}}
          \PP(M_{ij\mid S}^{II})
        \right).
    \end{split}
    \end{align}
    Let  
    \begin{align*}
      m^* \coloneqq \min\left\{|z_{ij \mid S}| \colon (i, j) \in P, S \subseteq V \setminus \{i, j\}, z_{ij \mid S} \neq 0\right\} > 0,
    \end{align*}  
    and
    choose the sequence $\alpha_n = 2(1 - \Phi(n^{1/2} m^{*}/2)) \to 0$ as $n \to \infty$.
    Consider now the type I error event $M_{ij\mid S}^{I}$. If $z_{ij\mid S} = 0$, we have
    \begin{align*}
      \PP\left( M_{ij \mid S}^{I}  \right)
      =&\
      \PP\left( 
        \left| Z_{ij\mid S}\right| > \sqrt{\frac{n}{n - |S| - 3}} \frac{m^{*}}{2} 
        \right)
      \leq
      \PP\left( 
        \left| Z_{ij\mid S} - z_{ij\mid S}\right| > \frac{m^{*}}{2} 
      \right).
    \end{align*}
    By the continuous mapping and the assumption of consistency of $\hat\rho_{ij \mid S}$, it follows that $Z_{ij\mid S}$ is consistent for $Z_{ij\mid S}$. Thus, $\PP(M_{ij \mid S}^{I}) \to 0$ as $n \to \infty$. 
    Consider now the type II error event $M_{ij \mid S}^{II}$. 
    There exists an $n_0 \geq 1$ such that for all $n \geq n_0$ it holds that
    \begin{align*}
      \sqrt{\frac{n}{n - |S| - 3}} < \frac{3}{2},
      \text{ and thus }
      \sqrt{\frac{n}{n - |S| - 3}} \frac{m^{*}}{2} < \frac{3}{4} m^{*}.
    \end{align*}
    If $z_{ij\mid S} \neq 0$ we have $\left|z_{ij\mid S} \right| \geq m^{*}$ by definition.
    Therefore, for all $n \geq n_0$ it holds that
    \begin{align*}
      M_{ij \mid S}^{II}
      =&\ 
      \left\{ 
        \left| Z_{ij\mid S}\right| \leq \sqrt{\frac{n}{n - |S| - 3}} \frac{m^{*}}{2} 
      \right\}\\
      \subseteq&\
      \left\{ 
        \left| Z_{ij\mid S}\right| \leq \frac{3}{4}m^{*} 
      \right\}
      \subseteq
      \left\{ 
        \left| Z_{ij\mid S} - z_{ij \mid S} \right| > \frac{m^{*}}{4}
      \right\}.
    \end{align*}
    Since $m^{*} / 4 > 0$ and the consistency of $ Z_{ij\mid S}$, we have that $\PP(M_{ij \mid S}^{II}) \to 0$ as $n \to \infty$.
    Putting together~\eqref{eq:algo-mistake-union-bound} with the fact that, as $n \to \infty$, $\PP(M_{ij \mid S}^{I}) \to 0$ and $\PP(M_{ij \mid S}^{II}) \to 0$ concludes the proof.
  \end{proof}

  \section{Additional results and examples for H\"usler--Reiss distributions}
   
  Let $\Y$ be a H\"usler--Reiss Pareto distribution with variogram matrix $\Gamma$ and precision matrix $\Theta$. In this section we collect some additional results that will be used throughout the paper.
  Recall from Example~\ref{ex:CI_test} that the auxiliary vector of a  \HR{} distribution admit the stochastic representation
  $
  \Y^k\stackrel{d}{=}R\mathbf{1}+\W^k
  $
  where $R$ is a standard exponential random variable and $\W^k$ is a degenerate normal vector with $W^k_k=0 $ almost surely, and $\W^k_{\setminus k} \sim N(-\Gamma_{\setminus k, k}/2, \Sigma^{(k)})$, where $\Sigma^{(k)} \coloneqq \frac{1}{2} \left(\Gamma_{uk} + \Gamma_{vk} - \Gamma_{uv}\right)_{u, v \in V \setminus k}$.
  Then, the nodes $i, j \in V$ are extremal conditionally independent given some collection of nodes $S \subseteq V \setminus \{i, j\}$, i.e., $Y_i\perp_{e}Y_j \mid \Y_S$, if and only if $W_i^k\ci W_j^k\mid \W_S^k$ for some $k\in S$ \citep{EH2020}.
  
  The marginal \HR{} exponent measure density $\lambda(\y_A)$ for some $A\subseteq V$ is again \HR{} with parameter matrix $\Gamma_{A,A}$.
  We define $\theta(A)$, $\mathbf{p}(A)$ and $\sigma^2(A)$ as in identity \eqref{eq:Fiedler-Bapat_margin}.
  
  \begin{proof}[Proof of Example~\ref{ex:HR_cond_prelim}]\label{prf:ex:HR_cond_prelim}
    Let $e_d=\frac{1}{d}\mathbf{1}$.
    The exponent measure density in \cite{HES2022} writes as
    \begin{align*}
      \lambda(\y)= c_\Theta \exp\left\{-\frac{1}{2}\y^{\top}\Theta\y-\y^{\top}(e_d-r_\Theta)\right\},
    \end{align*}
    where $r_\Theta=-\frac{1}{2}\Theta\Gamma e_d$.
    \citet[Corollary 3.7]{devriendt2022a} derives $\Theta\Gamma=-2I+2\mathbf{p}\mathbf{1}^{\top}$ from the Fiedler--Bapat identity, where $I$ denotes the identity matrix. 
    Thus, 
    \[r_\Theta=e_d-\mathbf{p},\]
    which yields the alternative representation with a modified constant $c_\Gamma =c_\Theta\exp(\sigma^2/2)$.
  \end{proof}
  
  \begin{lemma}\label{lem:condHR}
    Let $\Y$ be \HR{} multivariate Pareto distribution with variogram matrix $\Gamma$. We then have for any disjoint $A, C \subset V$ that
    \begin{align*}
      \lambda(\y_A\mid \y_{C}) &=  \frac{1}{\sqrt{\det(2\pi\Sigma_*)}}\exp\left(-\frac{1}{2} (\y_A-\mu^*)^{\top}\Sigma_*^{-1}(\y_A-\mu^*) \right),
    \end{align*}
    where 
    \begin{align*}
      \Sigma_*&=-\frac{1}{2}\Gamma_{A,A} -\begin{pmatrix}
        -\frac{1}{2}\Gamma_{A,C}& \einsfun\\
      \end{pmatrix}
      \begin{pmatrix}
      -\frac{1}{2}\Gamma_{C,C} & \mathbf{1}\\
      \mathbf{1}^{\top}&0\\
      \end{pmatrix}^{-1}
      \begin{pmatrix}
        -\frac{1}{2}\Gamma_{C,A}\\
        \einsfun^{\top}\\
      \end{pmatrix}
    \end{align*} 
    and
    \begin{align*}
      \mu_*=\begin{pmatrix}
        -\frac{1}{2}\Gamma_{A,C}, \einsfun\\
      \end{pmatrix}\begin{pmatrix}
        -\frac{1}{2}\Gamma_{C,C} & \mathbf{1}\\
        \mathbf{1}^{\top}&0\\
      \end{pmatrix}^{-1}\begin{pmatrix}
        \y_{C}\\
        1\\
      \end{pmatrix}.
    \end{align*}
    Furthermore, it holds that $\Sigma_*^{-1}=\left(M^{-1}\right)_{A,A}=\theta(A\cup C)_{A,A}$.
  \end{lemma}
  
  \begin{proof}
    The conditional probability density of this distribution is
    \begin{align*}
      \lambda(\y_A\mid \y_{C})&=\frac{\lambda(\y_A,\y_{C})}{\lambda(\y_{C})}\\
      &= \frac{c_{\Gamma_{A\cup C,A\cup C}}}{c_{\Gamma_{C,C}}}\frac{\exp\left(-\frac{1}{2}\begin{pmatrix}
          \y_A,\y_{C}^{\top},1\\
        \end{pmatrix}
        \begin{pmatrix}
          -\frac{1}{2}\Gamma_{A\cup C,A\cup C} & \mathbf{1}\\
          \mathbf{1}&0\\
        \end{pmatrix}^{-1}\begin{pmatrix}
          \y_A\\
          \y_{C}\\
          1\\
        \end{pmatrix}\right)}{\exp\left(-\frac{1}{2}\begin{pmatrix}
          \y_{C}^{\top},1\\
        \end{pmatrix}
        \begin{pmatrix}
          -\frac{1}{2}\Gamma_{C, C} & \mathbf{1}\\
          \mathbf{1}&0\\
        \end{pmatrix}^{-1}\begin{pmatrix}
          \y_{C}\\
          1\\
        \end{pmatrix}\right)}.
    \end{align*}
    Let $M_{11}=-\frac{1}{2}\Gamma_{A,A}$, $M_{21}=\begin{psmallmatrix}
      -\frac{1}{2}\Gamma_{C,A}\\
      1\\
    \end{psmallmatrix}$ and $M_{2,2}=\begin{psmallmatrix}
      -\frac{1}{2}\Gamma_{C,C} & \mathbf{1}\\
      \mathbf{1}&0\\
    \end{psmallmatrix}$ and let $M=	\begin{psmallmatrix}
      M_{11}&M_{21}^{\top}\\
      M_{21}&M_{22}\\
    \end{psmallmatrix}$.
    Note that $M_{22}$ is invertible.
    We call $M/M_{22}:=M_{1}-M_{21}^{\top}M_{22}^{-1}M_{21}$ the Schur complement of the block $C$.
    Then the quadratic form in the exponent in the numerator can be rewritten with standard Schur complement arguments as
    \begin{align*}
      &\begin{pmatrix}
        \y_A,\y_{C}^{\top},1\\
      \end{pmatrix}
      M^{-1}\begin{pmatrix}
        \y_A\\
        \y_{C}\\
        1\\
      \end{pmatrix}\\&=\begin{pmatrix}
        \y_A,\y_{C}^{\top},1\\
      \end{pmatrix}
      \begin{pmatrix}
        (M/M_{22})^{-1}&-(M/M_{22})^{-1}M_{21}^{\top}M_{22}^{-1}\\
        -M_{22}^{-1}M_{21}(M/M_{22})^{-1} & M_{22}^{-1}+M_{22}^{-1}M_{21}(M/M_{22})^{-1}M_{21}^{\top}M_{22}^{-1}\\
      \end{pmatrix}
      \begin{pmatrix}
        \y_A\\
        \y_{C}\\
        1\\
      \end{pmatrix}\\
      &=\y_A^{\top}(M/M_{22})^{-1} \y_A-\y_A^{\top}(M/M_{22})^{-1}M_{21}^{\top}M_{22}^{-1}	\begin{pmatrix}
        \y_{C}\\
        1\\
      \end{pmatrix}-\begin{pmatrix}
        \y_{C}^{\top}&1\\
      \end{pmatrix}M_{22}^{-1}M_{21}(M/M_{22})^{-1}\y_A\\ &\qquad+\begin{pmatrix}
        \y_{C}^{\top}&1\\
      \end{pmatrix}
      (M_{22}^{-1}+M_{22}^{-1}M_{21}(M/M_{22})^{-1}M_{21}^{\top}M_{22}^{-1})
      \begin{pmatrix}
        \y_{C}\\
        1\\
      \end{pmatrix}\\
      &=\left(\y_A-M_{21}^{\top}M_{22}^{-1}	\begin{pmatrix}
        \y_{C}\\
        1\\
      \end{pmatrix}\right)^{\top}(M/M_{22})^{-1}\left(\y_A-M_{21}^{\top}M_{22}^{-1}	\begin{pmatrix}
        \y_{C}\\
        1\\
      \end{pmatrix}\right)+
      \begin{pmatrix}
      \y_{C}^{\top}&1\\
      \end{pmatrix}
      M_{22}^{-1}
      \begin{pmatrix}
      \y_{C}\\
      1\\
      \end{pmatrix}\\
    \end{align*} 
    If we plug this into the conditional density, we obtain	
    \begin{align*}
      \lambda(\y_A\mid \y_{C})&= \frac{c_{\Gamma_{A\cup C,A\cup C}}}{c_{\Gamma_{C,C}}}\exp\left(-\frac{1}{2} \left(\y_A-M_{21}^{\top}M_{22}^{-1}	\begin{pmatrix}
        \y_{C}\\
        1\\
      \end{pmatrix}\right)^{\top}(M/M_{22})^{-1}\left(\y_A-M_{21}^{\top}M_{22}^{-1}	\begin{pmatrix}
        \y_{C}\\
        1\\
      \end{pmatrix}\right)\right)
    \end{align*}
    This is an $|A|$-variate Gaussian probability density with covariance matrix $\Sigma_*=(M/M_{22}$ and mean $\mu_*=M_{21}^{\top}M_{22}^{-1}	\begin{pmatrix}
      \y_{C}\\
      1\\
    \end{pmatrix}$.
    From the Fiedler--Bapat identity \eqref{eq:Fiedler-Bapat_margin} we further obtain that $(M/M_{22})^{-1}=\left(M^{-1}\right)_{A,A}=\theta(A\cup C)_{A,A}$.
  \end{proof}
   
  \begin{proof}[Proof of Example~\ref{ex:HR_cond}]\label{prf:ex:HR_cond}
    The example follows from Lemma~\ref{lem:condHR} with $A=\{v\}$ and $C=\pa(v)$.
  \end{proof}
  
  \begin{proof}[Proof of Example~\ref{ex:HR_struct}]\label{prf:ex:HR_struct}
    From Example~\ref{ex:HR_cond} we have
    \begin{align*}
      \mu^*&=\begin{pmatrix}
        -\frac{1}{2}\Gamma_{v,\pa(v)}, 1\\
      \end{pmatrix}\begin{pmatrix}
        -\frac{1}{2}\Gamma_{\pa(v),\pa(v)} & \mathbf{1}\\
        \mathbf{1}&0\\
      \end{pmatrix}^{-1}\begin{pmatrix}
        \X_{\pa(v)}\\
        1\\
      \end{pmatrix}\\
      &=\begin{pmatrix}
        -\frac{1}{2}\Gamma_{v,\pa(v)}, 1\\
      \end{pmatrix}\begin{pmatrix}
        \theta(\pa(v))&\mathbf{p}(\pa(v))\\
        \mathbf{p}^{\top}(\pa(v))&\sigma^2(\pa(v))\\
      \end{pmatrix}\begin{pmatrix}
        \X_{\pa(v)}\\
        1\\
      \end{pmatrix}\\
      &=\begin{pmatrix}
        -\frac{1}{2}\Gamma_{v,\pa(v)}\theta(\pa(v))+\mathbf{p}^{\top}(\pa(v)),-\frac{1}{2}\Gamma_{v,\pa(v)}\mathbf{p}(\pa(v))+\sigma^2(\pa(v))\\
      \end{pmatrix}\begin{pmatrix}
        \X_{\pa(v)}\\
        1\\
      \end{pmatrix}.
    \end{align*}
  \end{proof}
   
  \begin{lemma}\label{lem:precision_CI}
    Let $\Y$ be a \HR{} Pareto distribution with variogram matrix $\Gamma$ and precision matrix $\Theta$. The following statements are equivalent to $ Y_i\perp_eY_j \mid Y_C $: 
    \begin{enumerate}
      \item $\theta(\{i,j\}\cup C)_{ij}=0$,
      \item $-\frac{1}{2}\Gamma_{i,j} -\begin{pmatrix}
        -\frac{1}{2}\Gamma_{i,C}& 1\\
      \end{pmatrix}
      \begin{pmatrix}
        -\frac{1}{2}\Gamma_{C,C} & \mathbf{1}\\
        \mathbf{1}^{\top}&0\\
      \end{pmatrix}^{-1}
      \begin{pmatrix}
        -\frac{1}{2}\Gamma_{C,j}\\
        1\\
      \end{pmatrix}=0$,
      \item $\rho_{ij \mid C} \coloneqq -\theta(\{i,j\}\cup C)_{ij} / \sqrt{\theta(\{i,j\}\cup C)_{ii} \theta(\{i,j\}\cup C)_{jj}}=0,$
      \item 	$\det \left(\Theta_{\setminus (\{i\}\cup C),\setminus (\{j\}\cup C)}\right)=0. $
    \end{enumerate}
  \end{lemma}
  \begin{proof}
  
    \begin{enumerate}
      \item\, The signed Laplacian matrix $\theta(\{i,j\}\cup C)$ is the \HR{} precision matrix corresponding to the variogram $\Gamma_{\{i,j\}\cup C,\{i,j\}\cup C}$, see \eqref{eq:Fiedler-Bapat_margin}. As this is the variogram of the $\{i,j\}\cup C$-th \HR{} marginal, we obtain the result from \cite{HES2022}.
      \item\, From Lemma~\ref{lem:condHR} it follows that $\lambda(y_i,y_j\mid \y_C)$ is the density of a bivariate Gaussian with conditional covariance matrix 
      \begin{align}
        \Sigma_*&=-\frac{1}{2}\Gamma_{\{i,j\},\{i,j\}} -\begin{pmatrix}
          -\frac{1}{2}\Gamma_{\{i,j\},C}& \einsfun\\
        \end{pmatrix}
        \begin{pmatrix}
          -\frac{1}{2}\Gamma_{C,C} & \mathbf{1}\\
          \mathbf{1}^{\top}&0\\
        \end{pmatrix}^{-1}
        \begin{pmatrix}
          -\frac{1}{2}\Gamma_{C,\{i,j\}}\\
          \einsfun^{\top}\\
        \end{pmatrix}.\label{eq:cond_cov}
      \end{align}
      The conditional density $\lambda(y_i,y_j|\y_C)$ therefore factorizes into $\lambda(y_i|\y_C)\lambda(y_j|\y_C)$ if and only the covariance vanishes, i.e.~$(\Sigma_*)_{ij}=0$.
      
      \item\, Let $\vartheta=\theta(\{i,j\}\cup C)$ be the Laplacian corresponding to $\Gamma_{\{i,j\}\cup C,\{i,j\}\cup C}$.
      From standard Schur complement arguments it follows that the conditional covariance $\Sigma_*$ from \eqref{eq:cond_cov} satisfies
      \begin{align*}
        \Sigma_*^{-1}&=\begin{pmatrix}
          \vartheta_{ii} &\vartheta_{ij}\\
          \vartheta_{ij} &\vartheta_{jj}\\
        \end{pmatrix},
      \end{align*}
      such that
      \begin{align*}
        \Sigma_*&=\frac{1}{\vartheta_{ii}\vartheta_{jj}-\vartheta_{ij}^2}\begin{pmatrix}
          \vartheta_{jj} &-\vartheta_{ij}\\
          -\vartheta_{ij} &\vartheta_{ii}\\
        \end{pmatrix}.
          \end{align*}
      This yields a partial correlation coefficient as the conditional correlation from $\Sigma_*$, that is
      \[\rho_{ij|C}=-\frac{\vartheta_{ij}}{\sqrt{\vartheta_{ii}\vartheta_{jj}}}.\]
      
      \item\, We observe that for $ k\in C $ with standard Schur complement arguments it is
    \begin{align*}
      \text{Cov}( \W^{k}_{\setminus C}\mid \W^k_{C\setminus k}) & = \left(\Theta_{\setminus C}^{(k)}\right)^{-1}.
    \end{align*}
    This means that
    \begin{align*}
      Y_i\perp_eY_j\mid \Y_C &\Leftrightarrow W^{k}_{i}\indep W_{j}^{k}\mid \W^{k}_{C\setminus k} \\
      & \Leftrightarrow \det \left(\Theta_{\setminus (\{i\}\cup C),\setminus (\{j\}\cup C)}^{(k)}\right)=0 \\
      &\Leftrightarrow \det \left(\Theta_{\setminus (\{i\}\cup C),\setminus (\{j\}\cup C)}\right)=0.
    \end{align*}
  \end{enumerate}
  \end{proof}
   
  \begin{lemma}\label{lem:partial_correlation}
  Let $\Y$ be \HR{} with variogram matrix $\Gamma$ and precision matrix $\Theta$.
  For every $k\in V$, let $\rho_{ij \mid S}^{(k)}$ be the partial correlation coefficient of $\W^k_{\setminus k}$ for any nonempty set $S\subseteq V$ with $i,j\not\in S$. 
  It holds that for any $k\in S$, the partial correlation coefficient $\rho_{ij \mid S}^{(k)}$ is equal to the extremal partial correlation coefficient $\rho_{ij \mid S}$.
  \end{lemma}
  \begin{proof}
     As $\W^k_{V\setminus\{k\}}$ is Gaussian, the conditional covariance  $\Cov(W^k_i,W^k_j|\W^k_S)=\Sigma^{(k)}_{ij\mid S} \coloneqq \Sigma^{(k)}_{ij} - \Sigma^{(k)}_{iS}(\Sigma^{(k)}_{SS})^{-1}\Sigma^{(k)}_{Sj}$ vanishes if and only if $W^k_i\ci W^k_j|\W_S^k$.
     Let $\theta(\{i,j\}\cup S)^{(k)}:=\left(\Sigma^{(k)}_{\{i,j\}\cup S,\{i,j\}\cup S}\right)^{-1}$, such that a standard Schur complement argument yields the partial correlation $$\rho_{ij \mid S}^{(k)} \coloneqq -\theta(\{i,j\}\cup S)^{(k)}_{ij} / \sqrt{\theta(\{i,j\}\cup S)^{(k)}_{ii} \theta(\{i,j\}\cup S)^{(k)}_{jj}}$$
     for any $k\in S$.
     Conveniently, $\theta(\{i,j\}\cup S)^{(k)}_{uv}=\theta(\{i,j\}\cup S)_{uv}$ for any $u,v \in\{i,j\}\cup S\setminus \{k\}$ and $k\in S$, where the matrix $\theta(A)$ can be obtained from $\Gamma$ for any $A\subseteq V$ as in \eqref{eq:Fiedler-Bapat_margin}.
     Thus, we find that 
     \begin{align}
       \rho_{ij \mid S} \coloneqq -\theta(\{i,j\}\cup S)_{ij} / \sqrt{\theta(\{i,j\}\cup S)_{ii} \theta(\{i,j\}\cup S)_{jj}}\label{eq:partial_corr}
     \end{align}
     is equal to the partial correlations $\rho_{ij \mid S}^{(k)}$ for all $k\in S$.
  \end{proof}
   
  \subsection{Example for a linear \HR{} SCM}\label{app:ex-HR-SEM}
  
  \begin{example}
  We assume the graph $G_e$ in Figure~\ref{subfig:b} with 
  \begin{align*}
    B&=\begin{pmatrix}
      0&1&1&0\\
      0&0&0&b_{24}\\
      0&0&0&b_{34}\\
      0&0&0&0\\
    \end{pmatrix},
    &&\mathfrak{L}(G_e)=\begin{pmatrix}
      -1&1&0&0\\
      -1&0&1&0\\
      0&-b_{24}&-b_{34}&1\\			
    \end{pmatrix},
  \end{align*}
  where $b_{24}+b_{34}=1$.		
  For $\nu_2^2,\nu_3^2,\nu_4^2>0$ this yields a H\"usler--Reiss precision matrix
  \begin{align*}
    \Theta=\mathfrak{L}(G_e)^{\top} \begin{pmatrix}
      1/\nu_2^2&0&0\\
      0&1/\nu_3^2&0\\
      0&0&1/\nu_4^2\\
    \end{pmatrix}\mathfrak{L}(G_e)
    = \begin{pmatrix}
      1/\nu_2^2+1/\nu_3^2&-1/\nu_2^2&-1/\nu_3^2&0\\
      -1/\nu_2^2&1/\nu_2^2+b_{24}^2/\nu_4^2&b_{24}b_{34}/\nu_4^2&-b_{24}/\nu_4^2\\
      -1/\nu_3^2&b_{24}b_{34}/\nu_4^2&1/\nu_3^2+b_{34}^2/\nu_4^2&-b_{34}/\nu_4^2\\
      0&-b_{24}/\nu_4^2&-b_{34}/\nu_4^2&1/\nu_4^2\\
    \end{pmatrix}.
  \end{align*}
  If we construct a \HR{} SCM $\Y$ from $\Theta$ as in Proposition~\ref{prop:HR_SCM_construction}, then $\Theta_{14}=0$ infers that $Y_1\perp_{e}Y_4\mid \Y_{23}$. It is further $Y_2\perp_{e}Y_3\mid Y_1$, as with Lemma~\ref{lem:precision_CI} we observe
  \[\det(\Theta_{\setminus\{1,2\},\setminus\{1,3\}})=\det\begin{pmatrix}
    b_{24}b_{34}/\nu_4^2&-b_{34}/\nu_4^2\\
    -b_{24}/\nu_4^2&1/\nu_4^2\\
  \end{pmatrix}=0.\]
  Thus, $\Y$ is indeed an extremal directed graphical model with respect to $G_e$.
  \end{example}
   
  \newpage
  \section{Graphical representation of the river network}\label{app:additional-exp-n-figures}
  
  \begin{figure}[!h]
  \centering
  \includegraphics[width=.8\textwidth]{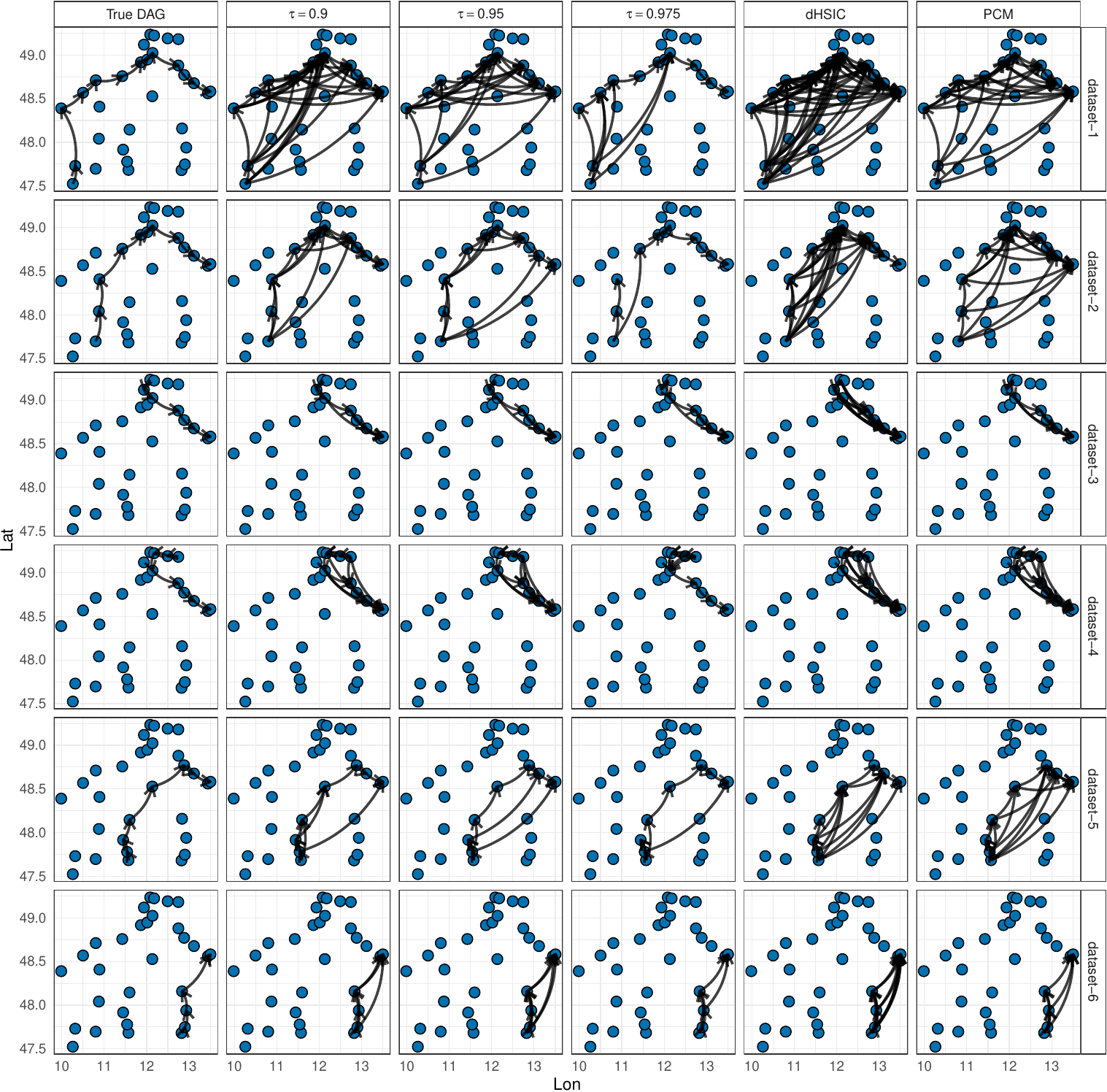}
  \caption{Recovered river network structures across different methods (columns) and river branches (rows) for a fixed subsample repetition of the experiment described in Section~\ref{sec:application}. 
  The X- and Y-axes represent the geographical coordinates over the river network, with nodes and edges representing stations and (inferred) connections, respectively.
  The first column presents the actual river network structure. The next three columns correspond to the graphs estimated by the extremal pruning algorithm using the extremal conditional independence test at threshold levels $\tau \in \{0.9, 0.95, 0.975\}$. The last two columns show the graphs learned using dHSIC and PCM instead of the extremal conditional independence test.
  }
  \label{fig:river-maps}
  \end{figure}

\end{appendix}
\end{document}